\documentclass{article}

\usepackage{silence}
\usepackage{sectsty} %\f?r att kunna anv?nda \allsectionsfont
	\allsectionsfont{\bfseries\sffamily} %g?r rubrik f?r chapter section osv till sans serif
\usepackage[latin1]{inputenc}
\usepackage[english]{babel}

\usepackage{csquotes} %för att kunna använda textquote

\usepackage{amssymb, amsmath, amsthm, mathrsfs, euscript}
\usepackage{mathtools} % extension of amsmath that also fixes two bugs
\usepackage{hyperref}
\usepackage{tensor}

\usepackage{verbatim} % för långa kommentarer comment
\usepackage{comment}

\usepackage{fancyhdr} %f?r sidhuvud
\usepackage{helvet} %helvetica
\usepackage[dvipsnames]{xcolor} %colors
\allowdisplaybreaks[3] %för radbrytningar, 1 betyder helst inte...

\usepackage{tikz} %bilder
	\usetikzlibrary{arrows} % needed for big arrow tips, "-implies"
	\usetikzlibrary{decorations.markings} % placing arrows halfway lines
	\usetikzlibrary{patterns} % hatch for wall
	\usetikzlibrary{calc} % for calculations

\usepackage{caption} % for size of captions
\usepackage[justification=centering]{subcaption} % allows subcaptionboxes in figures
\usepackage{enumerate}
\usepackage{enumitem} % for enumerating
\usepackage[margin=1.5in]{geometry}

\def\Z{\mathbb{Z}}                              %Heltal
\def\C{\mathbb{C}}                              %Komplexa tal
\def\Id{\textup{Id}}

\def\Re{\textup{Re}}
\def\Sp{\textup{Sp}}
\def\i{\mathrm{i}}

\theoremstyle{plain}
\newtheorem{theorem}{Theorem}[section]
\newtheorem{lemma}[theorem]{Lemma}
\newtheorem{cor}[theorem]{Corollary}
\newtheorem{prop}[theorem]{Proposition}

\theoremstyle{definition}

\theoremstyle{remark}

\numberwithin{equation}{section}
\numberwithin{figure}{section}

\makeatletter
\renewcommand*\env@matrix[1][\arraystretch]{%
  \edef\arraystretch{#1}%
  \hskip -\arraycolsep
  \let\@ifnextchar\new@ifnextchar
  \array{*\c@MaxMatrixCols c}}
\makeatother

\providecommand{\keywords}[1]{\noindent{\textbf{Keywords:}} #1}

\makeatletter
\g@addto@macro\bfseries{\boldmath}
\makeatother

\usepackage[final]{microtype}
\usepackage[style=numeric, sorting=nyt, giveninits=true, backend=biber]{biblatex}

\DeclareFieldFormat{pages}{#1} 

\DeclareFieldFormat*{title}{\mkbibitalic{#1}}
\DeclareFieldFormat*{journaltitle}{#1}
\DeclareFieldFormat[inproceedings]{booktitle}{#1}

\DeclareFieldFormat[article]{volume}{\mkbibbold{#1}}

\renewbibmacro{in:}{  
 \ifboolexpr{%
     test {\ifentrytype{article}}%
     or
     test {\ifentrytype{book}}%
  }{}{\printtext{\bibstring{in}\intitlepunct}}
}

\renewbibmacro*{publisher+location+date}{%
  \ifboolexpr{%
     test {\ifentrytype{inproceedings}}%
     or
     test {\ifentrytype{book}}%
  }{
		\printlist{publisher}%
		{\setunit*{\addcomma\space}}%
		\iflistundef{location}
			{\setunit*{\addcomma\space}}%
		\printlist{location}}%
	{\printlist{location}%
     \iflistundef{publisher}
      {\setunit*{\addcomma\space}}
      {\setunit*{\addcolon\space}}%
     \printlist{publisher}}
  \setunit*{\addcomma\space}%
  \usebibmacro{date}%
  \newunit}

\renewbibmacro*{volume+number+eid}{%
  \printfield{volume}%
  \setunit*{\addcomma\space}%
  \printfield{number}%
	\setunit*{\addnbspace}%
  \printfield{eid}%
	}
\usepackage[final]{microtype}
\usepackage[style=numeric, sorting=nyt, giveninits=true, backend=biber]{biblatex}

\DeclareFieldFormat{pages}{#1} 

\DeclareFieldFormat*{title}{\mkbibitalic{#1}}
\DeclareFieldFormat*{journaltitle}{#1}
\DeclareFieldFormat[inproceedings]{booktitle}{#1}

\DeclareFieldFormat[article]{volume}{\mkbibbold{#1}}

\renewbibmacro{in:}{  
 \ifboolexpr{%
     test {\ifentrytype{article}}%
     or
     test {\ifentrytype{book}}%
  }{}{\printtext{\bibstring{in}\intitlepunct}}
}

\renewbibmacro*{publisher+location+date}{%
  \ifboolexpr{%
     test {\ifentrytype{inproceedings}}%
     or
     test {\ifentrytype{book}}%
  }{
		\printlist{publisher}%
		{\setunit*{\addcomma\space}}%
		\iflistundef{location}
			{\setunit*{\addcomma\space}}%
		\printlist{location}}%
	{\printlist{location}%
     \iflistundef{publisher}
      {\setunit*{\addcomma\space}}
      {\setunit*{\addcolon\space}}%
     \printlist{publisher}}
  \setunit*{\addcomma\space}%
  \usebibmacro{date}%
  \newunit}

\renewbibmacro*{volume+number+eid}{%
  \printfield{volume}%
  \setunit*{\addcomma\space}%
  \printfield{number}%
	\setunit*{\addnbspace}%
  \printfield{eid}%
	}
\bibliography{./Referencestuff/References}

\title{Exact results for the six-vertex model with domain wall boundary conditions and a partially reflecting end}
\author{Linnea Hietala}
\date{}

\begin{document}
\maketitle
\begin{abstract}
The trigonometric six-vertex model with domain wall boundary conditions and one partially reflecting end on a lattice of size $2n\times m$, $m\leq n$, is considered. The partition function is computed using the Izergin--Korepin method, generalizing the result of Foda and Zarembo from the rational to the trigonometric case. Thereafter we specify the parameters in Kuperberg's way to get a formula for the number of states as a determinant of Wilson polynomials. We relate this to a type of ASM-like matrices.
\\
\keywords{six-vertex model, domain wall boundary conditions, partially reflecting end, partition function, triangular K-matrix}

\end{abstract}

%\Classification{} 
%\tableofcontents

\tikzset{midarrow/.style={
        decoration={markings,
            mark= at position 0.5 with {\arrow{#1}} ,
        },
        postaction={decorate}
    }
}
\tikzset{latearrow/.style={
        decoration={markings,
            mark= at position 0.7 with {\arrow{#1}} ,
        },
        postaction={decorate}
    }
}
\tikzset{earlyarrow/.style={
        decoration={markings,
            mark= at position 0.4 with {\arrow{#1}} ,
        },
        postaction={decorate}
    }
}

\section{Introduction}
The first example of a six-vertex (6V) model was the ice-model, where all states have the same weight. This and some other special cases of the 6V model with periodic boundary conditions were solved in 1967 by Lieb \cite{Lieb1967}. The same year, Sutherland \cite{Sutherland1967} solved the general case. %Lenard \cite{Lieb1967} (note added in proof) found a bijection from the states of the 6V model to three-colorings of a square lattice such that no adjacent squares have the same color and with the color in one corner fixed. Baxter \cite{Baxter1970} introduced the three-color model by assigning a weight to each color.

One of the first nontrivial examples of fixed boundaries were the domain wall boundary conditions (DWBC) \cite{Korepin1982}. 
In 1996, Zeilberger \cite{Zeilberger1996} proved the alternating sign matrix conjecture of Mills, Robbins and Rumsey \cite{MillsRobbinsRumsey1983}, which gives a formula for the number of alternating sign matrices (ASMs). There is a bijection between the ASMs and the states of the 6V model with DWBC. Izergin \cite{Izergin1987, IzerginCokerKorepin1992} showed that the partition function of the 6V model with DWBC can be expressed as a determinant, which Kuperberg \cite{Kuperberg1996} used to give another proof of the alternating sign matrix conjecture. %Kuperberg realized that in the special case where the parameters are cubic roots of unity, the partition function of the 6V model corresponds to counting the number of states. 

Tsuchiya \cite{Tsuchiya1998} used the Izergin--Korepin method to obtain a determinant formula for the partition function of the 6V model with one diagonal reflecting end and DWBC on the three other sides on a lattice of size $2n\times n$. Kuperberg \cite{Kuperberg2002} used this to give a formula for the number of the corresponding UASMs. The UASMs are alternating sign matrices with U-turns on one side, and generalize the vertically symmetric alternating sign matrices (VSASMs). 
%The eight-vertex (8V) model is a generalization of the 6V model. To solve the 8V model, Baxter \cite{Baxter1973} introduced the eight-vertex solid-on-solid (8VSOS) model, which is a two parameter generalization of the 6V model. %The name is a bit misleading, since it has only six different local states, and therefore the 8VSOS model is also called the elliptic SOS model. 
%Kuperberg's specialization of the parameters in the 6V model gives the ice model. Rosengren \cite{Rosengren2011} proved that the same specialization of the parameters in the 8VSOS model yields the three-color model. 

%Filali and Kitanine \cite{FilaliKitanine2010} found a single determinant formula for the partition function of the trigonometric SOS model with DWBC and a reflecting end. Filali \cite{Filali2011} found a single determinant formula for the elliptic case. For the elliptic 8VSOS model with DWBC, but without the reflecting end, no simple determinant formula has been found. In \cite{Hietala2020}, Kuperberg's specialization of the parameters in Filali's determinant formula for the partition function of the 8VSOS model with DWBC and one reflecting end was used to count the number of states of the three-color model with the corresponding boundary conditions. 

Foda and Wheeler \cite{FodaWheeler2012} found a determinant formula for the partition function of the 6V model with partial DWBC on a lattice of size $m\times n$, which generalizes the determinant formula of Korepin and Izergin. Foda and Zarembo \cite{FodaZarembo2016} found the corresponding generalization of Tsuchiya's determinant formula in the %(XXX)
rational case. They obtained a determinant formula for the rational
6V model on a lattice with $2n\times m$ sites, $m\leq n$, with DWBC and where the reflecting end has a triangular $K$-matrix. These boundary conditions are called DWBC with a partially reflecting end. Pozsgay \cite{Pozsgay2014} used the homogeneous limit of Tsuchiya's $2n\times n$ determinant to compute overlaps (i.e. inner products) between (off-shell) Bethe states and certain simple product states, such as the N\'eel states. In a similar way, Foda and Zarembo used their rational $2n\times m$ determinant formula to compute overlaps between Bethe states and more general objects which they call partial N\'eel states. 

Foda and Wheeler commented that in the case of partial DWBC on a lattice of size $n\times m$, it is not obvious if and how one could count ASM-like objects with Kuperberg's specialization, due to phases that vary between different states, coming from the trigonometric weights. However, in the present paper we find that in the case of DWBC and partial reflection, it is possible to count the states, since similar phases do not appear in this case. The reason for this is the alternating orientations of the lines. 

Counting ASMs can be generalized to $x$-enumerations. In the $x$-enumeration of ASMs, each state is counted with a weight $x^k$, where $k$ is the number of $-1$'s in the ASM. A formula in the general case is not known, but in some special cases, $x=1, 2$ and $3$, there are closed expressions \cite{Zeilberger1996, Kuperberg1996, MillsRobbinsRumsey1983}.
Colomo and Pronko \cite{ColomoPronko2006} obtained a simplified treatment of $x$-enumerations by rewriting the Hankel determinant representation of the partition function of the 6V model with DWBC in terms of orthogonal polynomials. The method can be used to find a formula for the $x$-enumerations for those $x$ where the underlying orthogonal polynomials belong to the Askey scheme of hypergeometric orthogonal polynomials. This works for the $1$-, $2$- and $3$-enumerations. In the original case of $1$-enumerations the orthogonal polynomials are continuous Hahn polynomials. 

In this paper, we study the trigonometric %(XXZ) 
6V model with DWBC on three sides and one partially reflecting end, on a lattice with $2n\times m$ sites, $m\leq n$. %i.e. it has a triangular $K$-matrix. 
We first find a determinant formula for the partition function. Then we specialize the parameters in Kuperberg's manner to finally find a formula that counts the number of states of the model. 

At first, in Section~\ref{sec:prel}, we introduce the model. In Section~\ref{sec:fodazarembo} we follow the Izergin--Korepin method to obtain a determinant formula for the partition function in Theorem~\ref{thm:determinantformula}, i.e. the trigonometric generalization of what Foda and Zarembo did in the rational case. %Given that one is familiar with the Izergin--Korepin method, it is straightforward and there should not be any big surprises in the present considered case.
%In Section~\ref{sec:fodawheeler}, we instead follow the steps of Foda and Wheeler to find the partition function. 
The alternative method of Foda and Wheeler to find the determinant formula is presented in Appendix A. 

In Section~\ref{sec:countingstates}, we find a formula for the number of states in terms of the partition function. We connect this to the enumeration of a type of generalized UASMs. The objective of Section~\ref{sec:rewritepartfcn} is to specialize the parameters in the determinant formula in Kuperberg's way. We rewrite the partition function following the ideas of Colomo and Pronko \cite{ColomoPronko2006}. The determinant can be represented by a matrix consisting of a Hankel matrix part and a Vandermonde matrix part, and the underlying orthogonal polynomials are Wilson polynomials. Then in Theorem~\ref{theorem:numberofstates}, we finally write down a determinant formula counting the number of states of the 6V model with DWBC and one partially reflecting end.

\section{Preliminaries}
\label{sec:prel}

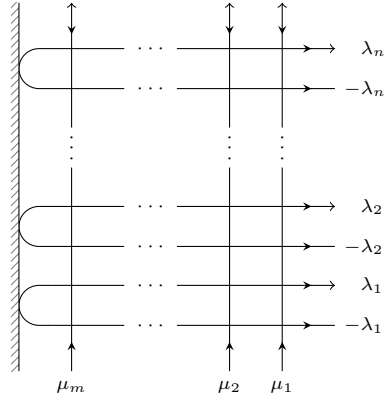
\begin{figure}[t]
\centering
\begin{tikzpicture}[baseline={([yshift=-.5*10pt*0.6]current bounding box.center)}, scale=0.7, font=\scriptsize]
	% first the horizontal border lines:
	\foreach \y in {1,2,4} {
		\draw[midarrow={stealth}] (5.55,1.5*\y-.25-.38) -- +(0.01,0);
		\draw[midarrow={stealth}] (5.55,1.5*\y-.25+.38) -- +(0.01,0);
		
		\draw (3, 1.5*\y-.25-.38) -- +(3, 0);
		\draw (0.38, 1.5*\y-.25-.38) -- +(1.62, 0);
		\draw[->] (3, 1.5*\y-.25+.38) -- +(3, 0);
		\draw[-] (0.38, 1.5*\y-.25+.38) -- +(1.62, 0);
		
		\foreach \x in {-1,...,1} \draw (2.5+.2*\x, 1.5*\y-.25-.38) node{$\cdot\mathstrut$};
		\foreach \x in {-1,...,1} \draw (2.5+.2*\x, 1.5*\y-.25+.38) node{$\cdot\mathstrut$};
		
		\draw (0.38,1.5*\y-.25+.38) arc (90:270:0.38);
	}
		
	\node[anchor=west] at (6, 1.5-.25-.38) {$-\lambda_1$};
	\node[anchor=west] at (6, 1.5-.25+.38) {$\phantom{-}\lambda_1$};
	\node[anchor=west] at (6, 1.5*2-.25-.38) {$-\lambda_2$};
	\node[anchor=west] at (6, 1.5*2-.25+.38) {$\phantom{-}\lambda_2$};
	\node[anchor=west] at (6, 1.5*4-.25-.38) {$-\lambda_n$};
	\node[anchor=west] at (6, 1.5*4-.25+.38) {$\phantom{-}\lambda_n$};
		
	% then the vertical border lines:
	\foreach \x in {1,4,5} {
		\draw (\x,0) -- +(0,4.5-0.25-.38); 
		\draw (\x,4.5-0.25+.38) -- +(0,1.63+0.25-0.38); 
			\draw[midarrow={stealth}] (\x,0.55) -- +(0,0.01);
		\draw[->] (\x,6.13)  -- +(0,.87);	
			\draw[midarrow={stealth reversed}] (\x,6.13+.43)  -- +(0,0.01);	
			
		\foreach \y in {-1,...,1} \draw (\x, 4.5-0.25+.2*\y) node{$\cdot\mathstrut$};
	}
	
	\node at (1,-0.3) {${\mu_m}$};
	\node at (4,-0.3) {${\mu_2}$};
	\node at (5,-0.3) {${\mu_1}$};
	
		%left border
		\draw (0,1.55-0.25) -- (0,1.56-0.25);
		\draw (.55,1.25-.38) -- +(0.01,0);
		\draw (.55,1.25+.38) -- +(0.01,0);
		
		% then the wall
	\fill[preaction={fill,white},pattern=north east lines, pattern color=gray] (0,0) rectangle (-.15,7) ; \draw (0,0) -- (0,7);
		
\end{tikzpicture}

\caption{The 6V model with DWBC and one partially reflecting end. The parameters $\lambda_i$ and $\mu_j$ are the spectral parameters.}
\label{fig:6vdwbcreflend}
\end{figure}

Consider a square lattice with $2n\times m$ lines, where the horizontal lines are connected pairwise at the left side, as in Figure~\ref{fig:6vdwbcreflend}. Each such pair of horizontal lines can be thought of as one single line turning at a wall on the left side, see Figure~\ref{fig:6vdwbcreflend}. 
We assign a spin $\pm 1$ to each edge. 
A lattice with a spin assigned to each edge is called a state. 

In order to assign weights to the states, we give each line an orientation. We choose a positive direction, which goes upwards for the vertical lines, to the left for the lower part of the horizontal double line, and to the right for the upper part. The positive direction is indicated by an arrow at the end of a line.
%Sometimes we will think of the positive direction of a line as a flow through the line going in the positive direction, and refer to ingoing and outgoing edges, e.g. on the boundaries. 
Graphically, spin $+1$ corresponds to an arrow pointing in the positive direction of the line, and spin $-1$ corresponds to an arrow pointing in the opposite direction. At each vertex, the so called ice rule must hold, which demands that two arrows must be pointing inwards to the vertex and two arrows must be pointing outwards.
Because of the ice rule, there are only six types of possible vertices, see Figure~\ref{fig:vertexweights}

%\begin{figure}[t]
%\centering
	%\begin{tikzpicture}[scale=0.7]
		%\draw[->] (0,1) node[left]{\scriptsize $\alpha$} -- (2,1) node[right]{\scriptsize $\alpha^\prime$};
		%\draw[->] (1,0) node[below]{\scriptsize $\beta$} -- (1,2) node[above]{\scriptsize $\beta^\prime$};
		%%\draw (1,-1.5) node{$w\w{\alpha}{\beta}{\gamma}{\delta}$};
	%\end{tikzpicture}
%\caption{A vertex with spins $\alpha, \beta, \alpha^\prime, \beta^\prime=\pm 1$ on the surrounding edges.}
%\label{fig:vertex}
%\end{figure}

%At each vertex with spins $\alpha, \beta, \alpha^\prime$ and $\beta^\prime$ on the four surrounding edges as in Figure~\ref{fig:vertex}, the equation 
%\begin{displaymath}
%\alpha+\beta=\alpha^\prime+\beta^\prime
%\end{displaymath}
%must hold. This rule is called the ice rule, and yields six possible types of vertices, see Figure~\ref{fig:vertexweights}.

To each vertical line, we assign a spectral parameter $\mu_j$, and to each horizontal double line, we assign a spectral parameter which is $-\lambda_i$ on the lower part of the double line and shifts to $\lambda_i$ on the upper part. In Figure~\ref{fig:6vdwbcreflend}, we write these parameters at the lines. Also define a fixed boundary parameter $\zeta\in \C$, associated to the reflecting wall at the turns.

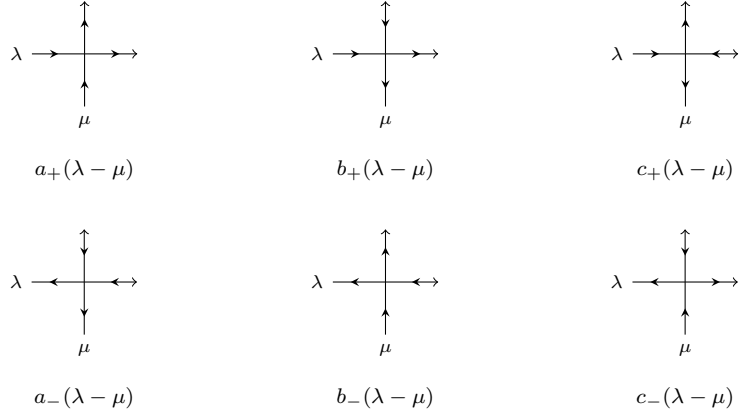
\begin{figure}[t]
\centering
    \subcaptionbox*{}{%	
    	\begin{tikzpicture}[scale=0.7, font=\footnotesize]
  	\draw[midarrow={stealth}] (0,1) node[left] {\scriptsize $\lambda$} -- (1,1); %left
		\draw[midarrow={stealth reversed}, <-] (2,1) -- (1,1); %right
		\draw[midarrow={stealth}] (1,0) node[below] {\scriptsize $\mu$} -- (1,1); %down
		\draw[midarrow={stealth reversed}, <-] (1,2) -- (1,1); %up
		\draw (1,-1.2) node{$%w\w{+}{+}{+}{+} = 
		a_+(\lambda-\mu)$};
		\end{tikzpicture}
	}\hfil
	\subcaptionbox*{}{%	
		\begin{tikzpicture}[scale=0.7, font=\footnotesize]
		\draw[midarrow={stealth}] (0,1) node[left] {\scriptsize $\lambda$} -- (1,1); %left
		\draw[midarrow={stealth reversed}, <-] (2,1) -- (1,1); %right
		\draw[midarrow={stealth reversed}] (1,0) node[below] {\scriptsize $\mu$} -- (1,1); %down
		\draw[midarrow={stealth}, <-] (1,2) -- (1,1); %up
		\draw (1,-1.2) node{$%w\w{+}{-}{+}{-} = 
		b_+(\lambda-\mu)$};
		\end{tikzpicture}
	}\hfil
	\subcaptionbox*{}{%	
		\begin{tikzpicture}[scale=0.7, font=\footnotesize]
		\draw[midarrow={stealth}] (0,1) node[left] {\scriptsize $\lambda$} -- (1,1); %left
		\draw[midarrow={stealth}, <-] (2,1) -- (1,1); %right
		\draw[midarrow={stealth reversed}] (1,0) node[below] {\scriptsize $\mu$} -- (1,1); %down
		\draw[midarrow={stealth reversed}, <-] (1,2) -- (1,1); %up
		\draw (1,-1.2) node{$%w\w{+}{-}{-}{+} = 
		c_+(\lambda-\mu)$};
		\end{tikzpicture}
	}\\
	\vspace{0mm}
    \subcaptionbox*{}{%	
		\begin{tikzpicture}[scale=0.7, font=\footnotesize]
		\draw[midarrow={stealth reversed}] (0,1) node[left] {\scriptsize $\lambda$} -- (1,1); %left
		\draw[midarrow={stealth}, <-] (2,1) -- (1,1); %right
		\draw[midarrow={stealth reversed}] (1,0) node[below] {\scriptsize $\mu$} -- (1,1); %down
		\draw[midarrow={stealth}, <-] (1,2) -- (1,1); %up
		\draw (1,-1.2) node{$%w\w{-}{-}{-}{-} = 
		a_-(\lambda-\mu)$};
		\end{tikzpicture}
	}\hfil
	\subcaptionbox*{}{%	
		\begin{tikzpicture}[scale=0.7, font=\footnotesize]
		\draw[midarrow={stealth reversed}] (0,1) node[left] {\scriptsize $\lambda$} -- (1,1); %left
		\draw[midarrow={stealth}, <-] (2,1) -- (1,1); %right
		\draw[midarrow={stealth}] (1,0) node[below] {\scriptsize $\mu$} -- (1,1); %down
		\draw[midarrow={stealth reversed}, <-] (1,2) -- (1,1); %up
		\draw (1,-1.2) node{$%w\w{-}{+}{-}{+} = 
		b_-(\lambda-\mu)$};
		\end{tikzpicture}
	}\hfil
	\subcaptionbox*{}{%	
		\begin{tikzpicture}[scale=0.7, font=\footnotesize]
		\draw[midarrow={stealth reversed}] (0,1) node[left] {\scriptsize $\lambda$} -- (1,1); %left
		\draw[midarrow={stealth reversed}, <-] (2,1) -- (1,1); %right
		\draw[midarrow={stealth}] (1,0) node[below] {\scriptsize $\mu$} -- (1,1); %down
		\draw[midarrow={stealth}, <-] (1,2) -- (1,1); %up
		\draw (1,-1.2) node{$%w\w{-}{+}{+}{-} = 
		c_-(\lambda-\mu)$};
		\end{tikzpicture}
	}\\
		\vspace{-6mm}
	\caption{The possible vertices and their vertex weights for the 6V model. The spins are indicated by an arrow halfway the edge, where right and up are positive spins, and left and down are negative spins. The vertex weights also depend on the spectral parameters $\lambda$ and $\mu$.}
	\label{fig:vertexweights} 
\end{figure}

\begin{figure}
\centering
\subcaptionbox*{}{%			
    \begin{tikzpicture}[baseline={([yshift=-10pt*0.6]current bounding box.center)}, scale=0.7, font=\footnotesize]
		\draw (0.5,0.5) arc (90:270:0.5);
		\draw[midarrow={stealth reversed}] (.5,-.5) -- (1.5,-.5) node[right] {\scriptsize $-\lambda$};
		%\draw (1.5, -.5) -- (2.5, -.5) node[right] {\scriptsize $-\lambda$};  
		\draw[midarrow={stealth}, ->] (.5,+.5) -- (1.5,+.5) node[right] {\scriptsize $\lambda$};
		%\draw[->] (1.5, +.5) -- (2.5, +.5);
		%\draw[->] (1.5,-1.5) node[below]{\scriptsize $\mu$} -- (1.5,1.5);
		%\draw[-stealth] (0,0.05) -- (0,0.06);
	
			% then the wall
	%\fill[preaction={fill,white},pattern=north east lines, pattern color=gray] (0,-0.5) rectangle (-.15,0.5) ; \draw (0,-0.5) -- (0,0.5);
	
	\draw (1.2,-1.4) node{$%w(+)=
	k_+(\lambda, \zeta)$};
    \end{tikzpicture}
}\hfil
\subcaptionbox*{}{%			
    \begin{tikzpicture}[baseline={([yshift=-10pt*0.6]current bounding box.center)}, scale=0.7, font=\footnotesize]
		\draw (0.5,0.5) arc (90:270:0.5);
		\draw[midarrow={stealth}] (.5,-.5) -- (1.5,-.5) node[right] {\scriptsize $-\lambda$};
		%\draw (1.5, -.5) -- (2.5, -.5) node[right] {\scriptsize $-\lambda$}; 
		\draw[midarrow={stealth reversed},->] (.5,+.5) -- (1.5,+.5) node[right] {\scriptsize $\lambda$};
		%\draw[->] (1.5, +.5) -- (2.5, +.5);
		%\draw[->] (1.5,-1.5) node[below]{\scriptsize $\mu$} -- (1.5,1.5);
		%\draw[-stealth] (0,-0.05) -- (0,-0.06);
		
			% then the wall
	%\fill[preaction={fill,white},pattern=north east lines, pattern color=gray] (0,-0.5) rectangle (-.15,0.5) ; \draw (0,-0.5) -- (0,0.5);
	
	\draw (1.2,-1.4) node{$%w(-)=
	k_-(\lambda, \zeta)$};
    \end{tikzpicture}
		}\hfil
\subcaptionbox*{}{%			
    \begin{tikzpicture}[baseline={([yshift=-10pt*0.6]current bounding box.center)}, scale=0.7, font=\footnotesize]
		\draw (0.5,0.5) arc (90:270:0.5);
		\draw[midarrow={stealth}] (.5,-.5) -- (1.5,-.5) node[right] {\scriptsize $-\lambda$};
		%\draw (1.5, -.5) -- (2.5, -.5) node[right] {\scriptsize $-\lambda$}; 
		\draw[midarrow={stealth},->] (.5,+.5) -- (1.5,+.5) node[right] {\scriptsize $\lambda$};
		%\draw[->] (1.5, +.5) -- (2.5, +.5);
		%\draw[->] (1.5,-1.5) node[below]{\scriptsize $\mu$} -- (1.5,1.5);
	 %\node at (0,0) {\textbullet};
	
		% then the wall
	%\fill[preaction={fill,white},pattern=north east lines, pattern color=gray] (0,-0.5) rectangle (-.15,0.5) ; \draw (0,-0.5) -- (0,0.5);
	
	\draw (1.2,-1.4) node{$%w(-)=
	k_c(\lambda, \zeta)$};
    \end{tikzpicture}
}

	\vspace{-6mm}
\caption{The possible boundary configurations and boundary weights for the triangular reflecting end. The weights depend on the spectral parameter $\lambda$ as well as on a boundary parameter $\zeta$.}
    \label{fig:reflectingends}
\end{figure}
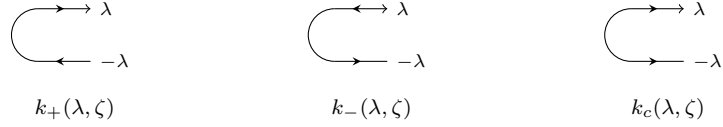

Define
$f(x)=2\sinh(x)$ and let $\gamma\notin 2\pi i\Z$ be a fixed parameter. Then define local weights 
\begin{equation}
\label{vertexweights}
a_\pm(\lambda)=1,\qquad b_\pm(\lambda)=e^{\mp \gamma}\frac{f(\lambda)}{f(\lambda+\gamma)}, \qquad c_\pm(\lambda)=e^{\pm \lambda}\frac{f(\gamma)}{f(\lambda+\gamma)},
\end{equation}
\begin{equation}
\label{turnweights}
k_\pm(\lambda, \zeta)=e^{\zeta\mp\lambda} f(\zeta\pm\lambda),\qquad k_c(\lambda, \zeta)=\varphi f(2\lambda),
\end{equation}
to each vertex and each turn as in Figure~\ref{fig:vertexweights} and Figure~\ref{fig:reflectingends}. Here $\varphi$ is a fixed number. We call 
the turns $k_\pm$ a 'positive' and 'negative' turn respectively, and $k_c$ can be seen as a turn with creation of arrows. 
%We choose the weights in this way, because they behave well in the limit where $\mu\to\infty$, which we will need in Section~\ref{sec:fodawheeler}. 
These choices of weights satisfy the Yang--Baxter equation and the reflection equation with a triangular $K$-matrix where $k_c(\lambda, \zeta)$ is an off-diagonal element (see Section~\ref{subsec:YBErefeq}). Observe that $k_c$ does not depend on $\zeta$, and for $\varphi=0$, we have diagonal reflection. Sometimes %when the spin configurations around a vertex is of interest, %or when $\lambda$ is given, 
we will refer to a '$w$~vertex', where $w$ is one of $a_\pm, b_\pm$ or $c_\pm$, meaning a vertex with spin configurations corresponding to weight $w(\lambda)$, for some $\lambda$. Similarly a '$k_\pm$~turn' or '$k_c$~turn' will refer to a turn with weight $k_\pm(\lambda, \zeta)$ or $k_c(\lambda, \zeta)$ respectively. % when $\lambda$ and $\zeta$ are given, or when it is the directions of the spins on the turning edge that are of importance. 

The local weight at a vertex with the positive directions up and to the right depends on the spins of the surrounding edges, as well as on the difference between the spectral parameters on the incoming lines from the left and the bottom. Because of the reflecting ends, we need to differentiate between the vertices on the left oriented and the right oriented horizontal lines. The vertices in the right oriented rows are depicted in Figure~\ref{fig:vertexweights}, and the vertices in the left oriented rows are the same, tilted 90 degrees counterclockwise, as in Figure~\ref{fig:lowerrow}. The (local) weight of the vertex in Figure~\ref{fig:upperrow} is $w(\lambda_i-\mu_j)$, and for the vertex in Figure~\ref{fig:lowerrow}, the weight is $w(\mu_j-(-\lambda_i))=w(\lambda_i+\mu_j)$, where $w$ is one of $a_\pm, b_\pm$ or $c_\pm$. 
The (local) boundary weight at each turn depends on the spin on the turning edge, but also on the spectral parameter $\lambda_i$ of the line going through the turn, and on the fixed boundary parameter $\zeta$, as in Figure~\ref{fig:reflectingends}. The weight of a state is the product of all local weights of the vertices and the turns. 

On the three sides without reflecting end, we impose the domain wall boundary conditions, with outgoing spin arrows to the right and ingoing arrows on the upper and lower boundaries. The ice rule implies that $n\geq m$ and that there are $n-m$ turns of type $k_c$. 

The model described above is the six-vertex (6V) model of size $2n\times m$ with DWBC and one partially reflecting end.
We want to find a determinant formula for the partition function 
\begin{equation}
\label{naivepartition}
Z_{n,m}(\boldsymbol\lambda, \boldsymbol\mu)=\sum_\text{state} \text{weight}(\text{state})
\end{equation} of this model, generalizing Tsuchiyas \cite{Tsuchiya1998} partition function for $m=n$. This also generalizes the results of Foda and Zarembo \cite{FodaZarembo2016} from the rational to the trigonometric case.

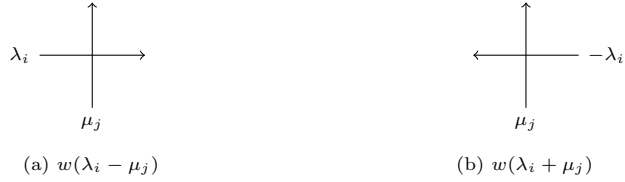
\begin{figure}[t]
\centering
	\subcaptionbox{$w(\lambda_i-\mu_j)$\label{fig:upperrow}}{%	
	\begin{tikzpicture}[baseline={(0,0.5)}, scale=0.7]
		\draw[->] (0,1) node[left]{\scriptsize$\lambda_i$} -- (2,1);
		\draw[->] (1,0) node[below]{\scriptsize$\mu_j$} -- (1,2);
		\node at (-1,0) {\phantom{$\bullet$}};
		\node at (3,0) {\phantom{$\bullet$}};
	\end{tikzpicture}
	}\hfil
	\subcaptionbox{$w(\lambda_i+\mu_j)$\label{fig:lowerrow}}{%	
	\begin{tikzpicture}[baseline={(0,0.5)}, scale=0.7]
		\draw[<-] (0,1) -- (2,1)  node[right]{\scriptsize$-\lambda_i$};
		\draw[->] (1,0) node[below]{\scriptsize$\mu_j$} -- (1,2);
		\node at (-1,0) {\phantom{$\bullet$}};
		\node at (3,0) {\phantom{$\bullet$}};
	\end{tikzpicture}
	}\\
	\vspace{0mm}
	\caption{The different vertex weights depending on the direction of the row in the 6V model with reflecting end, with spectral parameters $\lambda_i$ and $\mu_j$.}
	\label{fig:nodeupdown}
	\end{figure}

\subsection{The Yang--Baxter equation and the reflection equation}
\label{subsec:YBErefeq}
Define $V$ as a two-dimensional complex vector space. To each line of the lattice we associate a copy of $V$. Given a parameter $\lambda\in\C$, define operators $R(\lambda)\in \text{End}(V\otimes V)$ by
$$
R(\lambda)=
\begin{pmatrix}
    a_+(\lambda) & 0 & 0 & 0 \\
    0 & b_+(\lambda) & c_-(\lambda) & 0 \\
    0 & c_+(\lambda) & b_-(\lambda) & 0 \\
    0 & 0 & 0 & a_-(\lambda)
\end{pmatrix},
$$
with the weights parametrized as in \eqref{vertexweights}.
%$$R(\lambda)(e_{\alpha}\otimes e_{\beta})=\sum_{\alpha+\beta=\alpha^\prime+\beta^\prime} w\w{\alpha}{\beta}{\alpha^\prime}{\beta^\prime}(\lambda)e_{\alpha^\prime}\otimes e_{\beta^\prime},$$
%where $w\w{\alpha}{\beta}{\alpha^\prime}{\beta^\prime}(\lambda)$ is the weight $a_\pm(\lambda)$, $b_\pm(\lambda)$ or $c_\pm(\lambda)$ corresponding to a vertex with spins $\alpha, \beta, \alpha^\prime, \beta^\prime$ on the surrounding edges as in Figure~\ref{fig:vertex}.
%
The operator is called the $R$-matrix and satisfies the Yang--Baxter equation (YBE) on $V_1\otimes V_2\otimes V_3$ (where $V_i$ are copies of $V$), i.e.
\begin{multline*}
R_{12}(\lambda_1-\lambda_2)R_{13}(\lambda_1-\lambda_3)R_{23}(\lambda_2-\lambda_3)=R_{23}(\lambda_2-\lambda_3)R_{13}(\lambda_1-\lambda_3)R_{12}(\lambda_1-\lambda_2),
\end{multline*}
where the indices indicate on which spaces the $R$-matrix acts, e.g. 
 $$R_{12}(\lambda_1-\lambda_2)=R(\lambda_1-\lambda_2)\otimes \Id,$$ 
and similarly for $R_{23}$ and $R_{13}$. 
%$$R_{23}(\lambda_2-\lambda_3)=\Id\otimes R(\lambda_2-\lambda_3).$$ To define $R_{13}$ we first need to define a permutation operator $P\in \text{End}(V\otimes V)$, by $P (e_\alpha\otimes e_\beta)=e_\beta\otimes e_\alpha$. Then 
%\begin{equation}
%\label{R13}
%R_{13}=(\Id\otimes P)(R\otimes \Id)(\Id\otimes P).
%\end{equation}
The YBE is depicted in Figure~\ref{fig:YBE}.

\begin{figure}[h]
\centering
\begin{tikzpicture}[baseline={([yshift=-.5*11pt*0.6 +8.5pt]current bounding box.center)}, scale=0.6, font=\scriptsize]
    \pgfmathsetmacro{\csc}{1/sin(130)}
        \draw[->] (-130:1.5*\csc) node[below]{$\lambda_1$} -- (50:1.5*\csc);
        \draw[->] (.5,-1.5) node[below]{$\lambda_2$} -- (.5,1.5);
        \draw[->] (-50:1.5*\csc) node[below]{$\lambda_3$} -- (130:1.5*\csc);
        %    
        %\node at (-1,0) {$\rho$};
    \end{tikzpicture}
    \ \ = \ \
    \begin{tikzpicture}[baseline={([yshift=-.5*11pt*0.6 +8.5pt]current bounding box.center)}, scale=0.6, font=\scriptsize]
    \pgfmathsetmacro{\csc}{1/sin(130)}
        \draw[->] (-130:1.5*\csc) node[below]{$\lambda_1$} -- (50:1.5*\csc);
        \draw[->] (-.5,-1.5) node[below]{$\lambda_2$} -- (-.5,1.5);
        \draw[->] (-50:1.5*\csc) node[below]{$\lambda_3$} -- (130:1.5*\csc);
        %    
        %\node at (-1.25,0) {$\rho$};
    \end{tikzpicture}
\caption{The Yang--Baxter equation.}
\label{fig:YBE}		
\end{figure}
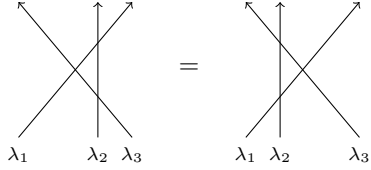

To describe the reflecting boundary, define a triangular operator $K(\lambda, \zeta)\in \text{End}(V)$, by 
$$K(\lambda, \zeta)=
\begin{pmatrix}
    k_+(\lambda, \zeta) & k_c(\lambda, \zeta)  \\
    0 & k_-(\lambda, \zeta)  \\
\end{pmatrix},
$$
with entries parametrized as in \eqref{turnweights}. The operator is called the $K$-matrix and satisfies the reflection equation for the above $R$-matrix on $V_0\otimes V_{0'}$ \cite{Sklyanin1988}, i.e. 
\begin{multline*}
R_{00'}(\lambda-\lambda')K_0(\lambda, \zeta)R_{0'0}(\lambda+\lambda')K_{0'}(\lambda', \zeta)\\
=K_{0'}(\lambda', \zeta)R_{00'}(\lambda+\lambda')K_0(\lambda, \zeta)R_{0'0}(\lambda-\lambda'),
\end{multline*}
where $K_0(\lambda, \zeta)=K(\lambda, \zeta)\otimes \Id$ and $K_{0'}(\lambda, \zeta)=\Id\otimes K(\lambda, \zeta)$, see Figure~\ref{fig:reflectioneq}. 

\begin{figure}[h]
\centering
    \begin{tikzpicture}[baseline={([yshift=-.5*10pt*0.6+3.5pt]current bounding box.center)},% 7 pt = scriptsize using text size 10 pt
    xscale=0.6, yscale=0.425, font=\scriptsize, triple/.style={postaction={draw,-,shorten >=.05},double,double distance=4pt,-implies}]{
    \pgfmathsetmacro\A{cos(60)}
    \pgfmathsetmacro\B{cos(20)}
        \coordinate(o1) at (0,-.75);
        \coordinate(e1) at ($(60:3.5)+(o1)$);
        \coordinate(o2) at (0,.75);
        \coordinate(e21) at ($(20:2)+(o2)$);
        \coordinate(e22) at ($(-20:2)+(o2)$);
        \coordinate(b1) at (intersection of o1--e1 and o2--e21);
        \coordinate(b2) at (intersection of o1--e1 and o2--e22);
        
        \draw ($(-60:1/\A)+(o1)$) -- (o1) -- (b1);
        \draw[->] (b1) -- (e1);
        \node at ($(-60:1.2/\A)+(o1)$) {${-\lambda'}$};
        
        \draw (e22) -- (b2);
        \draw (b2) -- (o2) -- (b1);
        \draw[->] (b1) -- (e21);
        \node at  ($(-20:2.5)+(o2)$) {${-\lambda}$};
        
        \fill[preaction={fill,white},pattern=north east lines, pattern color=gray] (0,-2.5) rectangle (-.15,2.25); \draw (0,-2.5) -- (0,2.25);

        %\draw ($(-7.5:1/\B)+(o2)$) +(-7.5:-4pt/\B) -- +(-7.5:4pt/\B);
        %\draw ($(7.5:1/\B)+(o2)$) +(7.5:-4pt/\B) -- +(7.5:4pt/\B);
        %\draw ($(-35:1/\A)+(o1)$) +(-35:-4pt/\A) -- +(-35:4pt/\A);
        %\draw ($(35:1/\A)+(o1)$) +(35:-4pt/\A) -- +(35:4pt/\A);
    
        %\node at (.4,1.9) {$\rho$};
				}
    \end{tikzpicture}
     \ \ = \ \
    \begin{tikzpicture}[baseline={([yshift=-.5*10pt*0.6+4pt]current bounding box.center)},% 7 pt = scriptsize using text size 10 pt
    xscale=0.6, yscale=0.425, font=\scriptsize, triple/.style={postaction={draw,-,shorten >=.05},double,double distance=4pt,-implies}]{
    \pgfmathsetmacro\A{cos(60)}
    \pgfmathsetmacro\B{cos(20)}
    \coordinate(o1) at (0,.5);
    \coordinate(e1) at ($(-60:3.5)+(o1)$);
    \coordinate(o2) at (0,-1);
    \coordinate(e21) at ($(-20:2)+(o2)$);
    \coordinate(e22) at ($(20:2)+(o2)$);
    \coordinate(b1) at (intersection of o1--e1 and o2--e21);
    \coordinate(b2) at (intersection of o1--e1 and o2--e22);
		
		%\node at (o1) {o1};
		%\node at (e1) {e1};
		%\node at (o2) {o2};
		%\node at (e21) {e21};
		%\node at (e22) {e22};
		%\node at (b1) {b1};
		%\node at (b2) {b2};

        \draw (e1) -- (b1);
        \draw[->] (b1) -- (o1) -- ($(60:1/\A)+(o1)$);
        \node at ($(-60:3.9)+(o1)$) {${-\lambda'}$};
        
        \draw (e21) -- (b1);
        \draw (b1) -- (o2) -- (b2);
        \draw[->] (b2) -- (e22);
        \node at ($(-20:2.5)+(o2)$) {${-\lambda}$};
        
        \fill[preaction={fill,white},pattern=north east lines, pattern color=gray] (0,-2.5) rectangle (-.15,2.25); \draw (0,-2.5) -- (0,2.25);
        
        %\draw ($(-7.5:1/\B)+(o2)$) +(-7.5:-4pt/\B) -- +(-7.5:4pt/\B);
        %\draw ($(7.5:1/\B)+(o2)$) +(7.5:-4pt/\B) -- +(7.5:4pt/\B);
        %\draw ($(-35:1/\A)+(o1)$) +(-35:-4pt/\A) -- +(-35:4pt/\A);
        %\draw ($(35:1/\A)+(o1)$) +(35:-4pt/\A) -- +(35:4pt/\A);

        %\node at (.4,1.9) {$\rho$};
				}
    \end{tikzpicture}
		\vspace{-1mm}
\caption{The reflection equation. %for the matrix $K(\lambda, \zeta)$. %Remark: the different angles have no significance.
}
\label{fig:reflectioneq}
\end{figure}
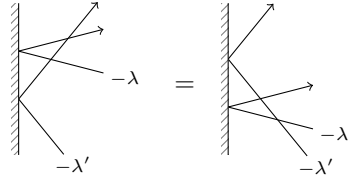

\section{Izergin--Korepin method}
\label{sec:fodazarembo}
Foda and Zarembo \cite{FodaZarembo2016} used the Izergin--Korepin method \cite{Korepin1982, Izergin1987} to find the partition function \eqref{naivepartition} in the rational case, i.e. where $f(x)=x$, which in turn means that all weights are rational. We follow the Izergin--Korepin procedure to find the determinant formula for the partition function in the trigonometric case. We will use the short hand notation $$f(x\pm y)\coloneqq f(x+y)f(x-y).$$ 

\begin{lemma}
\label{lemma:symmetryinlambda}
$Z_{n,m}(\boldsymbol\lambda, \boldsymbol\mu)$ is symmetric in $\lambda_i$ and $\mu_j$ separately. 
\end{lemma}

To prove this we will use the so called train argument. 

\begin{proof}
Consider two adjacent vertical lines with spectral parameters $\mu$ and $\mu'$. Insert an extra vertex below the lattice, as in Figure~\ref{fig:rttdwbc}. Since we have DWBC, this will be a vertex with weight $a_+(\mu-\mu')$. By the YBE, the extra vertex can be moved through the whole lattice to end up on the top, where it can be removed. %The alternating orientations on the horizontal lines do not affect the proof. 
On the top, the extra vertex has the weight $a_-(\mu-\mu')$. Since $a_+(\mu)=a_-(\mu)\neq 0$ for all $\mu$, these factors cancel. Hence this procedure switches $\mu$ and $\mu'$. 

To prove the symmetry in the $\lambda_i$'s, we add two extra vertices on the right, as in Figure~\ref{fig:globalreflectioneq}. The vertices can then be pulled through each other in a similar way as before, using the YBE and the reflection equation. Because of the boundary conditions and the ice rule, the extra vertices give rise to two extra factors on each side of the equation, see Figure~\ref{fig:globalreflectioneq}. Again since $a_+(\lambda)=a_-(\lambda)\neq 0$, the extra factors cancel. In this manner we can pairwise switch spectral parameters on adjacent lines, which proves the lemma. 
\end{proof}

\begin{figure}
\centering
	$a_+(\mu'-\mu)\ \times$
	\begin{tikzpicture}[baseline={([yshift=-.5*10pt*0.6+7pt]current bounding box.center)},% 7 pt = scriptsize using text size 10 pt
	scale=0.6,font=\small,]{
		\draw[<-] (0,1) -- (3,1);
		\draw[->] (0,2) -- (3,2);
		\draw[<-] (0,4) -- (3,4);
		\draw[->] (0,5) -- (3,5);
		\draw[midarrow={stealth}] (1,0) node[below]{$\mu\vphantom{'}$} -- (1,1.2); 
		\draw (1.2,1)  -- (2.6,1); 
		\draw[midarrow={stealth}] (2,0) node[below]{$\mu'$} -- (2,1.2); 
		\draw (1.2,2)  -- (2.6,2); 
		
		\draw(1,3.5) -- (1,5.1);
		\draw[midarrow={stealth reversed},->] (1,5.1) -- (1,6);
		\draw(2,3.5) -- (2,5.1);
		\draw[midarrow={stealth reversed},->] (2,5.1) -- (2,6);
		\draw (2,0.9) -- (2,2.6);
		\draw (1,0.9) -- (1,2.6);
		
		\foreach \x in {-1,...,1} \draw (1,3+.2*\x) node{$\cdot\mathstrut$}; 
		\foreach \x in {-1,...,1} \draw (2,3+.2*\x) node{$\cdot\mathstrut$}; % Use a forall-loop from values -1 to 1 for the variable that I call \x to draw the three dots for sites 3 to L-1 at x-coordinates 2.8, 3, 3.2 (thus spaced slightly tighter than they would be in \cdots). The mathstrut is to get the correct vertical centering of the dots. (I don't know how to use such a loop in the above while getting subscript 'L' at x=4)
	}
	\end{tikzpicture}
	\ \ = \ 
	\begin{tikzpicture}[baseline={([yshift=-.5*10pt*0.6+11pt]current bounding box.center)}, 
	scale=0.6, font=\small,]{
		\draw[<-] (0,1) -- (3,1);
		\draw[->] (0,2) -- (3,2);
		\draw[<-] (0,4) -- (3,4);
		\draw[->] (0,5) -- (3,5);
		\draw(1,3.4) -- (1,5.1);
		\draw[midarrow={stealth reversed},->] (1,5.1) -- (1,6);
		\draw(2,3.4) -- (2,5.1);
		\draw[midarrow={stealth reversed},->] (2,5.1) -- (2,6);
		
		\draw[latearrow={stealth}] (1,-0.5) node[below]{$\mu'$} -- (1.5,0);
		\draw[earlyarrow={stealth}, rounded corners=5pt] (1.5,0) -- (2,0.5) -- (2,0.9);
		\draw (2,0.9) -- (2,2.6);
		\draw[latearrow={stealth}] (2,-0.5) node[below]{$\mu\vphantom{'}$} -- (1.5,0);
		\draw[earlyarrow={stealth}, rounded corners=5pt] (1.5,0) -- (1,0.5) -- (1,0.9);
		\draw (1,0.9) -- (1,2.6);
		
		\foreach \x in {-1,...,1} \draw (1,3+.2*\x) node{$\cdot\mathstrut$}; 
		\foreach \x in {-1,...,1} \draw (2,3+.2*\x) node{$\cdot\mathstrut$};
	}
	\end{tikzpicture}
	\ \, 
	\\
	= \ 
	\begin{tikzpicture}[baseline={([yshift=-.5*10pt*0.6+3pt]current bounding box.center)},
	scale=0.6,font=\small,]{
		\draw[<-] (0,1) -- (3,1);
		\draw[->] (0,2) -- (3,2);
		\draw[<-] (0,4) -- (3,4);
		\draw[->] (0,5) -- (3,5);
		\draw[midarrow={stealth}] (1,0) node[below]{$\mu'$} -- (1,1.2); 
		\draw (1,1.2)  -- (1,2.6); 
		\draw[midarrow={stealth}] (2,0) node[below]{$\mu\vphantom{'}$} -- (2,1.2); 
		\draw (2,1.2)  -- (2,2.6); 
		
		\draw[rounded corners=5pt] (1,3.4)  -- (1,5); 
		\draw[latearrow={stealth reversed}, rounded corners=5pt] (1,5)  -- (1,5.5) -- (1.6,6.1); 
		\draw[midarrow={stealth reversed},->] (1.6,6.1) -- (2,6.5); 
		\draw[rounded corners=5pt] (2,3.4)  -- (2,5); 
		\draw[latearrow={stealth reversed}, rounded corners=5pt] (2,5)  -- (2,5.5) -- (1.4,6.1); 
		\draw[midarrow={stealth reversed},->] (1.4,6.1) -- (1,6.5); 
		
		\foreach \x in {-1,...,1} \draw (1,3+.2*\x) node{$\cdot\mathstrut$}; 
		\foreach \x in {-1,...,1} \draw (2,3+.2*\x) node{$\cdot\mathstrut$};
	}
	\end{tikzpicture}
	\ \ = \ 
	\begin{tikzpicture}[baseline={([yshift=-.5*10pt*0.6+7pt]current bounding box.center)}, % 7 pt = scriptsize using text size 10 pt
	scale=0.6, font=\small,]{
	  \draw[<-] (0,1) -- (3,1);
		\draw[->] (0,2) -- (3,2);
		\draw[<-] (0,4) -- (3,4);
		\draw[->] (0,5) -- (3,5);
		\draw[midarrow={stealth}] (1,0) node[below]{$\mu'$} -- (1,1.2); 
		\draw (1,1.2)  -- (1,2.6); 
		\draw[midarrow={stealth}] (2,0) node[below]{$\mu\vphantom{'}$} -- (2,1.2); 
		\draw (2,1.2)  -- (2,2.6); 
		
		\draw(1,3.4) -- (1,5.1);
		\draw[midarrow={stealth reversed},->] (1,5.1) -- (1,6);
		\draw(2,3.4) -- (2,5.1);
		\draw[midarrow={stealth reversed},->] (2,5.1) -- (2,6);
		
		\foreach \x in {-1,...,1} \draw (1,3+.2*\x) node{$\cdot\mathstrut$}; 
		\foreach \x in {-1,...,1} \draw (2,3+.2*\x) node{$\cdot\mathstrut$};
	}
	\end{tikzpicture}
	\ \, 
	$\times \ a_-(\mu'-\mu)$
\caption{The partition function is symmetric in the $\mu_j$'s. Because of the DWBC, an extra vertex can be moved through the lattice by using the YBE.}
\label{fig:rttdwbc}	
\end{figure}
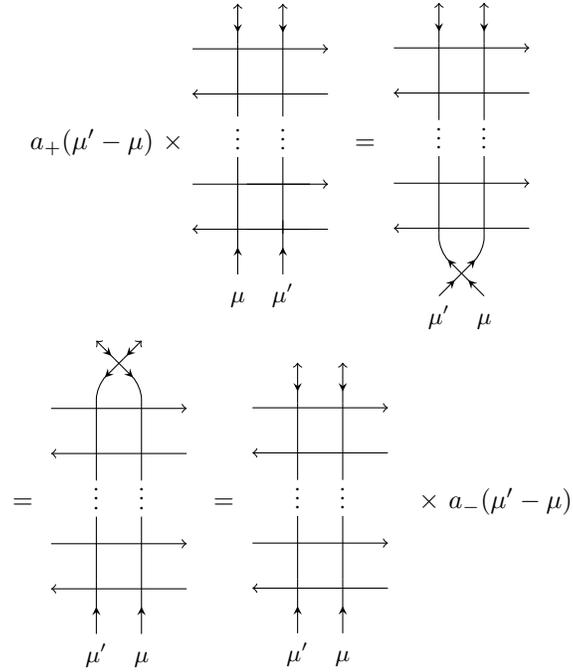

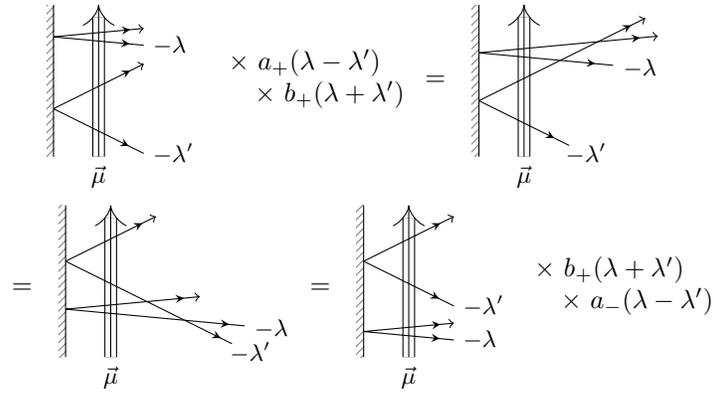
\begin{figure}
\centering
	\begin{tikzpicture}[baseline={([yshift=-.5*10pt*0.6+7.5pt]current bounding box.center)},% 7 pt = scriptsize using text size 10 pt
    xscale=0.6, yscale=0.425, font=\small, triple/.style={postaction={draw,-,shorten >=.05},double,double distance=4pt,-implies}]{
    \pgfmathsetmacro\A{cos(35)}
    \pgfmathsetmacro\B{cos(7.5)}
        \begin{scope}[shift={(0,-1)}]
            \draw[midarrow={stealth reversed}] (-35:2/\A) node [right] {${-\lambda'}$} -- (-35:1/\A);
            \draw (-35:1/\A)-- (0,0) -- (35:1/\A);
            \draw[latearrow={stealth},->] (35:1/\A) -- (35:2/\A);
        \end{scope}
        \begin{scope}[shift={(0,1.25)}]
            \draw[midarrow={stealth reversed}] (-7.5:2/\B) node [right] {${-\lambda}$} -- (-7.5:1/\B);
            \draw (-7.5:1/\B)-- (0,0) -- (7.5:1/\B);
            \draw[latearrow={stealth},->] (7.5:1/\B) -- (7.5:2/\B);
        \end{scope}
        \fill[preaction={fill,white},pattern=north east lines, pattern color=gray] (0,-2.5) rectangle (-.15,2.25); \draw (0,-2.5) -- (0,2.25);
        \draw[triple] (1,-2.5) node[below]{\small $\vec{\mu}$} -- (1,2.25);
        
        \draw ($(-7.5:1/\B)+(0,1.25)$) +(-7.5:-4pt/\B) -- +(-7.5:4pt/\B);
        \draw ($(7.5:1/\B)+(0,1.25)$) +(7.5:-4pt/\B) -- +(7.5:4pt/\B);
        \draw ($(-35:1/\A)+(0,-1)$) +(-35:-4pt/\A) -- +(-35:4pt/\A);
        \draw ($(35:1/\A)+(0,-1)$) +(35:-4pt/\A) -- +(35:4pt/\A);
				}
    \end{tikzpicture}
		  \! $\begin{array}{l} \times \ a_+(\lambda-\lambda') \\ \quad \times \ b_+(\lambda+\lambda') \end{array} =$ \ \
    \begin{tikzpicture}[baseline={([yshift=-.5*10pt*0.6+7.5pt]current bounding box.center)},% 7 pt = scriptsize using text size 10 pt
    xscale=0.6, yscale=0.425, font=\small, triple/.style={postaction={draw,-,shorten >=.05},double,double distance=4pt,-implies}]{
    \pgfmathsetmacro\A{cos(35)}
    \pgfmathsetmacro\B{cos(7.5)}
        \coordinate(o1) at (0,-.75);
        \coordinate(e1) at ($(35:4.5)+(o1)$);
        \coordinate(o2) at (0,.75);
        \coordinate(e21) at ($(7.5:4)+(o2)$);
        \coordinate(e22) at ($(-7.5:3)+(o2)$);
        \coordinate(b1) at (intersection of o1--e1 and o2--e21);
        \coordinate(b2) at (intersection of o1--e1 and o2--e22);
        
        \draw[midarrow={stealth reversed}] ($(-35:2/\A)+(o1)$) -- ($(-35:1/\A)+(o1)$);
        \draw ($(-35:1/\A)+(o1)$) -- (o1) -- (b1);
        \draw[latearrow={stealth},->] (b1) -- (e1);
        \node at ($(-35:2.375/\A)+(o1)$) {${-\lambda'}$};
        
        \draw[midarrow={stealth reversed}] (e22) -- (b2);
        \draw (b2) -- (o2) -- (b1);
        \draw[latearrow={stealth},->] (b1) -- (e21);
        \node at  ($(-7.5:3.6)+(o2)$) {${-\lambda}$};
        
        \fill[preaction={fill,white},pattern=north east lines, pattern color=gray] (0,-2.5) rectangle (-.15,2.25); \draw (0,-2.5) -- (0,2.25);
        \draw[triple] (1,-2.5) node[below]{$\vec{\mu}$} -- (1,2.25);

        \draw ($(-7.5:1/\B)+(o2)$) +(-7.5:-4pt/\B) -- +(-7.5:4pt/\B);
        \draw ($(7.5:1/\B)+(o2)$) +(7.5:-4pt/\B) -- +(7.5:4pt/\B);
        \draw ($(-35:1/\A)+(o1)$) +(-35:-4pt/\A) -- +(-35:4pt/\A);
        \draw ($(35:1/\A)+(o1)$) +(35:-4pt/\A) -- +(35:4pt/\A);
				}
    \end{tikzpicture}
		  \\
     \quad = \ \
    \begin{tikzpicture}[baseline={([yshift=-.5*10pt*0.6+5pt]current bounding box.center)},% 7 pt = scriptsize using text size 10 pt
    xscale=0.6, yscale=0.425, font=\small, triple/.style={postaction={draw,-,shorten >=.05},double,double distance=4pt,-implies}]{
    \pgfmathsetmacro\A{cos(35)}
    \pgfmathsetmacro\B{cos(7.5)}
    \coordinate(o1) at (0,.5);
    \coordinate(e1) at ($(-35:4.5)+(o1)$);
    \coordinate(o2) at (0,-1);
    \coordinate(e21) at ($(-7.5:4)+(o2)$);
    \coordinate(e22) at ($(7.5:3)+(o2)$);
    \coordinate(b1) at (intersection of o1--e1 and o2--e21);
    \coordinate(b2) at (intersection of o1--e1 and o2--e22);
        
        \draw[earlyarrow={stealth reversed}] (e1) -- (b1);
        \draw (b1) -- (o1) -- ($(35:1/\A)+(o1)$);
        \draw[latearrow={stealth},->] ($(35:1/\A)+(o1)$) -- ($(35:2/\A)+(o1)$);
        \node at ($(-35:5)+(o1)$) {${-\lambda'}$};
        
        \draw[earlyarrow={stealth reversed}] (e21) -- (b1);
        \draw (b1) -- (o2) -- (b2);
        \draw[latearrow={stealth},->] (b2) -- (e22);
        \node at ($(-7.6:4.6)+(o2)$) {${-\lambda}$};
        
        \fill[preaction={fill,white},pattern=north east lines, pattern color=gray] (0,-2.5) rectangle (-.15,2.25); \draw (0,-2.5) -- (0,2.25);
        \draw[triple] (1,-2.5) node[below]{$\vec{\mu}$} -- (1,2.25);
        
        \draw ($(-7.5:1/\B)+(o2)$) +(-7.5:-4pt/\B) -- +(-7.5:4pt/\B);
        \draw ($(7.5:1/\B)+(o2)$) +(7.5:-4pt/\B) -- +(7.5:4pt/\B);
        \draw ($(-35:1/\A)+(o1)$) +(-35:-4pt/\A) -- +(-35:4pt/\A);
        \draw ($(35:1/\A)+(o1)$) +(35:-4pt/\A) -- +(35:4pt/\A);
				}
    \end{tikzpicture}
    \! $=$ \ \
    \begin{tikzpicture}[baseline={([yshift=-.5*10pt*0.6+5pt]current bounding box.center)},% 7 pt = scriptsize using text size 10 pt
    xscale=0.6, yscale=0.425, font=\small, triple/.style={postaction={draw,-,shorten >=.05},double,double distance=4pt,-implies}]{
    \pgfmathsetmacro\A{cos(35)}
    \pgfmathsetmacro\B{cos(7.5)}
        \begin{scope}[shift={(0,.5)}]
            \draw[midarrow={stealth reversed}] (-35:2/\A) node [right] {${-\lambda'}$} -- (-35:1/\A);
            \draw (-35:1/\A)-- (0,0) -- (35:1/\A);
            \draw[latearrow={stealth},->] (35:1/\A) -- (35:2/\A);
        \end{scope}
        \begin{scope}[shift={(0,-1.7)}]
            \draw[midarrow={stealth reversed}] (-7.5:2/\B) node [right] {${-\lambda}$} -- (-7.5:1/\B);
            \draw (-7.5:1/\B)-- (0,0) -- (7.5:1/\B);
            \draw[latearrow={stealth},->] (7.5:1/\B) -- (7.5:2/\B);
        \end{scope}
        \fill[preaction={fill,white},pattern=north east lines, pattern color=gray] (0,-2.5) rectangle (-.15,2.25); \draw (0,-2.5) -- (0,2.25);
        \draw[triple] (1,-2.5) node[below]{$\vec{\mu}$} -- (1,2.25);
        
        \draw ($(-7.5:1/\B)+(0,-1.7)$) +(-7.5:-4pt/\B) -- +(-7.5:4pt/\B);
        \draw ($(7.5:1/\B)+(0,-1.7)$) +(7.5:-4pt/\B) -- +(7.5:4pt/\B);
        \draw ($(-35:1/\A)+(0,.5)$) +(-35:-4pt/\A) -- +(-35:4pt/\A);
        \draw ($(35:1/\A)+(0,.5)$) +(35:-4pt/\A) -- +(35:4pt/\A);

				}
    \end{tikzpicture}
    $\begin{array}{l} \times \ b_+(\lambda+\lambda') \\ \quad \times \ a_-(\lambda-\lambda') \end{array}$
\caption{The partition function is symmetric in the $\lambda_i$'s. Two adjacent double rows can switch places by adding extra vertices that can be moved through the lattice by using the YBE and the reflection equation. The triple arrows should be understood as $n$ vertical arrows.}
\label{fig:globalreflectioneq}
\end{figure}

\begin{lemma}
\label{lemma:2n1}
For any $1\leq j\leq m$, the function $$e^{(2n-2)\mu_j}\prod_{i=1}^n\prod_{j=1}^m f(\lambda_i\pm\mu_j+\gamma) Z_{n,m}(\boldsymbol\lambda, \boldsymbol\mu)$$ is a polynomial of degree $2n-1$ in $e^{2\mu_j}$. 
\end{lemma}

\begin{proof}
Multiply each vertex weight \eqref{vertexweights} by $f(\lambda+\gamma)$. This is the same as to multiply the partition function by
$$\prod_{i=1}^n\prod_{j=1}^m f(\lambda_i\pm\mu_j+\gamma).$$
Then the new weights are
$$\hat a_\pm(\lambda)=f(\lambda+\gamma),\quad \hat b_\pm(\lambda)=e^{\mp \gamma}f(\lambda), \quad \hat c_\pm(\lambda)=e^{\pm\lambda} f(\gamma).$$ For a given $j$, all vertices involving $\mu_j$ lie along the same vertical line. Along a vertical line the spin must be changed at least once from up spin on the bottom of the lattice to down spin on the top. Each vertical line thus contains at least one $\hat c_+(\lambda+\mu)$ or one $\hat c_-(\lambda-\mu)$ vertex since these are the only combinations changing the spin in this way, see Figure~\ref{fig:vertexweights}. These weights each contain a factor $e^{\mu_j}$. Each additional $\hat c_\pm$ vertex yields a factor $e^{\pm \mu_j}$. Each $\hat a_\pm$ and $\hat b_\pm$ vertex gives rise to a factor $e^{\mu_j+x}-e^{-\mu_j-x}$, where $x$ possibly contains $\lambda_i$ and $\gamma$. There can be a maximum of $2n-1$ vertices of type $\hat a_\pm$ and $\hat b_\pm$ involving $\mu_j$. We multiply $\prod_{i=1}^n\prod_{j=1}^m f(\lambda_i\pm\mu_j+\gamma)Z_{n,m}(\boldsymbol\lambda, \boldsymbol\mu)$ by $e^{(2n-2)\mu_j}$ to get a polynomial of degree $2n-1$ in $e^{2\mu_j}$.
%Fel tror jag $2n-1$ in $e^{2\mu_j}$. 
\end{proof}

\begin{lemma}
\label{lemma:inductionstep}
$Z_{n,m}(\boldsymbol{\lambda}, \boldsymbol{\mu})$ satisfies a recurrence relation when setting $\mu_k=\pm \lambda_l$, namely,
\begin{align*}
&Z_{n,m}(\boldsymbol{\lambda}, \boldsymbol{\mu})|_{\mu_k=\pm\lambda_l}= e^{-n\gamma}e^{\zeta+\mu_k} f(\zeta-\mu_k)
\prod_{i=1}^n \frac{f(\lambda_i+\lambda_l)}{f(\lambda_i+\lambda_l+\gamma)} \times Z_{n-1, m-1}(\boldsymbol{\hat\lambda_l}, \boldsymbol{\hat\mu_k}),
\end{align*}
where $$\boldsymbol{\hat\lambda_l}=(\lambda_1, \dots, \hat\lambda_l, \dots, \lambda_n), \qquad \boldsymbol{\hat\mu_k}=(\mu_1, \dots, \hat\mu_k, \dots, \mu_m),$$
and $\hat\lambda_l$ indicates that the variable $\lambda_l$ is omitted, and similarly $\hat\mu_k$ indicates that the variable $\mu_k$ is omitted.
\end{lemma}

\begin{proof}
The vertices marked with a red dot in Figure~\ref{lambdaneqmu1} can only be $b_+$ or $c_\pm$ vertices, because of the ice rule and the DWBC. Specialize $\mu_1=-\lambda_1$ as in the left lattice. Now if the vertex at the bottom right is of type $b_+$, it has zero weight, and does not contribute to the partition function. Hence the vertex must be a $c_+$ vertex. Then the ice rule determines the spins on the two horizontal bottom lines and the rightmost vertical line, i.e. the vertices within the frozen region (the area marked with a blue dotted line) are determined uniquely. The part of the lattice outside the frozen region is then a lattice with DWBC of size $2(n-1)\times (m-1)$. This yields a recursion relation for the partition function. Likewise, specializing $\mu_1=\lambda_n$, as in the right lattice, forces the vertex marked with a red dot to be a $c_-$ vertex, and the ice rule determines the rest of the vertex weights within the frozen region. Lemma~\ref{lemma:symmetryinlambda} yields that the recursion relations hold for $\mu_k=\pm\lambda_l$ for any $1\leq l\leq n$, $1\leq k\leq m$. 
\end{proof}

\begin{figure}[tb]
\centering
\subcaptionbox*{}{%
\begin{tikzpicture}[baseline={([yshift=-.5*10pt*0.6]current bounding box.center)}, scale=0.8, font=\scriptsize]
	% first the horizontal border lines:
	\foreach \y in {1,2,4} {
		\draw[midarrow={stealth}] (5.55,1.5*\y-.25-.38) -- +(0.01,0);
		\draw[midarrow={stealth}] (5.55,1.5*\y-.25+.38) -- +(0.01,0);
		
		\draw (3, 1.5*\y-.25-.38) -- +(3, 0);
		\draw (0.38, 1.5*\y-.25-.38) -- +(1.62, 0);
		\draw[->] (3, 1.5*\y-.25+.38) -- +(3, 0);
		\draw[-] (0.38, 1.5*\y-.25+.38) -- +(1.62, 0);
		
		\foreach \x in {-1,...,1} \draw (2.5+.2*\x, 1.5*\y-.25-.38) node{$\cdot\mathstrut$};
		\foreach \x in {-1,...,1} \draw (2.5+.2*\x, 1.5*\y-.25+.38) node{$\cdot\mathstrut$};
		
		\draw (0.38,1.5*\y-.25+.38) arc (90:270:0.38);
	}
		
	\node[anchor=west] at (6, 1.5-.25-.38) {$-\lambda_1$};
	\node[anchor=west] at (6, 1.5*2-.25-.38) {$-\lambda_2$};
	\node[anchor=west] at (6, 1.5*4-.25-.38) {$-\lambda_n$};
		
	% then the vertical border lines:
	\foreach \x in {1,4,5} {
		\draw (\x,0) -- +(0,4.5-0.25-.38); 
		\draw (\x,4.5-0.25+.38) -- +(0,1.63+0.25-0.38); 
			\draw[midarrow={stealth}] (\x,0.55) -- +(0,0.01);
		\draw[->] (\x,6.13)  -- +(0,.87);	
			\draw[midarrow={stealth reversed}] (\x,6.13+.43)  -- +(0,0.01);	
			
		\foreach \y in {-1,...,1} \draw (\x, 4.5-0.25+.2*\y) node{$\cdot\mathstrut$};
	}
	
	\node at (1,-0.3) {${\mu_m}$};
	\node at (4,-0.3) {${\mu_2}$};
	\node at (5,-0.3) {${-\lambda_{1}}$};
	
		%left border
		%\draw[-stealth] (0,1.55-0.25) -- (0,1.56-0.25);
		\draw[midarrow={stealth reversed}] (.55,1.25-.38) -- +(0.01,0);
		\draw[midarrow={stealth}] (.55,1.25+.38) -- +(0.01,0);
		
		%bulk horizontal lines
		\draw [midarrow={stealth reversed}] (1.55, 1.5-.25-.38) -- +(0.01, 0);
		\draw [midarrow={stealth reversed}] (3.55, 1.5-.25-.38) -- +(0.01, 0);
		\draw [midarrow={stealth reversed}] (4.55, 1.5-.25-.38) -- +(0.01, 0);
		\draw [midarrow={stealth}](1.55, 1.5-.25+.38) -- +(0.01, 0);
		\draw [midarrow={stealth}] (3.55, 1.5-.25+.38) -- +(0.01, 0);
		\draw [midarrow={stealth}] (4.55, 1.5-.25+.38) -- +(0.01, 0);
		
		\draw [midarrow={stealth}] (4.55, 3-.25-.38) -- +(0.01, 0);
		\draw [midarrow={stealth}] (4.55, 3-.25+.38) -- +(0.01, 0);

		%\draw [midarrow={stealth}] (1.55, 6-.25-.38) -- +(0.01, 0);
		%
		%\draw [midarrow={stealth}] (3.55, 6-.25-.38) -- +(0.01, 0);
		\draw [midarrow={stealth}] (4.55, 6-.25-.38) -- +(0.01, 0);
		%\draw [midarrow={stealth}] (1.55, 6-.25+.38) -- +(0.01, 0);
		%
		%\draw [midarrow={stealth}](3.55, 6-.25+.38) -- +(0.01, 0);
		\draw [midarrow={stealth}](4.55, 6-.25+.38) -- +(0.01, 0);
		
	%bulk vertical lines
	
			\draw [midarrow={stealth}] (1,0.87+.43) -- +(0,0.01); 
			\draw [midarrow={stealth}] (1,1.63+.43) -- +(0,0.01); 
			\draw [midarrow={stealth}] (4,0.87+.43) -- +(0,0.01); 
			\draw [midarrow={stealth}] (4,1.63+.43) -- +(0,0.01); 
			
			\draw [midarrow={stealth reversed}] (5,0.87+.43) -- +(0,0.01); 
			\draw [midarrow={stealth reversed}] (5,1.63+.43) -- +(0,0.01); 
			\draw [midarrow={stealth reversed}] (5,2.37+.43) -- +(0,0.01); 
			\draw [midarrow={stealth reversed}] (5,3.13+.43) -- +(0,0.01); 
			%\draw [midarrow={stealth reversed}] (5,3.87+.43) -- +(0,0.01);
			\draw [midarrow={stealth reversed}] (5,4.63+.43) -- +(0,0.01); 
			\draw [midarrow={stealth reversed}] (5,5.37+.43) -- +(0,0.01); 
			
			%freezing line
			\draw [densely dotted, blue] (0.02, 1.63+0.18) -- (4.75, 1.63+0.18) -- (4.75, 7-0.25) -- (5.75, 7-0.25) -- (5.75, 0.18) -- (0.02, 0.18) -- (0.02, 1.63+0.18);
			
			\node [red] at (5,1.5-0.25-0.38) {\textbullet};
			
			% then the wall
	\fill[preaction={fill,white},pattern=north east lines, pattern color=gray] (0,0) rectangle (-.15,7) ; \draw (0,0) -- (0,7);
\end{tikzpicture}}
\hfil
\subcaptionbox*{}{%
\begin{tikzpicture}[baseline={([yshift=-.5*10pt*0.6]current bounding box.center)}, scale=0.8, font=\scriptsize]
	% first the horizontal border lines:
	\foreach \y in {1,3,4} {
		\draw[midarrow={stealth}] (5.55,1.5*\y-.25-.38) -- +(0.01,0);
		\draw[midarrow={stealth}] (5.55,1.5*\y-.25+.38) -- +(0.01,0);
		
		\draw (3, 1.5*\y-.25-.38) -- +(3, 0);
		\draw (0.38, 1.5*\y-.25-.38) -- +(1.62, 0);
		\draw[->] (3, 1.5*\y-.25+.38) -- +(3, 0);
		\draw[-] (0.38, 1.5*\y-.25+.38) -- +(1.62, 0);
		
		\foreach \x in {-1,...,1} \draw (2.5+.2*\x, 1.5*\y-.25-.38) node{$\cdot\mathstrut$};
		\foreach \x in {-1,...,1} \draw (2.5+.2*\x, 1.5*\y-.25+.38) node{$\cdot\mathstrut$};
		
		\draw (0.38,1.5*\y-.25+.38) arc (90:270:0.38);
	}
		
	\node[anchor=west] at (6, 1.5-.25-.38) {$-\lambda_1$};
	\node[anchor=west] at (6, 1.5*3-.25-.38) {$-\lambda_{n-1}$};
	\node[anchor=west] at (6, 1.5*4-.25-.38) {$-\lambda_n$};
		
	% then the vertical border lines:
	\foreach \x in {1,4,5} {
		\draw (\x,0) -- +(0,3-0.25-.38); 
		\draw (\x,3-0.25+.38) -- +(0,3.13+0.25-0.38); 
			\draw[midarrow={stealth}] (\x,0.55) -- +(0,0.01);
		\draw[->] (\x,6.13)  -- +(0,.87);	
			\draw[midarrow={stealth reversed}] (\x,6.13+.43)  -- +(0,0.01);	
			
		\foreach \y in {-1,...,1} \draw (\x, 3-0.25+.2*\y) node{$\cdot\mathstrut$};
	}
	
	\node at (1,-0.3) {${\mu_m}$};
	\node at (4,-0.3) {${\mu_2}$};
	\node at (5,-0.3) {${\lambda_{n}}$};
	
		%left border

		%\draw[-stealth reversed] (0,6.05-0.25) -- (0,6.06-0.25);
		\draw[midarrow={stealth}] (.55,5.75-.38) -- +(0.01,0);
		\draw[midarrow={stealth reversed}] (.55,5.75+.38) -- +(0.01,0);
		
		%bulk horizontal lines
		%\draw [midarrow={stealth reversed}] (1.55, 1.5-.25-.38) -- +(0.01, 0);
		%
		%\draw [midarrow={stealth reversed}] (3.55, 1.5-.25-.38) -- +(0.01, 0);
		\draw [midarrow={stealth}] (4.55, 1.5-.25-.38) -- +(0.01, 0);
		%\draw [midarrow={stealth}](1.55, 1.5-.25+.38) -- +(0.01, 0);
		%
		%\draw [midarrow={stealth}] (3.55, 1.5-.25+.38) -- +(0.01, 0);
		\draw [midarrow={stealth}] (4.55, 1.5-.25+.38) -- +(0.01, 0);
		
		\draw [midarrow={stealth}] (4.55, 4.5-.25-.38) -- +(0.01, 0);
		\draw [midarrow={stealth}] (4.55, 4.5-.25+.38) -- +(0.01, 0);

		\draw [midarrow={stealth}] (1.55, 6-.25-.38) -- +(0.01, 0);
		\draw [midarrow={stealth}] (3.55, 6-.25-.38) -- +(0.01, 0);
		\draw [midarrow={stealth}] (4.55, 6-.25-.38) -- +(0.01, 0);
		\draw [midarrow={stealth reversed}] (1.55, 6-.25+.38) -- +(0.01, 0);
		\draw [midarrow={stealth reversed}](3.55, 6-.25+.38) -- +(0.01, 0);
		\draw [midarrow={stealth reversed}](4.55, 6-.25+.38) -- +(0.01, 0);
		
	%bulk vertical lines
	
			\draw [midarrow={stealth reversed}] (1,4.63+.43) -- +(0,0.01); 
			\draw [midarrow={stealth reversed}] (1,5.37+.43) -- +(0,0.01); 
			\draw [midarrow={stealth reversed}] (4,4.63+.43) -- +(0,0.01); 
			\draw [midarrow={stealth reversed}] (4,5.37+.43) -- +(0,0.01); 
			
			\draw [midarrow={stealth}] (5,0.87+.43) -- +(0,0.01); 
			\draw [midarrow={stealth}] (5,1.63+.43) -- +(0,0.01); 
			%\draw [midarrow={stealth}] (5,2.37+.43) -- +(0,0.01); 
			\draw [midarrow={stealth}] (5,3.13+.43) -- +(0,0.01); 
			\draw [midarrow={stealth}] (5,3.87+.43) -- +(0,0.01);
			\draw [midarrow={stealth}] (5,4.63+.43) -- +(0,0.01); 
			\draw [midarrow={stealth}] (5,5.37+.43) -- +(0,0.01); 
			
			%freezing line
			\draw [densely dotted, blue] (4.75, 0.18) -- (4.75, 5.37-0.18) -- (0.02, 5.37-0.18) -- (0.02, 7-0.25) -- (5.75, 7-0.25) -- (5.75, 0.18) -- (4.75, 0.18);

			\node [red] at (5,6-0.25+0.38) {\textbullet};
			
			% then the wall
	\fill[preaction={fill,white},pattern=north east lines, pattern color=gray] (0,0) rectangle (-.15,7) ; \draw (0,0) -- (0,7);
\end{tikzpicture}}
\vspace{-6mm}
\caption{Specializing $\mu_j=\pm\lambda_i$ determines the vertices inside the frozen region (the area marked with a blue dotted line). The part of the lattice outside the frozen region is then a lattice of size $2(n-1)\times (m-1)$ with DWBC. This yields a recursion relation for the partition function.}
\label{lambdaneqmu1}
\end{figure}
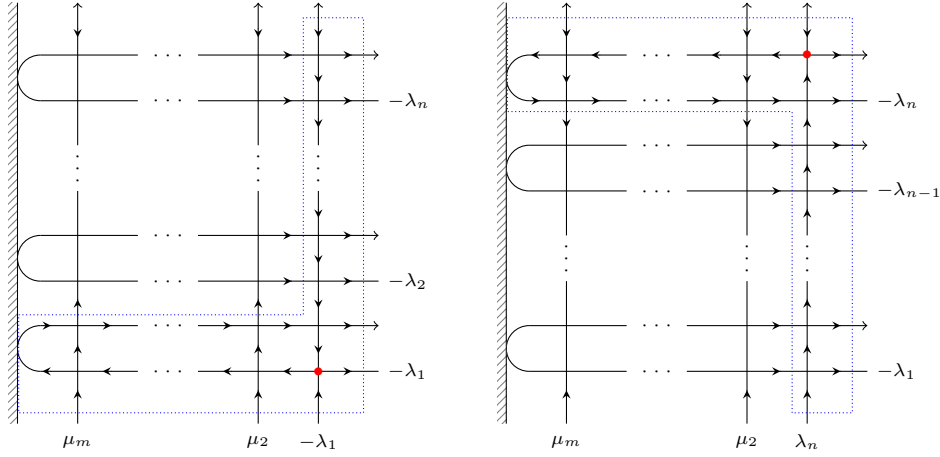

\begin{lemma}
\label{lemma:basestep}
For $m=0$, the partition function is
$$Z_{n,0}(\boldsymbol{\lambda})%=\prod_{i=1}^n k_c(\lambda_i)
=\varphi^n \prod_{i=1}^n f(2\lambda_i).$$
\end{lemma}

\begin{proof}
For $m=0$ there are no vertical lines. In this case, the partition function is just a product of $n$ turns of type $k_c$. %as in Figure~\ref{atmeq0}. 
\end{proof}

It makes little sense to think of a system with no lines, but for the sake of the following induction argument, we define $Z_{0,0}=1$, which is in line with the above lemma. 

%\begin{figure}[tb]
%\centering
%\begin{tikzpicture}[baseline={([yshift=-.5*10pt*0.6]current bounding box.center)}, scale=0.6, font=\scriptsize]
	%% first the horizontal border lines:
	%\foreach \y in {1,2,4} {
		%\draw[midarrow={stealth}] (0.55,1.5*\y-.25-.38) -- +(0.01,0);
		%\draw[midarrow={stealth}] (0.55,1.5*\y-.25+.38) -- +(0.01,0);
		%
		%\draw (0.38, 1.5*\y-.25-.38) -- +(0.62, 0);
		%\draw[->] (0.38, 1.5*\y-.25+.38) -- +(0.62, 0);
		%
		%\draw (0.38,1.5*\y-.25+.38) arc (90:270:0.38);
	%}
%
	%\foreach \y in {-1,...,1} \draw (0.7, 4.5-0.25+.2*\y) node{$\cdot\mathstrut$};
	%
	%\node[anchor=west] at (1, 1.5-.25-.38) {$-\lambda_1$};
	%\node[anchor=west] at (1, 1.5*2-.25-.38) {$-\lambda_{2}$};
	%\node[anchor=west] at (1, 1.5*4-.25-.38) {$-\lambda_n$};
	%
	%
%\end{tikzpicture}
%\caption{For $m=0$, the lattice consists of $n$ turns of type $k_c$.}
%\label{atmeq0}
%\end{figure}

Lemmas~\ref{lemma:2n1}, \ref{lemma:inductionstep} and \ref{lemma:basestep} together determine $Z_{n,m}(\boldsymbol\lambda, \boldsymbol\mu)$. A polynomial of degree $2n-1$ is uniquely determined by its values in $2n$ distinct points. Starting from the case $m=0$, we can hence establish $Z_{n,m}(\boldsymbol\lambda, \boldsymbol\mu)$ as a function of $\mu_j$ by induction, using Lemma~\ref{lemma:inductionstep}.

\begin{theorem}
\label{thm:determinantformula}
For the 6V model with DWBC and a partially reflecting end on a lattice of size $2n\times m$, $m\leq n$, the partition function is
\begin{multline}
Z_{n,m}(\boldsymbol{\lambda}, \boldsymbol{\mu})=\varphi^{n-m} e^{\left(\binom{m}{2}-nm\right)\gamma} f(\gamma)^m\prod_{i=1}^m \left[e^{\mu_i+\zeta} f(\mu_i-\zeta)\right]\prod_{i=1}^n f(2\lambda_i)\\
\times\frac{\prod_{i=1}^n\prod_{j=1}^m f(\mu_j\pm\lambda_i)}{\prod_{1\leq i<j\leq m}f(\mu_j\pm\mu_i)\prod_{1\leq i< j\leq n}\left[f(\lambda_i-\lambda_j)f(\lambda_i+\lambda_j+\gamma)\right]} \det_{1\leq i,j\leq n} M,
\label{determinantformula}
\end{multline}
where $M$ is an $n\times n$ matrix with
$$M_{ij}=
%\begin{pmatrix}[1.5]
%\frac{1}{f(\mu_1\pm\lambda_1)f(\mu_1\pm(\lambda_1+\gamma))} & \cdots & \frac{1}{f(\mu_1\pm\lambda_n)f(\mu_1\pm(\lambda_n+\gamma))}\\
%\vdots & \ddots & \vdots\\
%\frac{1}{f(\mu_m\pm\lambda_1)f(\mu_m\pm(\lambda_1+\gamma))} & \cdots & \frac{1}{f(\mu_m\pm\lambda_n)f(\mu_m\pm(\lambda_n+\gamma))}\\
%h((n-m-1)(2\lambda_1+\gamma)) & \cdots & h((n-m-1)(2\lambda_n+\gamma))\\
%\vdots & \ddots & \vdots\\
%h(2(2\lambda_1+\gamma)) & \cdots & h(2(2\lambda_n+\gamma))\\
%h(2\lambda_1+\gamma) & \cdots & h(2\lambda_n+\gamma)\\
%1 & \cdots & 1
%\end{pmatrix}
\begin{cases}
\dfrac{1}{f(\mu_i\pm\lambda_j)f(\mu_i\pm(\lambda_j+\gamma))}, &\text{for } i\leq m,\\[2ex]
h((n-i)(2\lambda_j+\gamma)), & \text{for } m< i<n,\\
1, &\text{for } i=n,
\end{cases}
$$
where $f(x)=2\sinh(x)$ and $h(x)=2\cosh(x)$.
\end{theorem}

To prove this, it is enough to check that the right hand side of \eqref{determinantformula} satisfies the properties stated in Lemmas~\ref{lemma:symmetryinlambda} through \ref{lemma:basestep}. 

\begin{proof}
Let $\tilde Z_{n,m}(\boldsymbol{\lambda}, \boldsymbol{\mu})$ be the right hand side of \eqref{determinantformula}. The symmetry in the $\lambda_i$'s and $\mu_j$'s is obvious, so the condition in Lemma~\ref{lemma:symmetryinlambda} holds for \eqref{determinantformula}. %interchanging $\lambda_i$ and $\lambda_{i+1}$ means switching the sign on $\prod_{1\leq i< j\leq n}f(\lambda_i-\lambda_j)$ and switching two adjacent columns in the matrix, which yields another switch of sign. Likewise, interchanging $\mu_i$ and $\mu_{i+1}$ yields one minus sign from the denominator and one minus sign from interchanging two adjacent rows of the matrix, so the condition in Lemma~\ref{lemma:symmetryinlambda} is true.
Proving that the statements in Lemma~\ref{lemma:inductionstep} and Lemma \ref{lemma:basestep} hold for \eqref{determinantformula} is straightforward: 
put $\lambda_l=\mu_k$ into \eqref{determinantformula}. Then all terms with $f(\mu_k-\lambda_l)=0$ vanish. The only terms left are the terms coming from the part of the determinant containing $f(\mu_k-\lambda_l)$, since the corresponding factor from the numerator cancels. Hence the terms left are
\begin{align*}
\tilde Z_{n,m}(\boldsymbol{\lambda}, \boldsymbol{\mu})|_{\mu_k=\lambda_l}&=\frac{(-1)^{k+l}e^{-n\gamma}f(\gamma) e^{\mu_k+\zeta}f(\mu_k-\zeta) f(2\lambda_l)}{f(\mu_k\pm(\lambda_l+\gamma))}\prod_{\mathclap{\substack{1\leq i \leq n\\ i\neq l}}}f(\mu_k\pm\lambda_i)\prod_{\mathclap{\substack{1\leq j \leq m\\ j\neq k}}}f(\mu_j\pm\lambda_l)\\
&\hphantom{=\ }\times\frac{1}{\prod_{1\leq i< l}\left[f(\lambda_i-\lambda_l)f(\lambda_i+\lambda_l+\gamma)\right]\prod_{l< i\leq n}\left[f(\lambda_l-\lambda_i)f(\lambda_l+\lambda_i+\gamma)\right]} \\
&\hphantom{=\ }\times\frac{1}{\prod_{1\leq i<k}f(\mu_k\pm\mu_i)\prod_{k<i\leq m}f(\mu_i\pm\mu_k)}\tilde Z_{n-1,m-1}(\boldsymbol{\hat\lambda_l}, \boldsymbol{\hat\mu_k}).
\end{align*}
%where $\det \left(M_{\text{minor}(k,l)}\right)$ is the minor of $M$ with the $k$th row and $l$th column removed. 
Simplifying this, we can conclude that the recurrence relation in Lemma~\ref{lemma:inductionstep} holds. The proof is similar for $\tilde Z_{n,m}(\boldsymbol{\lambda}, \boldsymbol{\mu})|_{\mu_k=-\lambda_l}$.

For $m=0$, the right hand side of \eqref{determinantformula} is
\begin{align*}
&\tilde Z_{n,0}(\boldsymbol{\lambda})=\frac{\varphi^n\prod_{i=1}^n f(2\lambda_i)}{\prod_{1\leq i< j\leq n}\left[f(\lambda_i-\lambda_j)f(\lambda_i+\lambda_j+\gamma)\right]} \det_{1\leq i,j\leq n} M,
\end{align*}
where $M$ is an $n\times n$ matrix with
$$M_{ij}=
\begin{cases}
h((n-i)(2\lambda_j+\gamma)), & \text{for } 1\leq i\leq n-1,\\
1, &\text{for } i=n.
\end{cases}
$$
Now let $$D=\prod_{1\leq i< j\leq n}\left[f(\lambda_i-\lambda_j)f(\lambda_i+\lambda_j+\gamma)\right]$$ and put $y_i=h(2\lambda_i+\gamma)$. Then we can write $D=\prod_{1\leq i< j\leq n} (y_i-y_j)$. By row operations, we can rewrite $\det M$ as a determinant of a Vandermonde matrix
$\tilde M_{ij}=y_j^{n-i}.$ %rewrite $y_j^a$, start from highest $a$ and realize that this is just a combination of lower rows, except for the part that is $\cosh(a(2\lambda_j+\gamma))$. 
The determinant of a Vandermonde matrix is exactly $D$. Hence the condition in Lemma~\ref{lemma:basestep} holds.

The only thing left is to prove the condition of Lemma~\ref{lemma:2n1}. Observe that both the determinant and the denominator of \eqref{determinantformula} are $0$ when $\mu_i=\pm\mu_j$. It is easy to see that $e^{2(m-1)\mu_i}$ times the denominator is a polynomial of order $2(m-1)$ in $e^{2\mu_i}$ and is hence determined by its value in $2(m-1)$ points (except for a constant factor). Since the determinant and the denominator share zeroes, the denominator is a factor in the determinant. This proves that $\tilde Z_{n,m}(\boldsymbol\lambda, \boldsymbol\mu)$ is a Laurent polynomial in $e^{\mu_i}$. 
To determine the order, we check what happens with the determinant when $\mu_i\to \infty$. 
%Each entry of the determinant involving $\mu_i$ is
Let
\begin{align*}
&F(x_i, y_j)=\frac{1}{f(\mu_i\pm\lambda_j)f(\mu_i\pm(\lambda_j+\gamma))}.
%&=\frac{e^{-4\mu_i}}{\splitfrac{(1-2e^{-\gamma-2\mu_i}\cosh(2\lambda_j+\gamma)+e^{-4\mu_i-2\gamma})}{\cdot(1-2e^{\gamma-2\mu_i}\cosh(2\lambda_j+\gamma)+e^{-4\mu_i+2\gamma})}}.
\end{align*}
For this part of the proof, let  $x_i=e^{-2\mu_i}$ and $y_j=h(2\lambda_j+\gamma)$. Then each entry of the determinant involving $\mu_i$ is
$$F(x_i, y_j)=\frac{x_i^2}{(1-e^{-\gamma}x_i y_j+e^{-2\gamma}x_i^2)(1-e^{\gamma}x_i y_j+e^{2\gamma}x_i^2)}.$$
We can see that $y_j$ always stands together with $x_i$, so in the Taylor series expansion, the power of $x_i$ is always bigger than the power of $y_j$. The Taylor series expansion around $x_i=0$ is 
\begin{equation*}
F(x_i, y_j)%=x_i^2 \sum_{\mathclap{0\leq l\leq k}} A_{k,l} x_i^k y_j^l
=x_i^2 \left(\sum_{\substack{l\leq n-m-1\\ l \leq k}} A_{k,l} x_i^k y_j^l+\sum_{\substack{n-m\leq l \\ l\leq k}} A_{k,l} x_i^k y_j^l\right),
\end{equation*}
for some constants $A_{k,l}$. The sum on the left can be written in terms of the last $n-m$ rows of the matrix. Hence in the determinant, the left sum can be removed by row operations, and we can change the entries $F(x_i, y_j)$ in the determinant to
$$x_i^2 \sum_{l=n-m}^k A_{k,l} x_i^k y_j^l.$$
Letting $x_i\to 0$ (i.e. letting $\mu_i\to \infty$), $A_{n-m,n-m} x_i^{n-m+2} y_j^{n-m}$ is the leading term of the determinant, so the degree of the determinant in the variable $e^{2\mu_i}$ is $m-n-2$. 
The degree of $e^{(2n-2)\mu_i}\prod_{i=1}^n\prod_{j=1}^m f(\lambda_i\pm\mu_j+\gamma) \tilde Z_{n,m}(\boldsymbol\lambda, \boldsymbol\mu)$ in $e^{2\mu_i}$ %when $\mu_i\to \infty$ 
is $$%2n-1+2n+2+2n-2(m-1)+\deg_{\mu_i} (\det M)=
3n-m+1+\deg_{e^{2\mu_i}} (\det M)%=n-m+1+m-n-2
=2n-1$$ which proves the condition in Lemma~\ref{lemma:2n1}.
%Fel tror jag The degree of $e^{(2n-1)\mu_i}\prod_{i=1}^n\prod_{j=1}^m f(\lambda_i\pm\mu_j+\gamma) Z_{n,m}(\boldsymbol\lambda, \boldsymbol\mu)=e^{(2m-2n-4)\mu_i}e^{(4n-2m+3)\mu_i}Z_{n,m}$ in $e^{\mu_i}$ when $\mu_i\to \infty$ is $(2m-2n-4)+(8n-4m+6)+\deg_{\mu_i} (\det M)=6n-2m+2+\deg_{\mu_i} (\det M)=6n-2m+2+2m-2n-4=4n-2$ which proves the condition in Lemma~\ref{lemma:2n1}.
\end{proof}

For $m=n$, the determinant formula in Theorem~\ref{thm:determinantformula} coincides with Tsuchiyas determinant formula.  

In \cite{FodaZarembo2016}, another way to find the partition function is suggested, by starting from Tsuchiyas determinant for $2n\times n$ lattices and step by step removing the extra vertical lines by letting the corresponding variables $\mu_j\to \infty$, following Foda and Wheeler \cite{FodaWheeler2012}. We prove Theorem~\ref{thm:determinantformula} with this method in Appendix A. 

\section{Kuperberg's specialization of the parameters}
\label{sec:countingstates}
In this section we specialize the parameters in the definition of the partition function \eqref{naivepartition}. By doing this, we find a way to count the total number of states of the model in terms of the partition function. %Now we will specialize $\gamma=4\pi \i /3$ and let $\lambda_i\to \gamma$ and $\mu_j \to 0$. 

First, we will state and prove the $2n\times m$ generalization of Lemma 3.1 in \cite{Hietala2020}. For any given state, let $\nu(w)$ denote the number of vertices of type $w$.
\begin{lemma}
\label{lemmanumberofbc}
For any given state of the 6V model with DWBC and one partially reflecting end on a $2n\times m$ lattice, we have
$$\nu(b_+)-\nu(b_-)=\binom{n+1}{2}-\binom{n-m+1}{2} \quad \text{ and } \quad \nu(c_+)-\nu(c_-)=m-2\nu(k_-).$$
\end{lemma}

\begin{proof}
We follow the same reasoning as in \cite{Hietala2020}. For the first result we count the number of different spin arrows. To be able to keep track of the arrow flow through the system, in particular where the arrows go up and down, we think of a $k_+$ turn as one left, one up and one right arrow, and likewise we think of a $k_-$ turn as one right, one down and one left arrow. To get the number of up arrows, we count the number of the different vertices and turns with up arrows, and similarly for the down arrows. To not count the edges twice, we count the vertices on every second row, and add the edges from the boundaries. Hence we get two different equations, depending on which rows we choose to count. Denote by $w^N$ a vertex of type $w$ on the upper part of a double line and similarly $w^S$ for a vertex of type $w$ on the lower part ($N$ for north, $S$ for south). We get the equations
\begin{equation}
\label{equ2}
\nu(\uparrow)=2\nu(a_+^S) + 2\nu(b_+^S) + \nu(c_+^S) + \nu(c_-^S) + \nu(k_+),
\end{equation}
\begin{equation}
\label{equ1}
\nu(\uparrow)=2\nu(a_+^N) + 2\nu(b_-^N) + \nu(c_+^N) + \nu(c_-^N) + \nu(k_+) + m,
\end{equation}
\begin{equation}
\label{equ4}
\nu(\downarrow)=2\nu(a_-^S) + 2\nu(b_-^S) + \nu(c_+^S) + \nu(c_-^S) + \nu(k_-)+m,
\end{equation}
and
\begin{equation}
\label{equ3}
\nu(\downarrow)=2\nu(a_-^N) + 2\nu(b_+^N) + \nu(c_+^N) + \nu(c_-^N) + \nu(k_-).
\end{equation}
To get the number of left arrows, we count the number of the different vertices with left arrows, plus the number of all $k_+$ and $k_-$ turns, then divide everything by two. The number of left arrows is
\begin{equation}
\label{equ5}
\nu(\leftarrow)=\nu(a_-^N) + \nu(b_-^N) + \nu(a_+^S) + \nu(b_-^S) + \frac{\nu(c_+) + \nu(c_-) +\nu(k_+)+\nu(k_-)}{2}.
\end{equation}
Similary, %we get the number of right arrows by counting the number of different vertices and turns with right arrows, plus the number of arrows on the right, and then divide by two. 
the number of right arrows is 
\begin{align}
\label{equ6}
\nu(\rightarrow)&=\nu(a_+^N) + \nu(b_+^N) + \nu(a_-^S) + \nu(b_+^S)+ \frac{\nu(c_+) + \nu(c_-)+\nu(k_+)+\nu(k_-)}{2}+\nu(k_c)+n.
\end{align}
Adding the equations \eqref{equ2}, \eqref{equ3} and two times \eqref{equ6}, and subtracting \eqref{equ4}, \eqref{equ1} and two times \eqref{equ5} yields
\begin{equation}
\label{equ7}
4[\nu(b_+)-\nu(b_-)]=2\nu(\rightarrow)-2\nu(\leftarrow)-2n+2m-2\nu(k_c).
\end{equation}
With the given boundary conditions we have $\nu(k_c)=n-m$, and furthermore, we must have
$\nu(\rightarrow)=2n(m+1)-\frac{m(m+1)}{2}$ and $\nu(\leftarrow)=\frac{m(m+1)}{2}$ to satisfy the ice rule. Insert this into \eqref{equ7}. The first part of the lemma follows. 

For the second part of the lemma we observe that at a $k_+$ double row, the upper part of the double line starts and ends with a right arrow, and the lower part of the line starts with a left arrow and ends with a right arrow. The only vertex types that can change the arrows are the $c_\pm$ vertices. Hence at each $k_+$ double row,
$$\nu(c_+^S)=\nu(c_-^S)+1, \qquad \nu(c_+^N)=\nu(c_-^N),$$
and similarly at each $k_-$ double row,
$$\nu(c_+^S)=\nu(c_-^S) \quad \text{ and } \quad \nu(c_+^N)=\nu(c_-^N)-1.$$ 
At a $k_c$ double row, each row starts and ends with right arrows, so
$$\nu(c_+^S)=\nu(c_-^S) \quad \text{ and } \quad \nu(c_+^N)=\nu(c_-^N).$$
Since $\nu(c_+)=\nu(c_-)$ on the $k_c$ rows, we only need to consider the $k_+$ and $k_-$ rows.
There are $m$ rows of this type. Hence we get the sum $\nu(c_+)=\nu(c_-)+m-2\nu(k_-)$, which is the second result of the lemma. 
\end{proof}

Furthermore we can count the number of $a_\pm$ vertices.
\begin{cor}
\label{lemmanumberofa}
For any given state of the 6V model with DWBC and one partially reflecting end on a $2n\times m$ lattice, the number of vertices of type $a$ (i.e. $a_+$ and $a_-$ together) is $$\nu(a)=mn+\binom{m}{2}-m+2[\nu(k_-)-\nu(b_-)-\nu(c_-)].$$
\end{cor}

\begin{proof}
We have that $\nu(a)=2mn-\nu(b)-\nu(c)$. Use Lemma~\ref{lemmanumberofbc} to compute $\nu(b)$ in terms of $\nu(b_-)$, and $\nu(c)$ in terms of $\nu(c_-)$. 
\end{proof}

%The partition function is $$Z_{n,m}(\lambda_1, \dots, \lambda_n, \mu_1, \dots, \mu_m)=\sum_{\text{states}} \prod_{\text{nodes}} w(\lambda_i\pm \mu_j) \prod_{\text{turns}} k(\lambda_i),$$ where $w$ is one of $a_\pm$, $b_\pm$ or $c_\pm$, and $k$ is either $k_\pm$ or $k_c$, and where $\lambda_i$ and $\mu_j$ are the spectral parameters belonging to each vertex or turn. 
Now let $\gamma=4\pi \i /3$ so that $e^\gamma$ becomes a third root of $1$. Then let $\lambda_i\to\gamma$ and $\mu_j\to 0$ in the partition function $Z_{n,m}(\boldsymbol\lambda, \boldsymbol\mu)$. In this limit, the weights \eqref{vertexweights} and \eqref{turnweights} are
%\begin{align*}
%&a_\pm(\gamma)=1, \qquad b_+(\gamma)=c_-(\gamma)=\frac{e^{-\gamma}f(\gamma)}{f(2\gamma)}, \qquad b_-(\gamma)=c_+(\gamma)=\frac{e^{\gamma}f(\gamma)}{f(2\gamma)},\\
%&k_+(\gamma)=e^{2\zeta}-e^{-2\gamma}, \qquad
%k_-(\gamma)=e^{2\zeta}-e^{2\gamma}, \qquad
%k_c(\gamma)=xf(2\gamma).
%\end{align*}
\begin{align*}
&a_\pm(\gamma)=1,&\ &b_+(\gamma)=c_-(\gamma)=-e^{-\gamma},&\  &b_-(\gamma)=c_+(\gamma)=-e^{\gamma}f(\gamma),\\
&k_+(\gamma)=e^{2\zeta}-e^{-2\gamma},&\
&k_-(\gamma)=e^{2\zeta}-e^{2\gamma}, &\
&k_c(\gamma)=\varphi f(2\gamma),
\end{align*}
and the partition function becomes
\begin{align*}
Z_{n,m}(\gamma, \dots, \gamma, 0, \dots, 0)&=\sum_{\text{states}} a_+(\gamma)^{\nu(a_+)}a_-(\gamma)^{\nu(a_-)}b_+(\gamma)^{\nu(b_+)} b_-(\gamma)^{\nu(b_-)} \\
&\hphantom{=\ } \times c_+(\gamma)^{\nu(c_+)}c_-(\gamma)^{\nu(c_-)}
k_+(\gamma)^{\nu(k_+)}k_-(\gamma)^{\nu(k_-)}k_c(\gamma)^{\nu(k_c)}.
\end{align*}
Here we have used that for $\gamma=4\pi \i /3$, it holds that $f(2\gamma)=-f(\gamma)$.
Inserting the expressions for the weights and using Lemma~\ref{lemmanumberofbc} and Corollary~\ref{lemmanumberofa}, the partition function can be computed to be
\begin{align*}
Z_{n,m}(\gamma, \dots, \gamma, 0, \dots, 0)
&= (-1)^{\binom{m}{2}-nm+n}\varphi^{n-m} e^{\left(\binom{m+1}{2}-nm\right)\gamma}f(\gamma)^{n-m}\\
&\hphantom{=\ } \times \sum_{k=0}^m N_k (e^{2\zeta}-e^{-2\gamma} )^{k}\left(\frac{e^{2\zeta}-e^{2\gamma}}{e^{2\gamma}} \right)^{m-k}, 
\end{align*}
where $N_k$ is the number of states with $\nu(k_+)=k$, i.e. where $k$ is the number of $k_+$ turns.
Let $Z_{n,m, k}(\boldsymbol\lambda, \boldsymbol\mu)$
be the partition function for the same model where the number of $k_+$ turns is fixed to $\nu(k_+)=k$. Then 
\begin{align} 
\label{Nk}
N_k&=(-1)^{\binom{m}{2}-nm+n}\varphi^{m-n} e^{\left(nm-\binom{m+1}{2}\right)\gamma}f(\gamma)^{m-n}\nonumber \\
&\hphantom{=\ } \times\left(\frac{1}{e^{2\zeta}-e^{-2\gamma}}\right)^k\left(\frac{e^{2\gamma}}{e^{2\zeta}-e^{2\gamma}} \right)^{m-k} Z_{n,m,k}(\gamma, \dots, \gamma, 0,\dots, 0).
\end{align}
Finally
\begin{equation}
\label{numberofstates}
A(m,n)\coloneqq\sum_{k=0}^m N_k
\end{equation}
 counts the total number of states of the model. 

\subsection{Connection to alternating sign matrices}
\label{subsec:ASMs}
By specializing the parameters in the partition function for the 6V model with DWBC (without reflecting end), Kuperberg \cite{Kuperberg1996} was able to count the number of alternating sign matrices (ASMs). ASMs are matrices where the entries consist of $1$, $-1$ and $0$ such that the nonzero elements in each row and each column alternate in sign and such that the sum of the elements of any row and any column is $1$. Hence the first and last nonzero element of each row and each column is always $1$. 

%In the case of DWBC and reflecting end where $n=m$, the formula $A(m,n)$ \eqref{numberofstates} corresponds to the number of UASMs \cite{Kuperberg2002} of size $2n\times n$. These are ASMs with U-turns on the left side, where the rows are connected pairwise. Following the lower part of such a double row from the right to the left, and thereafter the upper part from the left to the right, the nonzero elements have alternating signs and such a double row starts and ends with $1$. Vertically the rules are the same as for the original ASMs. 
%There is a bijection between the 6V model with DWBC and reflecting end and UASMs, see Figure~\ref{fig:bijection6VUASM}. A $c_-$ vertex in the upper part of a double row corresponds to $1$, and $c_+$ corresponds to $-1$. On the lower part of a double row, $c_+$ corresponds to $1$, and $c_-$ corresponds to $-1$. All other vertices correspond to $0$. 
%In this way we get a bijection from the 6V model with DWBC and one partially reflecting end to this special type of cropped UASMs of size $2n\times m$ that follow the rules for ASMs vertically, however horizontally the rules are a bit different. %From the right they follow the rules for the UASMs, but on the left boundary at the U-turn, the rules for ASMs and UASMs do not necessarily hold. 

In the case of DWBC and partially reflecting end, we also get a bijection to a type of ASM-like objects. 
Consider matrices of size $2n\times m$, $m\leq n$, consisting of elements $0$, $-1$ and $1$. Vertically and horizontally the nonzero elements alternate in sign. The sum of the elements of each column is $1$, as for ASMs. Horizontally connect the rows pairwise on the left edge to form a double row as in Figure~\ref{fig:croppedASMs}. A double row may consist of only zeroes. If a row has any nonzero elements, the rightmost of these must be $1$. %The sum of the entries in a row must always be $0$ or $1$, and 
Furthermore the sum of the entries in a double row must be $0$ or $1$. Equivalently, a double row can not consist of two rows both having 1 as their leftmost nonzero element.

\begin{prop}
\label{prop:croppedASMs}
The expression
$A(m,n)$
yields the number of matrices of size $2n\times m$, $m\leq n$, described above.
\end{prop}

\begin{proof}
There is a bijection between the states of the 6V model and matrices consisting of $0$, $1$ and $-1$ \cite{Kuperberg1996} (see Figure~\ref{fig:bijection6VUASM}), where $a_\pm$ and $b_\pm$ corresponds to $0$ in the matrix, and in the case of reflecting end, a $c_-$ vertex in the upper part of a double row corresponds to $1$, and $c_+$ corresponds to $-1$. On the lower part of a double row, $c_+$ corresponds to $1$, and $c_-$ corresponds to $-1$.
Due to the ice rule and the boundary conditions, the rest of the vertices are determined uniquely. 

More rows than columns allows for rows without $c_\pm$ vertices, corresponding to rows with only zeroes in the matrix. If there are any $c_\pm$ vertices on the lower part of a double row, $k_-$ and $k_c$ turns on that double row force the leftmost $c_\pm$ vertex to be $c_-$, corresponding to $-1$ in the matrix, whereas a $k_+$ turn forces it to be $c_+$, corresponding to $1$ in the matrix. Similarly if there are any $c_\pm$ vertices on the upper part of the double row, $k_+$ and $k_c$ force the leftmost nonzero element in the matrix to be $-1$, and $k_-$ yields that the leftmost nonzero matrix element is $1$. In this way we can see that both rows can not have $1$ as their leftmost nonzero element. In a similar way the DWBC impose that if a row has any nonzero element, then the rightmost of them must be $1$. 
\end{proof}

From the discussion above, it follows that $N_k$ \eqref{Nk} counts the number of matrices equivalent to states of the 6V model with exactly $k$ positive turns.
In the case $m=n$, the sum of the elements in each double row must be $1$, and $A(m,n)$ counts the number of UASMs \cite{Kuperberg2002}. %In this way the matrices in the proposition above can be seen as a type of cropped UASMs. Not all UASMs can be cropped to yield a matrix of the type we consider here, but the matrices that we consider can all be extended to UASMs by adding columns. 
The case where $m=n$ and $k=0$ corresponds to VSASMs.
 
%The states of the 6V model with DWBC and one partially reflecting end of size $2n\times m$, $m<n$, can be thought of as having removed the $n-m$ leftmost columns from a state of size $2n\times n$, see Appendix A. In the same way we can remove the $n-m$ leftmost columns of a corresponding UASM.

\begin{figure}
\centering
\subcaptionbox*{}{%	
		$
\tikz[baseline={([yshift=-.5*10pt*0.6+1pt]current bounding box.center)},scale=0.6]{
	\draw (0,0) arc (90:270:0.3);
	\draw (0,1.4) arc (90:270:0.3);
	\draw (0,2.8) arc (90:270:0.3);
	}
\left.
\begin{matrix}
1 & 0  \\
0 & 0 \\
-1 & 1 \\
0 & 0 \\
1 & 0 \\
0 & 0 \\
\end{matrix}
\right)
$}\hfil
\subcaptionbox*{}{%	
		$
\tikz[baseline={([yshift=-.5*10pt*0.6+1pt]current bounding box.center)},scale=0.6]{
	\draw (0,0) arc (90:270:0.3);
	\draw (0,1.4) arc (90:270:0.3);
	\draw (0,2.8) arc (90:270:0.3);
	}
\left.
\begin{matrix}
1 & 0  \\
0 & 0 \\
0 & 0 \\
-1 & 1 \\
1 & 0 \\
0 & 0 \\
\end{matrix}
\right)
$}\hfil
\subcaptionbox*{}{%	
		$
\tikz[baseline={([yshift=-.5*10pt*0.6+1pt]current bounding box.center)},scale=0.6]{
	\draw (0,0) arc (90:270:0.3);
	\draw (0,1.4) arc (90:270:0.3);
	\draw (0,2.8) arc (90:270:0.3);
	\draw (0,4.2) arc (90:270:0.3);
	}
\left.
\begin{matrix}
1 & 0 & 0 \\
0 & 0 & 0 \\
0 & 1 & 0 \\
0 & 0 & 0 \\
0 & -1 & 1 \\
-1 & 1 & 0 \\
1 & 0 & 0 \\
0 & 0 & 0 \\
\end{matrix}
\right)
$}
\hfil
\subcaptionbox*{}{%	
		$
\tikz[baseline={([yshift=-.5*10pt*0.6+1pt]current bounding box.center)},scale=0.6]{
	\draw (0,0) arc (90:270:0.3);
	\draw (0,1.4) arc (90:270:0.3);
	\draw (0,2.8) arc (90:270:0.3);
	}
\left.
\begin{matrix}
0 & 0 \\
0 & 0  \\
0 & 1 \\
1 & 0 \\
0 & 0 \\
0 & 0 \\
\end{matrix}
\right)
$}
\vspace{-3mm}
\caption{The three matrices on the left are counted by $A(m,n)$, whereas the rightmost matrix does not correspond to a state in the model considered in this paper.}
\label{fig:croppedASMs}
\end{figure}

\begin{figure}
\centering
$
\tikz[baseline={([yshift=-.5*10pt*0.6+1pt]current bounding box.center)},scale=0.6]{
	\draw (0,0) arc (90:270:0.3);
	\draw (0,1.4) arc (90:270:0.3);
	\draw (0,2.8) arc (90:270:0.3);
	}
\left.
\begin{matrix}
1 & 0 & 0 \\
-1 & 1 & 0 \\
0 & -1 & 1 \\
1 & 0 & 0 \\
0 & 1 & 0 \\
0 & 0 & 0 \\
\end{matrix}
\right)
\ \ \longleftrightarrow \ \ $
\begin{tikzpicture}[baseline={([yshift=-.5*10pt*0.6]current bounding box.center)}, scale=0.6, font=\small]
	% first the horizontal border lines:
	\foreach \y in {1,...,3} {
		\draw[midarrow={stealth}] (3.55,1.5*\y-.25-.38) -- +(0.01,0);
		\draw[midarrow={stealth}] (3.55,1.5*\y-.25+.38) -- +(0.01,0);
		
		\draw (1, 1.5*\y-.25-.38) -- +(3, 0);
		\draw[->] (1, 1.5*\y-.25+.38) -- +(3, 0);
		
		\draw (0.38,1.5*\y-.25+.38) arc (90:270:0.38);
		}
		
	% then the vertical border lines:
	\foreach \x in {1,...,3} {
		\draw (\x,0) -- +(0,.87); % dotted lines for \ket{0}
			\draw[midarrow={stealth}] (\x,0.55) -- +(0,0.01);
		\draw[->] (\x,4.63)  -- +(0,.87);	
			\draw[midarrow={stealth reversed}] (\x,4.63+.43)  -- +(0,0.01);	
			
		\draw (\x,.87) -- +(0,4.63); 
	}
	
		%left border
		%\draw[-stealth reversed] (0,1.55-0.25) -- (0,1.56-0.25);
		\draw(.38,1.25-.38) -- (1,1.25-.38);
			\draw[midarrow={stealth}] (.55,1.25-.38) -- +(0.01,0);
		\draw(.38,1.25+.38) -- (1,1.25+.38);
			\draw[midarrow={stealth reversed}] (.55,1.25+.38) -- +(0.01,0);
		%\node at (0.5, 1.25) {$1$};
		
		%\draw[-stealth] (0,3.05-0.25) -- (0,3.06-0.25);
		\draw (.38,2.75-.38) -- (1,2.75-.38);
			\draw[midarrow={stealth reversed}] (.55,2.75-.38) -- +(0.01,0);
		\draw(.38,2.75+.38) -- (1,2.75+.38);
			\draw[midarrow={stealth}] (.55,2.75+.38) -- +(0.01,0);
		%\node at (0.5, 2.75) {$2$};
		
		%\draw[-stealth reversed] (0,4.45-0.25) -- (0,4.56-0.25);
		\draw(.38,4.25-.38) -- (1,4.25-.38);
			\draw[midarrow={stealth}] (.55,4.25-.38) -- +(0.01,0);
		\draw (.38,4.25+.38) -- (1,4.25+.38);
			\draw[midarrow={stealth reversed}] (.55,4.25+.38) -- +(0.01,0);
		%\node at (0.5, 4.25) {$1$};
		
		%bulk horizontal lines
		\draw [midarrow={stealth}] (1.55, 1.5-.25-.38) -- +(0.01, 0);
		\draw [midarrow={stealth}] (2.55, 1.5-.25-.38) -- +(0.01, 0);
		\draw [midarrow={stealth reversed}](1.55, 1.5-.25+.38) -- +(0.01, 0);
		\draw [midarrow={stealth}] (2.55, 1.5-.25+.38) -- +(0.01, 0);
		
		\draw [midarrow={stealth}] (1.55, 3-.25-.38) -- +(0.01, 0);
		\draw [midarrow={stealth}] (2.55, 3-.25-.38) -- +(0.01, 0);
		\draw [midarrow={stealth}] (1.55, 3-.25+.38) -- +(0.01, 0);
		\draw [midarrow={stealth reversed}](2.55, 3-.25+.38) -- +(0.01, 0);
		
		\draw [midarrow={stealth reversed}] (1.55, 4.5-.25-.38) -- +(0.01, 0);
		\draw [midarrow={stealth}] (2.55, 4.5-.25-.38) -- +(0.01, 0);
		\draw [midarrow={stealth}](1.55, 4.5-.25+.38) -- +(0.01, 0);
		\draw [midarrow={stealth}] (2.55, 4.5-.25+.38) -- +(0.01, 0);
		
	%bulk vertical lines
			\draw [midarrow={stealth}] (1,0.87+.43) -- +(0,0.01); 
			\draw [midarrow={stealth}] (1,1.63+.43) -- +(0,0.01); 
			\draw [midarrow={stealth reversed}] (1,2.37+.43) -- +(0,0.01); 
			\draw [midarrow={stealth reversed}] (1,3.13+.43) -- +(0,0.01); 
			\draw [midarrow={stealth}] (1,3.87+.43) -- +(0,0.01);  
			
			\draw [midarrow={stealth}] (2,0.87+.43) -- +(0,0.01); 
			\draw [midarrow={stealth reversed}] (2,1.63+.43) -- +(0,0.01); 
			\draw [midarrow={stealth reversed}] (2,2.37+.43) -- +(0,0.01); 
			\draw [midarrow={stealth}] (2,3.13+.43) -- +(0,0.01); 
			\draw [midarrow={stealth reversed}] (2,3.87+.43) -- +(0,0.01);
			
			\draw [midarrow={stealth}] (3,0.87+.43) -- +(0,0.01); 
			\draw [midarrow={stealth}] (3,1.63+.43) -- +(0,0.01); 
			\draw [midarrow={stealth}] (3,2.37+.43) -- +(0,0.01); 
			\draw [midarrow={stealth reversed}] (3,3.13+.43) -- +(0,0.01); 
			\draw [midarrow={stealth reversed}] (3,3.87+.43) -- +(0,0.01);
			
			% then the wall
	\fill[preaction={fill,white},pattern=north east lines, pattern color=gray] (0,0) rectangle (-.15,5.5) ; \draw (0,0) -- (0,5.5);
	
\end{tikzpicture}
\caption{The bijection between a UASM and a state of the 6V model with DWBC and a reflecting end.}
\label{fig:bijection6VUASM}
\end{figure}
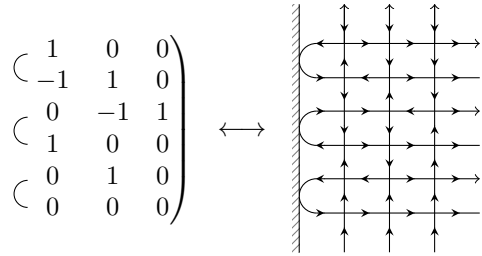

\section{Specialization in the determinant}
\label{sec:rewritepartfcn}
In this section we specialize the parameters in the determinant formula in the same way as in the previous section. We show that with this specialization of the variables, the partition function can be written as a determinant of Wilson polynomials, and finally we use this together with the results from the previous section to give a formula for the number of states of the model. 

\subsection{Moments and orthogonal polynomials}
We start by rewriting the determinant formula of the partition function \eqref{determinantformula}.
%
%\begin{align}
%&Z_{n,m}(\boldsymbol\lambda, \mu_1, \dots, \mu_m)=x^{n-m}e^{\left(\binom{m}{2}-nm\right)\gamma}f^m(\gamma)\nonumber\\
%&\cdot
%\frac{\prod_{i=1}^{m} \left[e^{\mu_i+\zeta} f(\mu_i-\zeta)\right]\prod_{i=1}^n f(2\lambda_i)\prod_{i=1}^n\prod_{j=1}^{m} f(\mu_j\pm\lambda_i)}{\prod_{1\leq i<j\leq m}f(\mu_j\pm\mu_i)\prod_{1\leq i<j\leq n}\left[f(\lambda_i-\lambda_j)f(\lambda_i+\lambda_j+\gamma)\right]}\nonumber\\
%&\cdot\det 
%\begin{pmatrix}[1.5]
%\frac{1}{f(\mu_1\pm\lambda_1)f(\mu_1\pm(\lambda_1+\gamma))} & \cdots & \frac{1}{f(\mu_1\pm\lambda_n)f(\mu_1\pm(\lambda_n+\gamma))}\\
%\vdots & \ddots & \vdots\\
%\frac{1}{f(\mu_m\pm\lambda_1)f(\mu_m\pm(\lambda_1+\gamma))} & \cdots & \frac{1}{f(\mu_m\pm\lambda_n)f(\mu_m\pm(\lambda_n+\gamma))}\\
%2\cosh(2(n-m-1)(\lambda_1+\gamma/2)) & \cdots & 2\cosh(2(n-m-1)(\lambda_n+\gamma/2))\\
%\vdots & \ddots & \vdots\\
%2\cosh(2(\lambda_1+\gamma/2)) & \cdots & 2\cosh(2(\lambda_n+\gamma/2))\\
%1 & \cdots & 1
%\end{pmatrix}.
%\end{align}
%
Let $x_i=\lambda_i-\gamma$ and $y_j=\mu_j$. 
In the previous section we let $\lambda_i\to \gamma$ and $\mu_i\to 0$ which corresponds to letting $x_i\to 0$ and $y_i\to 0$. 
For $1\leq i \leq n$, multiply column $i$ of the determinant by $f(2x_i)$ and for $1\leq j\leq m$, multiply row $j$ of the determinant by $f(2y_j)$. 
%Then
%\begin{align*}
%&\hat Z_{n,m}(x_1, \dots, x_n, y_1, \dots, y_m) \coloneqq Z_{n,m}(x_1 + \gamma, \dots, x_n+\gamma, y_1, \dots, y_m)\nonumber\\
%&=\frac{\varphi^{n-m}e^{\left(\binom{m}{2}-nm\right)\gamma}f^m(\gamma)\prod_{i=1}^{m} \left[e^{y_i+\zeta} f(y_i-\zeta)\right]}{\prod_{i=1}^n f(2x_i)\prod_{j=1}^m f(2y_j)}\nonumber\\
%&\hphantom{=\ }\cdot
%\frac{\prod_{i=1}^n f(2x_i+2\gamma)\prod_{i=1}^n\prod_{j=1}^{m} f(y_j\pm(x_i+\gamma))}{\prod_{1\leq i<j\leq m}f(y_j\pm y_i)\prod_{1\leq i<j\leq n}\left[f(x_i-x_j)f(x_i+x_j+3\gamma)\right]}
%%&\cdot\det 
%%\begin{pmatrix}[1.5]
%%\frac{f(2x_1)f(2y_1)}{f(y_1\pm (x_1+\gamma))f(y_1\pm(x_1+2\gamma))} & \cdots & \frac{f(2x_n)f(2y_1)}{f(y_1\pm(x_n+\gamma))f(y_1\pm(x_n+2\gamma))}\\
%%\vdots & \ddots & \vdots\\
%%\frac{f(2x_1)f(2y_m)}{f(y_m\pm(x_1+\gamma))f(y_m\pm(x_1+2\gamma))} & \cdots & \frac{f(2x_n)f(2y_m)}{f(y_m\pm(x_n+\gamma))f(y_m\pm(x_n+2\gamma))}\\
%%2f(2x_1)\cosh(2(n-m-1)(x_1+3\gamma/2)) & \cdots & 2f(2x_n)\cosh(2(n-m-1)(x_n+3\gamma/2))\\
%%\vdots & \ddots & \vdots\\
%%2f(2x_1)\cosh(2(x_1+3\gamma/2)) & \cdots & 2f(2x_n)\cosh(2(x_n+3\gamma/2))\\
%%f(2x_1) & \cdots & f(2x_n)
%%\end{pmatrix}.\\
%%=
%\det_{1\leq i, j \leq n} M.
%\end{align*}
%where 
%\begin{equation*}
%M_{ij}=\begin{cases}
%\frac{f(2x_j)f(2y_i)}{f(y_i\pm (x_j+\gamma))f(y_i\pm(x_j+2\gamma))} & \text{for } 1\leq i\leq m,\\ 
%f(2x_j)h(2(n-i)(x_j+3\gamma/2)) & \text{for } m+1\leq i\leq n-1,\\
%f(2x_j) & \text{for } i=n.
%\end{cases}
%\end{equation*}
Then specify $\gamma=4\pi \i /3$. 
%
%Then the partition function becomes
%\begin{align}
%&\hat Z_{n,m}(x_1, \dots, x_n, y_1, \dots, y_m) =\frac{x^{n-m}e^{\left(\binom{m}{2}-nm\right)\gamma}f^m(\gamma)}{\prod_{i=1}^n f(2x_i)\prod_{j=1}^m f(2y_j)}\nonumber\\
%&\cdot \frac{\prod_{i=1}^{m} \left[e^{y_i+\zeta} f(y_i-\zeta)\right]\prod_{i=1}^n f(2x_i+2\gamma)\prod_{i=1}^n\prod_{j=1}^{m} f(y_j\pm(x_i+\gamma))}{\prod_{1\leq i<j\leq m}f(y_j\pm y_i)\prod_{1\leq i<j\leq n}f(x_i\pm x_j)}\nonumber\\
%%&\cdot\det 
%%\begin{pmatrix}[1.5]
%%\frac{f(2x_1)f(2y_1)}{f(y_1\pm (x_1+\gamma))f(y_1\pm(x_1+2\gamma))} & \cdots & \frac{f(2x_n)f(2y_1)}{f(y_1\pm(x_n+\gamma))f(y_1\pm(x_n+2\gamma))}\\
%%\vdots & \ddots & \vdots\\
%%\frac{f(2x_1)f(2y_m)}{f(y_m\pm(x_1+\gamma))f(y_m\pm(x_1+2\gamma))} & \cdots & \frac{f(2x_n)f(2y_m)}{f(y_m\pm(x_n+\gamma))f(y_m\pm(x_n+2\gamma))}\\
%%2f(2x_1)\cosh(2(n-m-1)x_1) & \cdots & 2f(2x_n)\cosh(2(n-m-1)x_n)\\
%%\vdots & \ddots & \vdots\\
%%2f(2x_1)\cosh(2x_1) & \cdots & 2f(2x_n)\cosh(2x_n)\\
%%f(2x_1) & \cdots & f(2x_n)
%%\end{pmatrix}.\\
%%=
%&\begin{pmatrix}[1.5]
%\frac{f(2x_j)f(2y_i)}{f(y_i\pm (x_j+\gamma))f(y_i\pm(x_j+2\gamma))} & \text{for } 1\leq i\leq m\\ 
%2f(2x_j)\cosh(2(n-i)x_j) & \text{for } m+1\leq i\leq n-1\\
%f(2x_j) & \text{for } m+1\leq i\leq n
%\end{pmatrix}.
%\end{align}
By row operations on the lower part of the determinant we can then rewrite the partition function as
\begin{align}
\label{partfcn6}
%&\hat Z_{n,m}(x_1, \dots, x_n, y_1, \dots, y_m) \coloneqq 
&Z_{n,m}(x_1 + \gamma, \dots, x_n + \gamma, y_1, \dots, y_m)=\varphi^{n-m}e^{\left(\binom{m}{2}-nm\right)\gamma}f(\gamma)^m\nonumber\\
&\hphantom{=\ }\times
\dfrac{\prod_{i=1}^{m} \left[e^{y_i+\zeta} f(y_i-\zeta)\right]\prod_{i=1}^n f(2x_i+2\gamma)\prod_{i=1}^n\prod_{j=1}^{m} f(y_j\pm(x_i+\gamma))}{\prod_{1\leq i<j\leq m}f(y_j\pm y_i)\prod_{1\leq i<j\leq n}f(x_i\pm x_j)\prod_{i=1}^n f(2x_i)\prod_{j=1}^m f(2y_j)}\det_{1\leq i,j\leq n} \tilde M,
%&\hphantom{=\ }\times\det 
%\begin{pmatrix}[1.5]
%g(x_1, y_1) & \cdots & g(x_n, y_1)\\
%\vdots & \ddots & \vdots\\
%g(x_1, y_m) & \cdots & g(x_n, y_m)\\
%f(2(n-m)x_1) & \cdots & f(2(n-m)x_n)\\
%\vdots & \ddots & \vdots\\
%%f(4x_1) & \cdots & f(4x_n)\\
%f(2x_1) & \cdots & f(2x_n)
%\end{pmatrix},
\end{align}
where
\begin{equation*}
\tilde M_{ij}=
\begin{cases}
g(x_j, y_i), &\text{ for } 1\leq i \leq m,\\
f(2(n-i+1) x_j), &\text{ for } m+1\leq i\leq n,
\end{cases}
\end{equation*}
with $$g(x,y)=\frac{f(2x)f(2y)}{f(y\pm (x+\gamma))f(y\pm(x+2\gamma))}.$$
It is apparent that $f(x)=\sinh x$ and $g(x,y)$ are odd functions. Furthermore $x=0$ is a zero of $f(x)$, and $x=0$ and $y=0$ are zeroes of $g(x,y)$. Hence we can write the functions as $f(x)=x \hat f(x^2)$ and $g(x,y)=xy \hat g(x^2,y^2)$, where $\hat f$ and $\hat g$ are analytic at $x=y=0$. 
The partition function \eqref{partfcn6} can thus be written
\begin{align*}
%\hat Z_{n,m}(x_1, \dots, x_n, y_1, \dots, y_m)
&Z_{n,m}(x_1 + \gamma, \dots, x_n + \gamma, y_1, \dots, y_m)\\
%=\frac{x^{n-m}e^{\left(\binom{m}{2}-nm\right)\gamma}f^m(\gamma)}{\prod_{i=1}^n f(2x_i)\prod_{j=1}^m f(2y_j)}\nonumber\\
%&\cdot
%\frac{\prod_{i=1}^{m} \left[e^{y_i+\zeta} f(y_i-\zeta)\right]\prod_{i=1}^n f(2x_i+2\gamma)\prod_{i=1}^n\prod_{j=1}^{m} f(y_j\pm(x_i+\gamma))}{\prod_{1\leq i<j\leq m}f(y_j\pm y_i)\prod_{1\leq i<j\leq n}f(x_i\pm x_j)}\nonumber\\
%&\cdot\det_{0\leq i, j \leq n}
%\begin{pmatrix}[1.5]
%g(x_j, y_i) & \text{for } 1\leq i\leq m\\ 
%f(2(n+1-i)x_j) & \text{for } m+1\leq i\leq n
%\end{pmatrix}\\
%=\frac{x^{n-m}e^{\left(\binom{m}{2}-nm\right)\gamma}f^m(\gamma)}{\prod_{i=1}^n f(2x_i)\prod_{j=1}^m f(2y_j)}\nonumber\\
%&\cdot
%\frac{\prod_{i=1}^{m} \left[e^{y_i+\zeta} f(y_i-\zeta)\right]\prod_{i=1}^n f(2x_i+2\gamma)\prod_{i=1}^n\prod_{j=1}^{m} f(y_j\pm(x_i+\gamma))}{\prod_{1\leq i<j\leq m}f(y_j\pm y_i)\prod_{1\leq i<j\leq n}f(x_i\pm x_j)}\nonumber\\
%&\cdot\det
%\begin{pmatrix}[1.5]
%x_1y_1\hat g(x_1^2, y_1^2) & \cdots & x_ny_1\hat g(x_n^2, y_1^2)\\
%\vdots & \ddots & \vdots\\
%x_1y_m\hat g(x_1^2, y_m^2) & \cdots & x_ny_m\hat g(x_n^2, y_m^2)\\
%2(n-m) x_1 \hat f((2i(n-m)x_1)^2) & \cdots & 2(n-m) x_n\hat f((2i(n-m)x_n)^2)\\
%\vdots & \ddots & \vdots\\
%2x_1 \hat f((2ix_1)^2) & \cdots & 2x_n\hat f((2ix_n)^2)
%\end{pmatrix}\\
&=2^{n-m}(n-m)!\ \varphi^{n-m}e^{\left(\binom{m}{2}-nm\right)\gamma}f(\gamma)^m \prod_{i=1}^{m} \left[e^{y_i+\zeta} f(y_i-\zeta)\right] \nonumber\\
&\hphantom{=\ }\times
\frac{\prod_{i=1}^n f(2x_i+2\gamma)\prod_{i=1}^n\prod_{j=1}^{m} f(y_j\pm(x_i+\gamma))\prod_{j=1}^n x_j \prod_{k=1}^m y_k}{\prod_{1\leq i<j\leq m}f(y_j\pm y_i)\prod_{1\leq i<j\leq n}f(x_i\pm x_j)\prod_{i=1}^n f(2x_i)\prod_{j=1}^m f(2y_j)}\det_{1\leq i, j\leq n}\hat M,
%&\hphantom{=\ }\times\det
%\begin{pmatrix}[1.5]
%\hat g(x_1^2, y_1^2) & \cdots & \hat g(x_n^2, y_1^2)\\
%\vdots & \ddots & \vdots\\
%\hat g(x_1^2, y_m^2) & \cdots & \hat g(x_n^2, y_m^2)\\
%\hat f((2(n-m)x_1)^2) & \cdots & \hat f((2(n-m)x_n)^2)\\
%\vdots & \ddots & \vdots\\
 %\hat f((2 x_1)^2) & \cdots & \hat f((2 x_n)^2)
%\end{pmatrix}, 
\end{align*}
where
\begin{equation*}
\hat M_{ij}=
\begin{cases}
\hat g(x_j^2, y_i^2), &\text{ for } 1\leq i \leq m,\\
\hat f((2(n-i+1) x_j)^2), &\text{ for } m+1\leq i\leq n.
\end{cases}
\end{equation*}

Regarding $x_i^2$ and $y_j^2$ as our variables, then in the limit where all $x_i,y_j \to 0$, the determinant can be written in terms of partial derivatives. Subtract the first column from the second and divide by $x^2_2-x^2_1$, then in the limit $x_i\to 0$, this equals the first partial derivative. Similarly for the $k$th column, subtract the $(k-1)$th order Taylor series expansion of each entry around $x_k^2=0$, which in the limit is a sum of the first $k-1$ columns. Then divide by $\prod_{i=1}^{k-1} (x_k^2-x_i^2)/k!$. By l'H\^opital's rule this limit equals the $(k-1)$th partial derivative. A similar argument for the first $m$ rows lets us write the entries as partial derivatives in $y_i^2$ as well. The method is described in detail in \cite{IzerginCokerKorepin1992}.
Hence in the limit, the partition function equals 
\begin{align*}
&Z_{n,m}(\gamma, \dots, \gamma, 0, \dots, 0)\\ %=\hat Z_{n,m}(0, \dots, 0, 0, \dots, 0)\\
&=\lim_{x_i, y_j \to 0} 2^{n-m}(n-m)!\ \varphi^{n-m}e^{\left(\binom{m}{2}-nm\right)\gamma}f(\gamma)^m \prod_{i=1}^{m} \left[e^{y_i+\zeta} f(y_i-\zeta)\right]\nonumber\\
&\hphantom{=\ }\times
\frac{\prod_{i=1}^n f(2x_i+2\gamma)\prod_{i=1}^n\prod_{j=1}^{m} f(y_j\pm(x_i+\gamma))\prod_{j=1}^n x_j \prod_{i=1}^m y_i}{\prod_{1\leq i<j\leq m}f(y_j\pm y_i)\prod_{1\leq i<j\leq n}f(x_i\pm x_j)\prod_{i=1}^n f(2x_i)\prod_{j=1}^m f(2y_j)}\nonumber\\
&\hphantom{=\ }\times\dfrac{\prod_{1\leq i < j\leq n} (x_j^2-x_i^2)\prod_{1\leq i < j\leq m} (y_j^2-y_i^2)}{\prod_{i=1}^n (i-1)!\prod_{j=1}^m (j-1)!}\nonumber\\
&\hphantom{=\ }\times\det
\begin{pmatrix}[2]
\hat g(x_1^2, y_1^2) & \frac{\partial}{\partial x_2^2}\hat g(x_2^2, y_1^2) & \cdots & \frac{\partial^{n-1}}{\partial x_n^{2(n-1)}}\hat g(x_n^2, y_1^2)\\
 \frac{\partial}{\partial y_2^2}\hat g(x_1^2, y_2^2) & \frac{\partial^2}{\partial x_2^2 y_2^2}\hat g(x_2^2, y_2^2) & \cdots & \frac{\partial^n}{\partial x_n^{2(n-1)}y_2^2}\hat g(x_n^2, y_2^2)\\
\vdots & \vdots & \ddots & \vdots\\
 \frac{\partial^{m-1}}{\partial y_m^{2(m-1)}} \hat g(x_1^2, y_m^2) & \frac{\partial^m}{\partial x_2^2 y_m^{2(m-1)}}\hat g(x_2^2, y_m^2) & \cdots & \frac{\partial^{n+m-2}}{\partial x_n^{2(n-1)}y_m^{2(m-1)}}\hat g(x_n^2, y_m^2)\\
\hat f((2(n-m)x_1)^2) & \frac{\partial}{\partial x_2^2}\hat f((2(n-m)x_2)^2) & \cdots & \frac{\partial^{n-1}}{\partial x_n^{2(n-1)}}\hat f((2(n-m)x_n)^2)\\
\vdots & \vdots & \ddots & \vdots\\
\hat f((2 x_1)^2) & \frac{\partial}{\partial x_2^2}\hat f((2 x_2)^2) & \cdots & \frac{\partial^{n-1}}{\partial x_n^{2(n-1)}}\hat f((2 x_n)^2)
\end{pmatrix}.
\end{align*}
To cancel the zeroes in the denominators, the prefactors can be simplified by
\begingroup
\allowdisplaybreaks[0]
\begin{align*}
&\lim_{x_i,y_j\to 0}\frac{\prod_{i=1}^n x_i \prod_{j=1}^m y_j \prod_{1\leq i < j\leq n} (x_j^2-x_i^2)\prod_{1\leq i < j\leq m} (y_j^2-y_i^2)}{\prod_{i=1}^n f(2x_i)\prod_{j=1}^m f(2y_j)\prod_{1\leq i<j\leq n}f(x_i\pm x_j)\prod_{1\leq i<j\leq m}f(y_j\pm y_i)}\\
%&=\lim_{x_i,y_j\to 0} \frac{1}{\prod_{i=1}^n\frac{f(2x_i)}{ x_i}  \prod_{j=1}^m\frac{f(2y_j)}{ y_j} \prod_{1\leq i < j\leq n}\frac{f(x_i\pm x_j)}{ x_j^2-x_i^2} \prod_{1\leq i < j\leq m}\frac{f(y_j\pm y_i)}{ y_j^2-y_i^2}}\\
%&=\lim_{x_i,y_j\to 0} \left(\frac{1}{\prod_{i=1}^n \left(e^{-2x_i}\frac{(e^{4x_i}-e^0)}{x_i-0}\right)  \prod_{j=1}^m\left(e^{-2y_j}\frac{(e^{4y_j}-e^0)}{y_j-0}\right)}\right.\\
%&\frac{1}{\prod_{1\leq i < j\leq n}\left(-e^{-2x_i}\frac{e^{2(x_i+x_j)}-e^0}{x_i+x_j-0}\frac{e^{2(x_i-x_j)}-e^0}{x_i-x_j-0}\right)}\\
%&\left.\frac{1}{ \prod_{1\leq i < j\leq m}\left(e^{-2y_j}\frac{e^{2(y_j+y_i)}-e^0}{y_j+y_i-0}\frac{e^{2(y_j-y_i)}-e^0}{y_j-y_i-0}\right)}\right)\\
%&=\frac{1}{\prod_{i=1}^n (e^0 \cdot4e^0)  \prod_{j=1}^m (e^0 \cdot 4e^0) \prod_{1\leq i < j\leq n}(-e^0\cdot 2e^0\cdot 2e^0)\prod_{1\leq i < j\leq m}(e^0\cdot 2e^0\cdot 2e^0)}\\
%&=\frac{1}{\prod_{i=1}^n 4 \prod_{j=1}^m 4 \prod_{1\leq i < j\leq n}(-4)\prod_{1\leq i < j\leq m}4}\\
%&=\frac{(-1)^{\binom{n}{2}}}{4^{n+m+\binom{n}{2}+\binom{m}{2}}}\\
%&=\frac{(-1)^{\binom{n}{2}}}{2^{2n+2m+2\binom{n}{2}+2\binom{m}{2}}}\\
&=\frac{(-1)^{\binom{n}{2}}}{2^{n^2+n+m^2+m}}.
\end{align*}
\endgroup
Moreover
$$f(2l x)%=-2i\sin(2ilx)=-2i(2ilx-\frac{(2ilx)^3}{3!}+\frac{(2ilx)^5}{5!}-\cdots)
=4lx\left(1+\frac{(2 lx)^2}{3!}+\frac{(2 lx)^4}{5!}+\cdots\right),$$
so $$\hat f(x)=2\left(1+\frac{x}{3!}+\frac{x^2}{5!}+\cdots\right),$$
and 
\begin{align*}
\frac{\partial^k}{\partial x^{2k}}\hat f((2 lx)^2)\big|_{x=0}
=
\frac{2(2 l)^{2k}k!}{(2k+1)!}.
\end{align*}
For $\gamma=4\pi \i /3$, 
$$g(x,y)%\coloneqq\frac{f(2x)f(2y)}{f(y\pm (x+\gamma))f(y\pm(x+2\gamma))}
%=\frac{f(2x)f(2y)f(y\pm x)}{f(y\pm x)f(y\pm (x+\gamma))f(y\pm(x+2\gamma))}
%=\frac{f(2x)f(2y)f(y\pm x)}{f(3y\pm 3x)}
=\frac{f(x-y)}{f(3x-3y)}-\frac{f(x+y)}{f(3x+3y)}.$$ 
Following \cite{ColomoPronko2006}, we use Fourier transforms to see that 
$$\frac{f(x)}{f(3x)}=\frac{1}{2\sqrt 3}\int_{-\infty}^{\infty} e^{-\i xt}\frac{\sinh(\pi t/6)}{\sinh(\pi t /2)}dt.$$
From this, it follows that 
\begin{align*}
\frac{f(x-y)}{f(3x-3y)}-\frac{f(x+y)}{f(3x+3y)}=%\\
%\frac{1}{2\sqrt 3}\int_{-\infty}^\infty \left(e^{-it(x-y)}-e^{-it(x+y)} \right)\frac{\sinh(\pi t/6)}{\sinh(\pi t /2)}dt=\nonumber\\
%\frac{1}{2\sqrt 3}\int_0^\infty (-e^{i(x+y)t}+e^{i(x-y)t}+e^{-i(x-y)t}-e^{-i(x+y)t}) \frac{\sinh(\pi t/6)}{\sinh(\pi t /2)}dt=\nonumber\\
%\frac{2}{4 i^2\sqrt 3}\int_0^\infty (e^{ixt}-e^{-ixt})(e^{iyt}-e^{-iyt}) \frac{\sinh(\pi t/6)}{\sinh(\pi t /2)}dt=\nonumber\\
\frac{2}{\sqrt 3}\int_0^\infty \sin(xt) \sin(yt)\frac{\sinh(\pi t/6)}{\sinh(\pi t /2)}dt.
\end{align*}
By Taylor expansion and Lebesgue's dominated convergence theorem we get
\begin{align*}
g(x,y)%=\frac{2}{\sqrt 3}\int_0^\infty \left(xt-\frac{(xt)^3}{3!}+\frac{(xt)^5}{5!}-\cdots\right)\left(yt-\frac{(yt)^3}{3!}+\frac{(yt)^5}{5!}-\cdots\right)\frac{\sinh(\pi t/6)}{\sinh(\pi t /2)}dt\nonumber\\
%=\frac{2}{\sqrt 3}\int_0^\infty \sum_{k,l=0}^\infty \frac{(-1)^{k+l}(xt)^{2k+1}(yt)^{2l+1}}{(2k+1)!(2l+1)!} \frac{\sinh(\pi t/6)}{\sinh(\pi t /2)}dt\nonumber\\
%=\frac{2}{\sqrt 3}\sum_{k,l=0}^\infty\left(\frac{(-1)^{k+l}x^{2k+1}y^{2l+1}}{(2k+1)!(2l+1)!}\int_0^\infty t^{2k+2l+2}\frac{\sinh(\pi t/6)}{\sinh(\pi t /2)}dt\right)\nonumber\\
=xy\frac{2}{\sqrt 3}\sum_{k,l=0}^\infty\left(\frac{(-1)^{k+l}x^{2k}y^{2l}}{(2k+1)!(2l+1)!}\int_0^\infty t^{2k+2l+2}\frac{\sinh(\pi t/6)}{\sinh(\pi t /2)}dt\right),
\end{align*}
for $|x|,|y|<\pi/3$, 
%(
%Let $$f_n(t)=\sum_{l=0}^n \underbrace{\frac{(-1)^l(xt)^{2l+1}}{(2l+1)!}\sin(yt)\frac{\sinh(\pi t/6))}{\sinh(\pi t/2)}}_{h_l(t)}.$$
%Let $f(t)=\sum_{l=0}^\infty h_l(t) =\lim_{n\to\infty}\sum_{l=0}^n h_l(t)$.
%Now $$\int f(t)dt=\int \lim_{n\to\infty}\sum_{l=0}^n h_l(t) dt.$$
%It is always true that $$\lim_{n\to\infty}\sum_{l=0}^n \int h_l(t) dt=\lim_{n\to\infty} \int \underbrace{\sum_{l=0}^n h_l(t)}_{f_n(t)} dt.$$
%If $\left|f_n(t)\right|\leq g(t)$ for all $t$ and $n$ where 
%$$\int \left|g(t)\right| dt<\infty,$$
%then $$\lim_{n\to\infty} \int f_n(t) dt=\int \underbrace{\lim_{n\to\infty} f_n(t)}_{f(t)} dt$$
%according to Lebesgue dominated convergence theorem.
%The first holds because 
%\begin{align*}
%\left|f_n(t)\right|&=\left|\sum_{l=0}^n \frac{(-1)^l(xt)^{2l+1}}{(2l+1)!}\sin(yt)\frac{\sinh(\pi t/6))}{\sinh(\pi t/2)}\right|\\
%&\leq\sum_{l=0}^n \left|\frac{(xt)^{2l+1}}{(2l+1)!}\sin(yt)\frac{\sinh(\pi t/6))}{\sinh(\pi t/2)}\right|\\
%&\leq\sum_{l=0}^n \left|\frac{(xt)^{2l+1}}{(2l+1)!}\frac{\sinh(\pi t/6))}{\sinh(\pi t/2)}\right|\\
%&\leq\sum_{l=0}^\infty \left|\frac{(xt)^{2l+1}}{(2l+1)!}\frac{\sinh(\pi t/6))}{\sinh(\pi t/2)}\right|=g(t).
%\end{align*}
%We have to prove that $$\int_0^\infty \left|g(t)\right| dt<\infty.$$ For $\left| x\right| < \pi/3$, this seems to hold by tests in Mathematica. 
%\begin{align*}
%g(t)=\sum_{l=0}^\infty \left|\frac{(xt)^{2l+1}}{(2l+1)!}\frac{\sinh(\pi t/6)}{\sinh(\pi t/2)}\right| =\frac{\sinh(\pi t/6)}{\sinh(\pi t/2)}\sum_{l=0}^\infty \left|\frac{(xt)^{2l+1}}{(2l+1)!}\right| =\frac{\sinh(\pi t/6)}{\sinh(\pi t/2)}\sinh(\left| x t\right|)
%\end{align*}
%)
so $$\hat g(x^2,y^2) =\frac{2}{\sqrt 3} \sum_{k,l=0}^\infty\left(\frac{(-1)^{k+l}x^{2k}y^{2l}}{(2k+1)!(2l+1)!}\int_0^\infty t^{2k+2l+2}\frac{\sinh(\pi t/6)}{\sinh(\pi t /2)}dt\right).$$
Hence
\begin{align*}
\frac{\partial^{k+l}}{\partial x^{2k}y^{2l}} \hat g(x^2, y^2)\big|_{x=y=0}
=\frac{2(-1)^{k+l}k!l!}{\sqrt 3 (2k+1)!(2l+1)!}\int_0^\infty t^{2k+2l+2}\frac{\sinh(\pi t/6)}{\sinh(\pi t /2)}dt.
\end{align*}
Let $c_k$ denote the moment 
\begin{equation*}
c_k=\int t^k d\mu(t),
\end{equation*}
where 
$$\int f(t) d\mu(t)=\int_0^\infty f(t^2) t^2 \frac{\sinh(\pi t/6)}{\sinh(\pi t /2)}dt.$$
The corresponding inner product is
$$\left\langle p(t), q(t)\right\rangle=\int_0^\infty p(t^2) q(t^2) t^2 \frac{\sinh(\pi t/6)}{\sinh(\pi t /2)}dt,$$
for polynomials $p$ and $q$.

In the limit $x_i, y_j\to 0$, the partition function thus equals
\begin{align*}
&Z_{n,m}(\gamma, \dots, \gamma, 0, \dots, 0)
=
(-1)^{\binom{n}{2}+mn}\prod_{i=1}^n \frac{1}{(2i-1)!}\prod_{j=1}^m \frac{1}{(2j-1)!}\nonumber\\ 
&\quad\times\frac{2^{n-2m-m^2-n^2}}{3^{m/2}}(n-m)!\ \varphi^{n-m}e^{\left(\binom{m}{2}-mn\right)\gamma}f(\gamma)^{m+2mn}f(2\gamma)^n(1-e^{2\zeta})^m \nonumber\\
&\quad\times\det
\begin{pmatrix}[1.5]
c_0 & -c_1 & \cdots & (-1)^{n+1} c_{n-1}\\
-c_1 & c_2 & \cdots & (-1)^{n+2} c_n\\
\vdots & \vdots & \ddots & \vdots\\
(-1)^{m+1} c_{m-1} & (-1)^{m+2}c_m & \cdots & (-1)^{n+m} c_{n+m-2}\\
1 & (2(n-m))^2 & \cdots & (2(n-m))^{2(n-1)}\\
\vdots & \vdots & \ddots & \vdots\\
1 & 2^2 & \cdots & 2^{2(n-1)}
\end{pmatrix}.
\end{align*}
%For odd $n=2j+1$, we have 
%\begin{equation*}
%\frac{n(n-1)}{2}+\left\lfloor \frac{n}{2}\right\rfloor%=\frac{(2j+1)2j}{2}+\left\lfloor \frac{2j+1}{2}\right\rfloor=(2j+1)j+j=2j^2+j+j
%=2j^2+2j,
%\end{equation*}
%which is even, 
%and for even $n=2j$, we have
%\begin{equation*}
%\frac{n(n-1)}{2}+\left\lfloor \frac{n}{2}\right\rfloor%=\frac{2j(2j-1)}{2}+\left\lfloor \frac{2j}{2}\right\rfloor=j(2j-1)+j=2j^2-j+j
%=2j^2,
%\end{equation*}
%which is also even. Hence $(-1)^{\frac{n(n-1)}{2}+\left\lfloor \frac{n}{2}\right\rfloor}=1$. 
We can write the entries $c_k$ in terms of inner products. Since $(-1)^{\frac{n(n-1)}{2}+\left\lfloor \frac{n}{2}\right\rfloor}=1$ for all $n$, the partition function above becomes
\begin{align*}
&Z_{n,m}(\gamma, \dots, \gamma, 0, \dots, 0)=(-1)^{mn-m(m-1)/2}\prod_{i=1}^n \frac{1}{(2i-1)!}\prod_{j=1}^m \frac{1}{(2j-1)!}\nonumber\\
&\quad\times \frac{2^{n-2m-m^2-n^2}}{3^{m/2}} (n-m)!\ \varphi^{n-m}e^{\left(\binom{m}{2}-nm\right)\gamma}f(\gamma)^{2mn+m}f(2\gamma)^n(1-e^{2\zeta})^m\nonumber\\
&\quad\times\det
\begin{pmatrix}[1.5]
\left\langle 1,1\right\rangle & \left\langle t,1\right\rangle & \cdots & \left\langle t^{n-1},1\right\rangle\\
\left\langle 1,t\right\rangle & \left\langle t,t\right\rangle & \cdots & \left\langle t^{n-1},t\right\rangle\\
\vdots & \vdots & \ddots & \vdots\\
\left\langle 1, t^{m-1}\right\rangle & \left\langle t, t^{m-1}\right\rangle & \cdots & \left\langle t^{n-1}, t^{m-1}\right\rangle\\
1 & -(2(n-m))^2 & \cdots & (-1)^{n-1}(2(n-m))^{2(n-1)}\\
\vdots & \vdots & \ddots & \vdots\\
%1 & (4\i)^2 & \cdots & (4\i)^{2(n-1)}\\
1 & -2^2 & \cdots & (-1)^{n-1}2^{2(n-1)}
\end{pmatrix}.
\end{align*}
%\textcolor[rgb]{1,0,0}{This is the determinant (3.6) in \cite{ColomoPronko2006}. }
Similar determinants show up e.g. in \cite{ColomoPronko2006}. By row and column operations, the determinant can be rewritten, and the partition function is
\begin{align*}
&Z_{n,m}(\gamma, \dots, \gamma, 0, \dots, 0) =(-1)^{mn-m(m-1)/2}\prod_{j=1}^n \frac{1}{(2j-1)!}\prod_{j=1}^m \frac{1}{(2j-1)!}\\
&\quad\times \frac{2^{n-2m-m^2-n^2}}{3^{m/2}} (n-m)!\ \varphi^{n-m}e^{\left(\binom{m}{2}-nm\right)\gamma}f(\gamma)^{2mn+m}f(2\gamma)^n(1-e^{2\zeta})^m\\
&\quad \times\det
\begin{pmatrix}[1.5]
\left\langle p_0(t), q_0(t)\right\rangle & \left\langle p_1(t), q_0(t)\right\rangle & \cdots & \left\langle p_{n-1}(t),q_0(t)\right\rangle\\
\left\langle p_0(t), q_1(t)\right\rangle & \left\langle p_1(t), q_1(t)\right\rangle & \cdots & \left\langle p_{n-1}(t),q_1(t)\right\rangle\\
\vdots & \vdots & \ddots & \vdots\\
\left\langle p_0(t), q_{m-1}(t)\right\rangle & \left\langle p_1(t), q_{m-1}(t)\right\rangle & \cdots & \left\langle p_{n-1}(t), q_{m-1}(t)\right\rangle\\
p_0(-(2(n-m))^2) & p_1(-(2(n-m))^2) & \cdots & p_{n-1}(-(2(n-m))^2)\\
\vdots & \vdots & \ddots & \vdots\\
p_0(-4^2) & p_1(-4^2) & \cdots & p_{n-1}(-4^2)\\
p_0(-2^2) & p_1(-2^2) & \cdots & p_{n-1}(-2^2)
\end{pmatrix},
\end{align*}
where $p_k(t)$ and $q_k(t)$ are arbitrary monic polynomials of degree $k$. We can choose $\{p_i=q_i\}_{i=1}^n$ orthogonal, i.e.
\begin{equation}
\label{orthogonality}
\left\langle p_k(t), p_l(t)\right\rangle=\int p_k(t) p_l(t) d\mu(t)=h_k \delta_{kl},
\end{equation}
for some $h_k$. Then most entries in the upper part of the matrix become $0$, and we can simplify the expression further. We also rearrange the rows.
Then the partition function is
\begin{align}
\label{partfcn61}
&Z_{n,m}(\gamma, \dots, \gamma, 0, \dots, 0)\nonumber\\
&=\frac{2^{n-2m-m^2-n^2}}{3^{m/2}} (n-m)!\ \varphi^{n-m}e^{\left(\binom{m}{2}-nm\right)\gamma}f(\gamma)^{2mn+m}f(2\gamma)^n(1-e^{2\zeta})^m\nonumber\\
&\hphantom{=\ }\times (-1)^{\binom{n}{2}+m}\prod_{j=1}^n \frac{1}{(2j-1)!}\prod_{j=1}^m \frac{1}{(2j-1)!} \prod_{i=0}^{m-1} \left\langle p_i(t), p_i(t)\right\rangle\nonumber\\
&\hphantom{=\ }\times \det
\begin{pmatrix}[1.5]
p_m(-2^2) & p_{m+1}(-2^2) &\cdots & p_{n-1}(-2^2)\\
p_m(-4^2) & p_{m+1}(-4^2) &\cdots & p_{n-1}(-4^2)\\
\vdots & \vdots &\ddots & \vdots\\
p_m(-(2(n-m))^2) & p_{m+1}(-(2(n-m))^2) & \cdots & p_{n-1}(-(2(n-m))^2)
\end{pmatrix}.
\end{align}

\subsection{Wilson polynomials}
Colomo and Pronko \Cite{ColomoPronko2006} found a way to write the determinant formula of the partition function of the 6V model with DWBC in terms of orthogonal polynomials. In this section we show that we can rewrite the determinant formula \ref{determinantformula} in terms of orthogonal polynomials as well.  

We can rewrite the weight $\mu(t)$ in terms of the gamma function defined as the analytic continuation of the integral 
$$\Gamma(z)=\int_0^\infty x^{z-1} e^{-x} dx,$$ 
which is defined only for $\Re(z)>0$. It can be shown that the gamma function satisfies the following identities:
$$\left|\Gamma(1+\i x)\right|^2=\frac{\pi x}{\sinh(\pi x)}$$
and
\begin{equation}
\label{additionformulagamma}
\prod_{k=0}^{l-1}\Gamma\left(x+\frac{k}{l}\right)=(2\pi)^{\frac{l-1}{2}} l^{\frac{1}{2}-lx}\Gamma(lx).
\end{equation}
Hence we have
\begin{align}
\label{orthoweight}
\mu(t)= t^2 \frac{\sinh(\pi t/6)}{\sinh(\pi t /2)}=\frac{3^2}{2^2\pi^3}\left|\frac{\Gamma(\i t/6+1/3)\Gamma(\i t/6+1/2)\Gamma(\i t/6+2/3)\Gamma(\i t/6+1)}{\Gamma(\i t/3)}\right|^2.
\end{align}
This is the orthogonality weight for a known family of polynomials, namely, the Wilson polynomials. These are defined in terms of generalized hypergeometric series as (see e.g. \cite{KoekoekSwarttouw1998})
\begin{equation*}
\frac{W_k(x^2; a, b,c,d)}{(a+b)_k (a+c)_k (a+d)_k}= {}_4 F_3 \left( \begin{matrix} -k, k+a+b+c+d-1, a+\i x, a-\i x\\ a+b, a+c, a+d \end{matrix} \bigg| 1 \right),
\end{equation*}
where the orthogonality condition reads
\begin{align}
\label{wilsonpolynomial}
&\frac{1}{2\pi} \int_0^\infty W_k(x^2; a, b,c,d) W_l(x^2; a, b,c,d)\left|\frac{\Gamma(\i x+a)\Gamma(\i x+b)\Gamma(\i x+c)\Gamma(\i x+d)}{\Gamma(2\i x)} \right|^2 dx\nonumber\\
&=\frac{\Gamma(k+a+b)\Gamma(k+a+c)\Gamma(k+a+d)\Gamma(k+b+c)\Gamma(k+b+d)\Gamma(k+c+d)}{\Gamma(2k+a+b+c+d)}\nonumber\\
&\hphantom{=\ }\times (k+a+b+c+d-1)_k k! \delta_{kl},
\end{align}
and where $(a)_k=a(a+1)\cdots(a+k-1)$, for $k\geq 1$, and $(a)_0=1$ is the rising factorial. 
Comparing this to \eqref{orthogonality} and \eqref{orthoweight}, we choose the parameters $a=1/3$, $b=1/2$, $c=2/3$, $d=1$ and $x=t/6$. Then the Wilson polynomials are
\begin{align}
\label{wilsonpoly}
W_k\left(\left(\frac{t}{6}\right)^2; \frac{1}{3}, \frac{1}{2}, \frac{2}{3}, 1\right)
%&= (5/6)_k (4/3)_k k! {}_4 F_3 \left( \begin{matrix} -k, k+3/2, 1/3+it/6, 1/3-it/6\\ 5/6, 1, 4/3 \end{matrix} \bigg| 1 \right)\\
%&= (5/6)_k (4/3)_k k!  \sum_{j=0}^\infty \frac{(-k)_j (3/2+k)_j(1/3+it/6)_j(1/3-it/6)_j}{(5/6)_j (4/3)_j (j!)^2}\\
= (5/6)_k (4/3)_k k!  \sum_{j=0}^k \frac{(-k)_j (3/2+k)_j(1/3+\i t/6)_j(1/3-\i t/6)_j}{(5/6)_j (4/3)_j (j!)^2}.
\end{align}
These are polynomials in $t^2$ of degree $k$, and leading coefficient
\begin{align}
\label{kappak}
\kappa_k  %t^{2k}%&=(5/6)_k (4/3)_k k!  \frac{(-k)_k (3/2+k)_k(it/6)^k(-it/6)^k}{(5/6)_k (4/3)_k (k!)^2}\\
%&=(-1)^k (5/6)_k (4/3)_k k!  \frac{k! (3/2+k)_k %(t^2)^k
%}{6^{2k} (5/6)_k (4/3)_k (k!)^2}\\ &
=  \frac{(-1)^k (3/2+k)_k %t^{2k}
}{6^{2k}}. 
\end{align}
The polynomials can hence be written $W_k\left(\left(\frac{t}{6}\right)^2; \frac{1}{3}, \frac{1}{2}, \frac{2}{3}, 1\right)=\kappa_k p_k(t^2)$, where $p_k(t)$ is a monic polynomial of degree $k$ in $t$.
The right hand side of \eqref{wilsonpolynomial} with the above choices of parameters is 
\begin{align}
\label{rhs}
&\frac{\Gamma(k+5/6)\Gamma(k+1)\Gamma(k+4/3)\Gamma(k+7/6)\Gamma(k+3/2)\Gamma(k+5/3)}{\Gamma(2k+5/2)} (k+3/2)_k k!\delta_{kl}\nonumber\\ &=
%\frac{(2\pi)^{5/2} (k+3/2)_k k! }{6^{6k+9/2}} \frac{\Gamma(6k+5)}{\Gamma(2k+5/2)}\delta_{kl}\\=
%\frac{2^{4k+3}(2\pi)^{5/2} (k+3/2)_k k! }{6^{6k+9/2}} \frac{\Gamma(6k+5)\Gamma(2k+2)}{\sqrt{\pi}\Gamma(4k+4)}\delta_{kl}\\=
%\frac{2^{4k+11/2} \pi^2 (k+3/2)_k k! }{6^{6k+9/2}} \frac{\Gamma(6k+5)\Gamma(2k+2)}{\Gamma(4k+4)}\delta_{kl}\\=
%\frac{2^{4k+11/2} \pi^2 (k+3/2)_k k! (6k+4)!(2k+1)!}{6^{6k+9/2}(4k+3)!}\delta_{kl}=\\
\frac{\pi^2 (k+3/2)_k k! (6k+4)!(2k+1)!}{2^{2k-1}3^{6k+9/2}(4k+3)!}\delta_{kl},
\end{align}
%(OK)
where we have used \eqref{additionformulagamma} and the fact that $\Gamma(j+1)=j!$ if $j$ is a non-negative integer.

%We thus have
%\begin{align*}
%&\frac{1}{12\pi} \int_0^\infty \kappa_k p_k(t^2) \kappa_l p_l(t^2)  \\
%&\cdot \left|\frac{\Gamma(\i t/6+1/3)\Gamma(\i t/6+1/2)\Gamma(\i t/6+2/3)\Gamma(\i t/6+1)}{\Gamma(\i t/3)}\right|^2 dt\\
%&=\frac{\pi^2 (k+3/2)_k k!(6k+4)!(2k+1)!}{2^{2k-1}3^{6k+9/2}(4k+3)!} \delta_{kl}.
%\end{align*}
We insert \eqref{wilsonpolynomial}-\eqref{rhs} into \eqref{orthogonality}, and get
\begin{align*}
&\left\langle p_k(t), p_l(t) \right\rangle 
%=\frac{3^2}{2^2\pi^3}\int_0^\infty p_k(t^2) p_l(t^2) \\
%&\cdot\left|\frac{\Gamma(\i t/6+1/3)\Gamma(\i t/6+1/2)\Gamma(\i t/6+2/3)\Gamma(\i t/6+1)}{\Gamma(\i t/3)}\right|^2 dt\\
%&=\frac{3^2}{2^2\pi^3}\int_0^\infty \frac{W_k\left(\left(\frac{t}{6}\right)^2; \frac{1}{3}, \frac{1}{2}, \frac{2}{3}, 1\right)W_l\left(\left(\frac{t}{6}\right)^2; \frac{1}{3}, \frac{1}{2}, \frac{2}{3}, 1\right)}{\kappa_k \kappa_l} \\
%&\cdot\left|\frac{\Gamma(\i t/6+1/3)\Gamma(\i t/6+1/2)\Gamma(\i t/6+2/3)\Gamma(\i t/6+1)}{\Gamma(\i t/3)}\right|^2 dt\\
%&=\frac{3^2}{2^2\pi^3\kappa_k\kappa_l}\int_0^\infty W_k\left(\left(\frac{t}{6}\right)^2; \frac{1}{3}, \frac{1}{2}, \frac{2}{3}, 1\right)W_l\left(\left(\frac{t}{6}\right)^2; \frac{1}{3}, \frac{1}{2}, \frac{2}{3}, 1\right) \\
%&\cdot\left|\frac{\Gamma(\i t/6+1/3)\Gamma(\i t/6+1/2)\Gamma(\i t/6+2/3)\Gamma(\i t/6+1)}{\Gamma(\i t/3)}\right|^2 dt\\
%&=\frac{3^2\cdot 12\pi}{2^2\pi^3\kappa_k\kappa_l}\frac{1}{12\pi}\int_0^\infty W_k\left(\left(\frac{t}{6}\right)^2; \frac{1}{3}, \frac{1}{2}, \frac{2}{3}, 1\right)W_l\left(\left(\frac{t}{6}\right)^2; \frac{1}{3}, \frac{1}{2}, \frac{2}{3}, 1\right) \\
%&\cdot\left|\frac{\Gamma(\i t/6+1/3)\Gamma(\i t/6+1/2)\Gamma(\i t/6+2/3)\Gamma(\i t/6+1)}{\Gamma(\i t/3)}\right|^2 dt\\
%&=\frac{3^2\cdot 12\pi}{2^2\pi^3\kappa_k\kappa_l}\frac{\pi^2 (k+3/2)_k k! (6k+4)!(2k+1)!}{2^{2k-1}3^{6k+9/2}(4k+3)!}\delta_{kl}\\
%=\frac{(k+3/2)_k k! (6k+4)!(2k+1)!}{2^{2k-1}3^{6k+3/2}\kappa_k\kappa_l(4k+3)!}\delta_{kl}\\
=\frac{2^{2k+1} k! (6k+4)!(2k+1)!}{3^{2k+3/2} (3/2+k)_k(4k+3)!}\delta_{kl}.
\end{align*}
Hence 
\begin{align*}
\prod_{j=0}^{m-1} \left\langle p_j(t), p_j(t)\right\rangle%&=\prod_{j=0}^{m-1} \frac{2^{2j+1}(6j+4)!(2j+1)! j!}{3^{2j+3/2} (3/2+j)_j (4j+3)!}\\
%&=\frac{2^{m+m(m-1)}}{3^{m(m-1)+3m/2} }\prod_{j=0}^{m-1} \frac{(6j+4)!(2j+1)!j!}{ (3/2+j)_j  (4j+3)!}\\
%&=\frac{2^{m^2}}{3^{m^2+m/2} }\prod_{j=0}^{m-1} \frac{(6j+4)!(2j+1)!j!}{(3/2+j)_j  (4j+3)!}\\
%&=\frac{2^{2m^2-m}}{3^{m^2+m/2} }\prod_{j=0}^{m-1} \frac{(6j+4)!((2j+1)!)^2(2j)!}{ (4j+1)! (4j+3)!}
&=\frac{2^{2m^2-m}}{3^{m^2+m/2} }\prod_{j=1}^{m} \frac{(6j-2)!((2j-1)!)^2(2j-2)!}{ (4j-3)! (4j-1)!},
\end{align*}
where we used that 
\begin{equation}
\label{factorialingamma}
(x)_k=\frac{\Gamma(x+k)}{\Gamma(x)}.
\end{equation}
%We can conclude that
%\begin{align*}
%&h_k \delta_{kl}=\frac{ (k+3/2)_k k!(6k+4)!(2k+1)!}{2^{2k-1}3^{6k+3/2}\kappa_k\kappa_l (4k+3)!} \delta_{kl}.
%\end{align*}
%Inserting $\kappa_k=\kappa_l$ for $k=l$ yields
%\begin{align*}
%h_k%=\frac{2^{4k+3} (k+3/2)_k k!(6k+4)!(2k+1)!}{2^{6k+2}3^{6k+3/2}\kappa_k^2 (4k+3)!}\\
%%=\frac{(k+3/2)_k k!(6k+4)!(2k+1)!}{2^{2k-1}3^{6k+3/2}\left(\frac{(-1)^k (3/2+k)_k}{6^{2k}}\right)^2 (4k+3)!}\\
%=\frac{2^{2k+1} k! (6k+4)!(2k+1)!}{3^{2k+3/2} (3/2+k)_k  (4k+3)!},
%\end{align*}
%and 
%\begin{equation*}
%p_k(t^2)%=\frac{(5/6)_k (4/3)_k k!}{\kappa_k}  \sum_{j=0}^k \frac{(-k)_j (3/2+k)_j(1/3+it/6)_j(1/3-it/6)_j}{(5/6)_j (4/3)_j (j!)^2}
%=\frac{6^{2k} (5/6)_k (4/3)_k k!}{(-1)^k (3/2+k)_k}  \sum_{j=0}^k \frac{(-k)_j (3/2+k)_j(1/3+\i t/6)_j(1/3-\i t/6)_j}{(5/6)_j (4/3)_j (j!)^2}.
%\end{equation*}
Insert $p_k(t^2)=W_k\left(\left(\frac{t}{6}\right)^2; \frac{1}{3}, \frac{1}{2}, \frac{2}{3}, 1\right)/\kappa_k$ into the partition function \eqref{partfcn61}, which then reads
\begin{align*}
&Z_{n,m}(\gamma, \dots, \gamma, 0, \dots, 0)\\
%&=x^{n-m}e^{\left(\binom{m}{2}-nm\right)\gamma}f^{2mn+m}(\gamma)f^n(2\gamma)(1-e^{2\zeta})^m\frac{2^{n-2m-m^2-n^2}}{3^{m/2}} \frac{2^{2m^2-m}}{3^{m^2+m/2} }\ (n-m)!\nonumber\\
%&\cdot (-1)^{\binom{n}{2}+m}\prod_{i=1}^n \frac{1}{(2i-1)!}\prod_{j=1}^m \frac{1}{(2j-1)!}\prod_{j=1}^m\frac{(6j-2)!((2j-1)!)^2(2j-2)!}{ (4j-3)! (4j-1)!}\nonumber\\
%&\cdot \det
%\begin{pmatrix}[1.5]
%p_m((2i(n-m))^2) & p_{m+1}((2i(n-m))^2) & \cdots & p_{n-1}((2i(n-m))^2)\\
%\vdots & \vdots &\ddots & \vdots\\
% p_m((4i)^2) & p_{m+1}((4i)^2) &\cdots & p_{n-1}((4i)^2)\\
% p_m((2i)^2) & p_{m+1}((2i)^2) &\cdots & p_{n-1}((2i)^2)
%\end{pmatrix}\\
%&=
%\frac{2^{n-n^2-3m+m^2}}{3^{m^2+m} } (n-m)!\ x^{n-m}e^{\left(\binom{m}{2}-nm\right)\gamma}f^{2mn+m}(\gamma)f^n(2\gamma)(1-e^{2\zeta})^m\nonumber\\
%&\cdot (-1)^{\binom{n}{2}+m}\prod_{i=m+1}^n \frac{1}{(2i-1)!}\prod_{j=1}^m \frac{(6j-2)!(2j-2)!}{(4j-3)! (4j-1)!}\det_{1\leq l,j\leq n-m}(p_{m+j-1}((2l\i)^2))\nonumber\\
&=
\frac{2^{n-n^2-3m+m^2}}{3^{m^2+m} } (n-m)!\ \varphi^{n-m}e^{\left(\binom{m}{2}-nm\right)\gamma}f(\gamma)^{2mn+m}f(2\gamma)^n(1-e^{2\zeta})^m\nonumber\\
&\hphantom{=\ }\times (-1)^{\binom{n}{2}+m}\prod_{i=m+1}^n \frac{1}{(2i-1)!}\prod_{j=1}^m \frac{(6j-2)!(2j-2)!}{(4j-3)! (4j-1)!}\nonumber\\
&\hphantom{=\ }\times\det_{1\leq l,j\leq n-m}\left(W_{m+j-1}\left(-\frac{l^2}{9}; \frac{1}{3}, \frac{1}{2}, \frac{2}{3}, 1\right)\bigg/\kappa_{m+j-1}\right).
\end{align*}
We factor out $\kappa_{m+j-1}$ from each column of the determinant and use \eqref{factorialingamma}, and get
\begin{align}
\label{lastpartfcn}
&Z_{n,m}(\gamma, \dots, \gamma, 0, \dots, 0)\nonumber\\
%&=
%\frac{2^{n-n^2-3m+m^2}}{3^{m^2+m} } (n-m)!\ x^{n-m}e^{\left(\binom{m}{2}-nm\right)\gamma}f^{2mn+m}(\gamma)f^n(2\gamma)(1-e^{2\zeta})^m\nonumber\\
%&\cdot (-1)^{\binom{n}{2}+m}\prod_{i=m+1}^n \frac{1}{(2i-1)!}\prod_{j=1}^m \frac{(6j-2)!(2j-2)!}{(4j-3)! (4j-1)!}\prod_{i=m}^{n-1}\frac{1}{\kappa_i}\nonumber\\
%&\cdot\det_{1\leq l,j\leq n-m}\left(W_{m+j-1}\left(\left(\frac{2l\i}{6}\right)^2; \frac{1}{3}, \frac{1}{2}, \frac{2}{3}, 1\right)\right)\\
%&=
%\frac{2^{n-n^2-3m+m^2}}{3^{m^2+m} } (n-m)!\ x^{n-m}e^{\left(\binom{m}{2}-nm\right)\gamma}f^{2mn+m}(\gamma)f^n(2\gamma)(1-e^{2\zeta})^m\nonumber\\
%&\cdot (-1)^{\binom{n}{2}+m}\prod_{i=m+1}^n \frac{1}{(2i-1)!}\prod_{j=1}^m \frac{(6j-2)!(2j-2)!}{(4j-3)! (4j-1)!}\prod_{j=m}^{n-1}\frac{6^{2j}}{(-1)^{j} (3/2+j)_j}\nonumber\\
%&\cdot\det_{1\leq l,j\leq n-m}\left(W_{m+j-1}\left(\left(\frac{l\i}{3}\right)^2; \frac{1}{3}, \frac{1}{2}, \frac{2}{3}, 1\right)\right)\nonumber\\
%&=
%\frac{2^{n-n^2-3m+m^2}}{3^{m^2+m} } (n-m)!\ x^{n-m}e^{\left(\binom{m}{2}-nm\right)\gamma}f^{2mn+m}(\gamma)f^n(2\gamma)(1-e^{2\zeta})^m\nonumber\\
%&\quad\cdot (-1)^{\binom{m+1}{2}}\prod_{i=m+1}^n \frac{1}{(2i-1)!}\prod_{j=1}^m \frac{(6j-2)!(2j-2)!}{(4j-3)! (4j-1)!}\prod_{j=m}^{n-1}\frac{6^{2j}}{(3/2+j)_j}\nonumber\\
%&\quad\cdot\det_{1\leq l,j\leq n-m}\left(W_{m+j-1}\left(\left(\frac{l\i}{3}\right)^2; \frac{1}{3}, \frac{1}{2}, \frac{2}{3}, 1\right)\right)\\
&=
\frac{2^{n^2-n-m^2-m}}{3^{2m^2-n^2+n} } (n-m)!\ \varphi^{n-m}e^{\left(\binom{m}{2}-nm\right)\gamma}f(\gamma)^{2mn+m}f(2\gamma)^n(1-e^{2\zeta})^m\nonumber\\
&\hphantom{=\ }\times (-1)^{\binom{m+1}{2}}\prod_{j=1}^n \frac{(2j-2)!}{(4j-3)!} \prod_{j=1}^m \frac{(6j-2)!}{ (4j-1)!}\prod_{j=m+1}^n \frac{1}{(j-1)!}\nonumber\\
&\hphantom{=\ }\times\det_{1\leq l,j\leq n-m}\left(W_{m+j-1}\left(-\frac{l^2}{9}; \frac{1}{3}, \frac{1}{2}, \frac{2}{3}, 1\right)\right).
\end{align}
%The above becomes
%\begin{align}
%&Z_{n,m}(\gamma, \dots, \gamma, 0, \dots, 0)\nonumber\\
%&= (-1)^{\binom{n+1}{2}+\binom{n-m}{2}}\frac{2^{n-n^2-3m+m^2}}{3^{m^2+m-mn}} (n-m)!\ x^{n-m}e^{\left(\binom{m}{2}-nm\right)\gamma}f^{m+n}(\gamma)(1-e^{2\zeta})^m\nonumber\\
%&\quad \cdot \prod_{j=m+1}^n \frac{1}{(2j-1)!}\prod_{j=1}^m \frac{(6j-2)!(2j-2)!}{(4j-3)! (4j-1)!} \det_{1\leq l,j\leq n-m}(P_{m+j-1}((2l\i)^2))
%\end{align}
%where
%\begin{equation}
%\label{formulaforpk}
%P_k(t^2)
%=\frac{6^{2k} (5/6)_k (4/3)_k k!}{(3/2+k)_k}  \sum_{j=0}^k \frac{(-k)_j (3/2+k)_j(1/3+\i t/6)_j(1/3-\i t/6)_j}{(5/6)_j (4/3)_j (j!)^2}.
%\end{equation}

\subsection{A formula for the number of states}
\label{subsec:final}
%In this section we want to rewrite $(1-e^{2\zeta})^m$. Assume
%\begin{align*}
%1-e^{2\zeta}=Q_1(e^{2\zeta}-e^{-2\gamma})+Q_2\left(\frac{e^{2\zeta}-e^{2\gamma}}{e^{2\gamma}}\right).
%\end{align*}
%Letting $\zeta=\gamma$ yields
%%\begin{align*}
%%1-e^{2\gamma}=Q_1(e^{2\gamma}-e^{-2\gamma})+Q_2\left(\frac{e^{2\gammma}-e^{2\gamma}}{e^{2\gamma}}\right)
%%=Q_1(e^{2\gamma}-e^{-2\gamma})
%%\end{align*}
%\begin{align*}
%Q_1%=\frac{1-e^{2\gamma}}{(e^{2\gamma}-e^{-2\gamma}}=\frac{e^{\gamma} (e^{-\gamma}-e^{\gamma})}{f(2\gamma)}
%=-\frac{e^{\gamma} f(\gamma)}{f(2\gamma)},
%\end{align*}
%and letting $\zeta=-\gamma$ yields
%%\begin{align*}
%%1-e^{-2\gamma}=Q_1(e^{-2\gamma}-e^{-2\gamma})+Q_2\left(\frac{e^{-2\gamma}-e^{2\gamma}}{e^{2\gamma}}\right)
%%=-Q_2\left(\frac{f(2\gamma)}{e^{2\gamma}}\right).
%%\end{align*}
%\begin{align*}
%Q_2%=-\frac{e^{2\gamma}(1-e^{-2\gamma})}{f(2\gamma)}=-\frac{e^{\gamma}(e^\gamma-e^{-\gamma})}{f(2\gamma)}
%=-\frac{e^{\gamma}f(\gamma)}{f(2\gamma)}.
%\end{align*}
%Letting $\zeta=\gamma$ and then $\zeta=-\gamma$ yields
%\begin{align*}
%Q_1=Q_2
%=-\frac{e^{\gamma}f(\gamma)}{f(2\gamma)}.
%\end{align*}
Next, we rewrite the partition function as a sum over the number of positive turns, $k$, to be able to identify the terms with the terms in the formulas from Section~\ref{sec:countingstates}. 
We write
\begin{align*}
1-e^{2\zeta}=-\frac{e^{\gamma}f(\gamma)}{f(2\gamma)}\left((e^{2\zeta}-e^{-2\gamma})+\frac{e^{2\zeta}-e^{2\gamma}}{e^{2\gamma}}\right), 
\end{align*}
%With $\gamma=4\pi \i /3$, we have $\frac{f(\gamma)}{f(2\gamma)}=-1$. Hence
%\begin{align*}
%1-e^{2\zeta}=e^{\gamma}\left((e^{2\zeta}-e^{-2\gamma})+\frac{e^{2\zeta}-e^{2\gamma}}{e^{2\gamma}}\right),
%\end{align*}
and by using the binomial theorem we get
\begin{align*}
(1-e^{2\zeta})^m%=\left(e^{\gamma}\left((e^{2\zeta}-e^{-2\gamma})+\frac{e^{2\zeta}-e^{2\gamma}}{e^{2\gamma}}\right)\right)^m\\
%=e^{m\gamma}\left((e^{2\zeta}-e^{-2\gamma})+\frac{e^{2\zeta}-e^{2\gamma}}{e^{2\gamma}}\right)^m\\
=\left(-\frac{e^\gamma f(\gamma)}{f(2\gamma)}\right)^m\sum_{k=0}^m \left(\binom{m}{k}(e^{2\zeta}-e^{-2\gamma})^k\left(\frac{e^{2\zeta}-e^{2\gamma}}{e^{2\gamma}}\right)^{m-k}\right).
\end{align*}
For $\gamma=4\pi \i /3$, we have $f(2\gamma)=-f(\gamma)$ and $f(\gamma)^2=-3$.
The partition function \eqref{lastpartfcn} thus becomes
\begin{align*}
&Z_{n,m}(\gamma, \dots, \gamma, 0, \dots, 0)=\sum_{k=0}^m \left(\binom{m}{k}(e^{2\zeta}-e^{-2\gamma})^k\left(\frac{e^{2\zeta}-e^{2\gamma}}{e^{2\gamma}}\right)^{m-k}\right)\nonumber\\
&\hphantom{=\ }\times(-1)^{\binom{m+1}{2}+mn+n}\frac{2^{n^2-n-m^2-m}}{3^{2m^2-n^2+n-mn}} (n-m)!\ \varphi^{n-m}e^{\left(\binom{m+1}{2}-nm\right)\gamma}f(\gamma)^{m+n}\nonumber\\
&\hphantom{=\ }\times \prod_{j=1}^n \frac{(2j-2)!}{(4j-3)!} \prod_{j=1}^m \frac{(6j-2)!}{ (4j-1)!}\prod_{j=m+1}^n \frac{1}{(j-1)!}\det_{1\leq l,j\leq n-m}\left(W_{m+j-1}\left(-\frac{l^2}{9}; \frac{1}{3}, \frac{1}{2}, \frac{2}{3}, 1\right)\right).
\end{align*}

The terms $(e^{2\zeta}-e^{-2\gamma})^k\left(\frac{e^{2\zeta}-e^{2\gamma}}{e^{2\gamma}}\right)^{m-k}$ are linearly independent for different $k$'s as functions of $\zeta$, since the $k$th term has a zero of
degree $m-k$ in $\zeta =\gamma$. Therefore we can fix a $k$ as in the previous section, and we get 
\begin{align}
\label{partfcnk}
&Z_{n,m,k}(\gamma, \dots, \gamma, 0,\dots, 0)=\binom{m}{k}(e^{2\zeta}-e^{-2\gamma})^k\left(\frac{e^{2\zeta}-e^{2\gamma}}{e^{2\gamma}}\right)^{m-k}\nonumber\\
&\hphantom{=\ }\times(-1)^{\binom{m+1}{2}+mn+n}\frac{2^{n^2-n-m^2-m}}{3^{2m^2-n^2+n-mn}} (n-m)!\ \varphi^{n-m}e^{\left(\binom{m+1}{2}-nm\right)\gamma}f(\gamma)^{m+n}\nonumber\\
&\hphantom{=\ }\times \prod_{j=1}^n \frac{(2j-2)!}{(4j-3)!} \prod_{j=1}^m \frac{(6j-2)!}{ (4j-1)!}\prod_{j=m+1}^n \frac{1}{(j-1)!}\det_{1\leq l,j\leq n-m}\left(W_{m+j-1}\left(-\frac{l^2}{9}; \frac{1}{3}, \frac{1}{2}, \frac{2}{3}, 1\right)\right).
\end{align}
Now we can go back to $N_k$ \eqref{Nk}, the number of states %with $\nu(k_+)=k$, i.e. 
where $k$ is the number of $k_+$ turns. We insert \eqref{partfcnk} and get the following theorem.
\begin{theorem}
\label{theorem:numberofstates}
For the 6V model with DWBC and a partially reflecting end, the number of states with exactly $k$ turns of type $k_+$ is
\begin{align}
\label{thmNk} 
N_k
&=\binom{m}{k}\frac{2^{n^2-n-m^2-m}(n-m)!}{3^{2m^2-m-n^2+n-mn}} \prod_{j=1}^n \frac{(2j-2)!}{(4j-3)!} \prod_{j=1}^m \frac{(6j-2)!}{ (4j-1)!}\prod_{j=m+1}^n \frac{1}{(j-1)!}\nonumber\\
&\hphantom{=\ }\times\det_{1\leq l,j\leq n-m}\left(W_{m+j-1}\left(-\frac{l^2}{9}; \frac{1}{3}, \frac{1}{2}, \frac{2}{3}, 1\right)\right).
\end{align}
\end{theorem}
As a corollary we get that the total number of states \eqref{numberofstates} of the model is
\begin{align*}
A(m,n)
%&=\sum_{k=0}^m \binom{m}{k}\frac{2^{n^2-n-m^2-m}(n-m)!}{3^{2m^2-m-n^2+n-mn}} \prod_{j=m+1}^n \frac{(2j-2)!}{(4j-3)!(j-1)!}\prod_{j=1}^m \frac{(6j-2)!(2j-2)!}{(4j-3)! (4j-1)!}\nonumber\\
%&\quad\cdot\det_{1\leq l,j\leq n-m}\left(W_{m+j-1}\left(\left(\frac{l\i}{3}\right)^2; \frac{1}{3}, \frac{1}{2}, \frac{2}{3}, 1\right)\right)\nonumber\\
&=\frac{2^{n^2-n-m^2}(n-m)!}{3^{2m^2-m-n^2+n-mn}} \prod_{j=1}^n \frac{(2j-2)!}{(4j-3)!} \prod_{j=1}^m \frac{(6j-2)!}{ (4j-1)!}\prod_{j=m+1}^n \frac{1}{(j-1)!}\nonumber\\
&\hphantom{=\ }\times\det_{1\leq l,j\leq n-m}\left(W_{m+j-1}\left(-\frac{l^2}{9}; \frac{1}{3}, \frac{1}{2}, \frac{2}{3}, 1\right)\right).
\end{align*}
From Proposition~\ref{prop:croppedASMs} it thus follows that the number of matrices described in Section~\ref{subsec:ASMs} is also given by this expression. 

\subsection{An alternative expression for the number of states}
In this section we derive another way to present the expression \eqref{thmNk}.
We can insert the formula \eqref{wilsonpoly} for $W_{m+j-1}\left(-\frac{k^2}{9}; \frac{1}{3}, \frac{1}{2}, \frac{2}{3}, 1\right)$ into the determinant $$D=\det_{1\leq k,j\leq n-m}\left(W_{m+j-1}\left(-\frac{k^2}{9}; \frac{1}{3}, \frac{1}{2}, \frac{2}{3}, 1\right)\right).$$ 
Factor out $(5/6)_{m+j-1} (4/3)_{m+j-1} (m+j-1)!$ from each column $j$. We can rewrite the sum in each determinant entry $$\sum_{l=0}^{m+j-1} \frac{(1-m-j)_l (1/2+m+j)_l(1/3-k/3)_l(1/3+k/3)_l}{(5/6)_l(4/3)_l (l!)^2}$$ of row $k$ and column $j$ to go from $0$ to $n-1$, since the terms disappear when $l> m+j-1$. We then get
\begin{align*}
D%&=\det_{1\leq l,j\leq n-m}\left(W_{m+j-1}\left(\left(\frac{l\i}{3}\right)^2; \frac{1}{3}, \frac{1}{2}, \frac{2}{3}, 1\right)\right)\\
%&=\prod_{k=m}^{n-1}(5/6)_k (4/3)_k k! \\
%&\det_{1\leq i,j\leq n-m}\left( \sum_{l=0}^{m+j-1} \frac{(-m+1-j)_l (3/2+m-1+j)_l(1/3+i/3)_l(1/3-i/3)_l}{(5/6)_l (4/3)_l (l!)^2}\right) \\
%&=\prod_{k=m}^{n-1}(5/6)_k (4/3)_k k! \\
%&\det_{1\leq i,j\leq n-m}\left( \sum_{l=0}^{n-1} \frac{(-m+1-j)_l (3/2+m-1+j)_l(1/3+i/3)_l(1/3-i/3)_l}{(5/6)_l (4/3)_l (l!)^2}\right)\\
&=\prod_{j=m}^{n-1}(5/6)_j (4/3)_j j!  \\
&\hphantom{=\ }\times\det_{1\leq k,j\leq n-m}\left( \sum_{l=0}^{n-1} \frac{(1-m-j)_l (1/2+m+j)_l(1/3+k/3)_l(1/3-k/3)_l}{(5/6)_l (4/3)_l (l!)^2}\right).
\end{align*}
%When $k=1$, the factor $(1/3-k/3)_l=1$ for $l=0$ and disappears otherwise. 
We use linearity of the rows and write the determinant as a sum of determinants. The above becomes
\begin{align*}
D
%&=\prod_{k=m}^{n-1} (5/6)_k (4/3)_k k!  \\
%&\sum_{l_1, l_2, \dots, l_{n-m}=0}^{n-1}\det_{1\leq i,j\leq n-m}\left(\frac{(1-m-j)_{l_i} (1/2+m+j)_{l_i}((1+i)/3)_{l_i}((1-i)/3)_{l_i}}{(5/6)_{l_i} (4/3)_{l_i} (l_i!)^2}\right) \\ =
&=\prod_{j=m}^{n-1}(5/6)_j (4/3)_j j!\sum_{l_1, l_2, \dots, l_{n-m}=0}^{n-1}\left(\prod_{i=1}^{n-m} \frac{((1-i)/3)_{l_i}((1+i)/3)_{l_i}}{(5/6)_{l_i} (4/3)_{l_i} (l_i!)^2}\right.\\
&\hphantom{=\ }\left.\times\det_{1\leq k,j\leq n-m}\left((1-m-j)_{l_k} (1/2+m+j)_{l_k}\right)\vphantom{\prod_{i=1}^{n-m}}\right).
\end{align*}
Since $(x)_j/(x)_i=(x+i)_{j-i}$, for $j\geq i$, we can write 
\begin{multline*}
(1-m-j)_{l_k} (1/2+m+j)_{l_k}\\
%&=(1-n+l_i+n-m-j-l_i)_{n-m-j-n+m+j+l_i}(m+3/2+l_i+j-1-l_i)_{j-1-j+1+l_i}\\
%&=\frac{(1-n+l_i)_{n-m-j}}{(1-n+l_i)_{n-m-j-l_i}}\frac{(m+3/2+l_i)_{j-1}}{(m+3/2+l_i)_{j-1-l_i}}\\
%&=\frac{(1-n)_{l_i}(1-n+l_i)_{n-m-j}}{(1-n)_{n-m-j}}\frac{(m+3/2+l_i)_{j-1}(m+3/2)_{l_i}}{(m+3/2)_{j-1}}\\
=\frac{(1-n)_{l_k}(m+3/2)_{l_k}}{(1-n)_{n-m-j}(m+3/2)_{j-1}}(1-n+l_k)_{n-m-j}(m+3/2+l_k)_{j-1},
\end{multline*}
and we can factor out all factors that depend either only on the row $k$ or only on the column~$j$. 
The determinant $D$ now becomes
\begin{align*}
D&=
%\prod_{j=m}^{n-1}(5/6)_j (4/3)_j j! \\ 
%& \sum_{l_1, l_2, \dots, l_{n-m}=0}^{n-1}\prod_{i=1}^{n-m} \frac{((1-i)/3)_{l_i}((1+i)/3)_{l_i}}{(5/6)_{l_i} (4/3)_{l_i} (l_i!)^2}\\
%&\det_{1\leq i,j,\leq n-m}\left(\frac{(1-n)_{l_i}(m+3/2)_{l_i}}{(1-n)_{n-m-j}(m+3/2)_{j-1}}(1-n+l_i)_{n-m-j}(m+3/2+l_i)_{j-1}\right)\\ &=
\prod_{j=m}^{n-1}(5/6)_j (4/3)_j j! \sum_{l_1, l_2, \dots, l_{n-m}=0}^{n-1}\left(\prod_{i=1}^{n-m} \frac{((1-i)/3)_{l_i}((1+i)/3)_{l_i}(1-n)_{l_i}(m+3/2)_{l_i}}{(5/6)_{l_i} (4/3)_{l_i} (l_i!)^2 (1-n)_{i-1}(m+3/2)_{i-1}}\right. \\
&\hphantom{=\ }\times\left.\det_{1\leq k,j,\leq n-m}\left((1-n+l_k)_{n-m-j}(m+3/2+l_k)_{j-1}\right)\vphantom{\prod_{i=1}^{n-m}}\vphantom{\prod_{i=1}^{n-m}}\right).
\end{align*}
Each element of the alternant matrix $$\hat D=\det_{1\leq k,j,\leq n-m}\left((1-n+l_k)_{n-m-j}(m+3/2+l_k)_{j-1}\right)$$ is a polynomial in $l_k$ of degree $n-m-1$. This is a special case of a bigger family of determinants, see further \cite[Lemma~3]{Krattenthaler2001}. The determinant is 
$$\hat D=C \prod_{1\leq i < j \leq n-m} (l_i-l_j)$$ 
for some $C$ which does not depend on the $l_k$'s. To find $C$, we put $l_k=-k-m-1/2$ which makes the matrix triangular. Then 
\begin{align*}
%&\det_{1\leq i,j,\leq n-m}\left((1-n+l_i)_{n-m-j}(m+3/2+l_i)_{j-1}\right)\\
%&=\det_{1\leq i,j,\leq n-m}\left((1/2-n-m-i)_{n-m-j}(1-i)_{j-1}\right)\\
%&=(1/2-n-m-1)_{n-m-1}(-1)(1/2-n-m-2)_{n-m-2}(-2)_2(1/2-n-m-3)_{n-m-3}\cdots(1-n+m-1)_{n-m-2}(1/2-n-m-n+m+1)_{1}(1-n+m)_{n-m-1}(1/2-n-m-n+m)_{0}\\
%=
&\prod_{j=0}^{n-m-1}(-1)^j j!(-1/2-n-m-j)_{n-m-1-j}
=C \prod_{1\leq i < j \leq n-m} (j-i),
\end{align*}
so 
\begin{equation*}
C
%=\frac{\prod_{i=0}^{n-m-1}(-1)^i i!(-1/2-n-m-i)_{n-m-1-i}}{\prod_{1\leq i < j \leq n-m} (j-i)}\\
%=\frac{\prod_{i=0}^{n-m-1}(-1)^i i!(-1/2-n-m-i)_{n-m-1-i}}{\prod_{i=1}^{n-m-1} i!}\\
%=\prod_{i=0}^{n-m-1}(-1)^i (-1/2-n-m-i)_{n-m-1-i}\\
%=\prod_{j=0}^{n-m-1}(-1)^{j} (1/2-2n+j)_{j}\\
%=(-1)^{\binom{n-m}{2}}\prod_{j=0}^{n-m-1} (1/2-2n+j)_{j}\\
=\prod_{j=1}^{n-m}(5/2+2n-2j)_{j-1}.
\end{equation*}
Hence the determinant $D$ becomes
\begin{align*}
&D
=\prod_{j=m}^{n-1}(5/6)_j (4/3)_j j!\prod_{j=1}^{n-m} \frac{(5/2+2n-2j)_{j-1}}{(1-n)_{j-1}(m+3/2)_{j-1}}\\
&\times\sum_{l_1, l_2, \dots, l_{n-m}=0}^{n-1}\left(\prod_{i=1}^{n-m} \frac{((1-i)/3)_{l_i}((1+i)/3)_{l_i}(1-n)_{l_i}(m+3/2)_{l_i}}{(5/6)_{l_i} (4/3)_{l_i} (l_i!)^2}\vphantom{\prod_{1\leq i < j \leq n-m}}\prod_{1\leq i < j \leq n-m} (l_i-l_j)\right). 
\end{align*}
Series of similar type appear e.g. in \cite{GustafsonKrattenthaler1997, Schlosser2000}. 

Inserting the above into \eqref{thmNk} of the last section, we get an expression for the number of states with exactly $k$ turns of type $k_+$ as an $(n-m)$-fold hypergeometric sum, 
\begin{align*}
&N_k
%&=\binom{m}{k}\frac{2^{n^2-n-m^2-m}(n-m)!}{3^{2m^2-m-n^2+n-mn}} \prod_{j=m+1}^n \frac{(2j-2)!}{(4j-3)!(j-1)!}\prod_{j=1}^m \frac{(6j-2)!(2j-2)!}{(4j-3)! (4j-1)!}\nonumber\\
%&\quad\prod_{k=m}^{n-1}(5/6)_k (4/3)_k k!\prod_{k=1}^{n-m} \frac{(5/2+2n-2k)_{k-1}}{(1-n)_{k-1}(3/2+m)_{k-1}}\nonumber\\
%&\quad\cdot\sum_{l_1, l_2, \dots, l_{n-m}=0}^{n-1}\left(\prod_{i=1}^{n-m} \frac{((1-i)/3)_{l_i}((1+i)/3)_{l_i}(1-n)_{l_i}(m+3/2)_{l_i}}{(5/6)_{l_i} (4/3)_{l_i} (l_i!)^2}\prod_{1\leq i < j \leq n-m} (l_i-l_j)\right)\nonumber\\
%&=\binom{m}{k}\frac{2^{n^2-n-m^2-m}(n-m)!}{3^{2m^2-m-n^2+n-mn}} \prod_{j=m+1}^n \frac{(2j-2)!(5/6)_{j-1} (4/3)_{j-1} }{(4j-3)!}\nonumber\\
%&\quad\cdot\prod_{j=1}^m \frac{(6j-2)!(2j-2)!}{(4j-3)! (4j-1)!}\prod_{j=1}^{n-m} \frac{(5/2+2n-2k)_{k-1}}{(1-n)_{j-1}(3/2+m)_{j-1}}\nonumber\\
%&\quad\cdot\sum_{l_1, l_2, \dots, l_{n-m}=0}^{n-1}\left(\prod_{i=1}^{n-m} \frac{((1-i)/3)_{l_i}((1+i)/3)_{l_i}(1-n)_{l_i}(m+3/2)_{l_i}}{(5/6)_{l_i} (4/3)_{l_i} (l_i!)^2}\prod_{1\leq i < j \leq n-m} (l_i-l_j)\right)\nonumber\\
=\binom{m}{k}\frac{2^{n^2-n-m^2-m}(n-m)!}{3^{2m^2-m-n^2+n-mn}} \prod_{j=1}^n \frac{(2j-2)!}{(4j-3)!}\prod_{j=1}^m \frac{(6j-2)!}{(4j-1)!} \nonumber\\
&\times\prod_{j=m+1}^n (5/6)_{j-1} (4/3)_{j-1}\prod_{j=1}^{n-m} \frac{(5/2+2n-2j)_{j-1}}{(1-n)_{j-1}(m+3/2)_{j-1}}\nonumber\\
&\times\sum_{l_1, l_2, \dots, l_{n-m}=0}^{n-1}\left(\prod_{i=1}^{n-m} \frac{((1-i)/3)_{l_i}((1+i)/3)_{l_i}(1-n)_{l_i}(m+3/2)_{l_i}}{(5/6)_{l_i} (4/3)_{l_i} (l_i!)^2}\vphantom{\prod_{1\leq i < j \leq n-m}}\prod_{1\leq i < j \leq n-m} (l_i-l_j)\right).
\end{align*}

\subsection*{Acknowledgments}
First I would like to thank my supervisor Hjalmar Rosengren for his support and encouragement throughout the whole research process. I am also very thankful for discussions with my cosupervisor Jules Lamers. 

\appendix
\section*{Appendix A: Foda--Wheeler method}
\addcontentsline{toc}{section}{Appendix A: Foda--Wheeler method}
\setcounter{section}{1}
\setcounter{theorem}{0}
\setcounter{equation}{0}
\setcounter{figure}{0}
\label{sec:fodawheeler}
%In Section~\ref{sec:fodazarembo}, we find a determinant formula for the partition function of the 6V model on a $2n\times m$ lattice with DWBC and one partially reflecting end by using the Izergin--Korepin method. 
Another way to find the partition function is suggested in \cite{FodaZarembo2016}, by starting from Tsuchiya's determinant for the partition function on a $2n\times n$ lattice, and step by step removing the extra vertical lines by letting the corresponding variables $\mu_j\to \infty$, following Foda and Wheeler \cite{FodaWheeler2012}. 

To do this we need to choose weights that behave well in the limit. The weights 
%\ref{vertexweights} 
that we have chosen in Section~\ref{sec:prel} will do the job. 
In the lattice, the local weights are either $w(\lambda-\mu)$ for vertices on the upper part of a double line, or $w(\mu+\lambda)$ for vertices on the lower part, where $w$ is one of $a_\pm$, $b_\pm$ and $c_\pm$. When letting $\mu_i\to \infty$ the local weights become
\begingroup
\allowdisplaybreaks
\begin{align*}
&\lim_{\mu\to\infty} a_\pm(\lambda-\mu)=\lim_{\mu\to\infty} a_\pm(\lambda+\mu)=1,&\ \\
&\lim_{\mu\to\infty} b_-(\lambda-\mu)=e^{2\gamma},&\ &\lim_{\mu\to\infty} b_+(\lambda-\mu)=1,\\
&\lim_{\mu\to\infty} b_-(\lambda+\mu)=1, &  &\lim_{\mu\to\infty} b_+(\lambda+\mu)=e^{-2\gamma},\\
&\lim_{\mu\to\infty} c_-(\lambda-\mu)=-e^\gamma(e^\gamma-e^{-\gamma}), &\  &\lim_{\mu\to\infty} c_+(\lambda-\mu)=-\lim_{\mu\to\infty} e^\gamma(e^\gamma-e^{-\gamma}) \frac{1}{e^{2\mu-2\lambda}},\\
&\lim_{\mu\to\infty} c_+(\lambda+\mu)=\frac{e^\gamma-e^{-\gamma}}{e^\gamma},&\  &\lim_{\mu\to\infty} c_-(\lambda+\mu)=\lim_{\mu\to\infty} \left(\frac{e^\gamma-e^{-\gamma}}{e^\gamma}\right) \frac{1}{e^{2\mu+2\lambda}}.
\end{align*}
\endgroup
We start from Tsuchiya's determinant for $2n\times n$ lattices,
\begin{align}
Z_n(\boldsymbol{\lambda}, \boldsymbol{\mu})&= e^{-\binom{n+1}{2}\gamma} f(\gamma)^n \frac{\prod_{i=1}^n\left[e^{\mu_i+\zeta} f(\mu_i-\zeta)f(2\lambda_i)\right]\prod_{i=1}^n\prod_{j=1}^n f(\mu_j\pm\lambda_i)}{\prod_{1\leq i<j\leq n}\left[f(\mu_j\pm\mu_i)f(\lambda_i-\lambda_j)f(\lambda_i+\lambda_j+\gamma)\right]} \det_{1\leq i,j\leq n} M,
\label{Tsuchiyadeterminantformula}
\end{align}
where $M$ is an $n\times n$ matrix with
$$M_{ij}=
\frac{1}{f(\mu_i\pm\lambda_j)f(\mu_i\pm(\lambda_j+\gamma))}.
$$
Tsuchiya considered a diagonal $K$-matrix, but we have a triangular $K$-matrix. This does not affect the determinant formula in the $2n\times n$ case, since the ice rule implies that there can not be any $k_c$ turns, coming from the off-diagonal entry in the $K$-matrix.

We first consider the $2\times 1$ lattice with reflecting end, since the observations that can be made for this case are important in the general case. This partition function consists of two terms (see Figure~\ref{figure:kpkmiskc}), 
$$c_+(\lambda+\mu)k_+(\lambda)b_+(\lambda-\mu)+b_+(\lambda+\mu)k_-(\lambda)c_-(\lambda-\mu),$$
 which in the limit is $\frac{e^\gamma-e^{-\gamma}}{e^\gamma}(k_+(\lambda)-k_-(\lambda))$. It is easy to check that $k_+(\lambda)-k_-(\lambda)=\frac{k_c(\lambda)}{\varphi}$. Hence letting $\mu\to\infty$ yields that the partition function becomes $\frac{e^\gamma-e^{-\gamma}}{\varphi e^\gamma} k_c(\lambda)$. This makes sense in the picture as well, where removing the vertical line means that both the $k_+$ and the $k_-$ turn in the $2\times 1$ lattice become a $k_c$ turn in the $2\times 0$ lattice. Conversely, the $k_c$ turn in the smaller lattice can come from removing the vertical line from a state with a $k_+$ turn, as well as from a state with a $k_-$ turn.

\begin{figure}[h]
\centering
\begin{tikzpicture}[baseline={([yshift=-.5*10pt*0.6+6]current bounding box.center)}, scale=0.6, font=\scriptsize]
%horizontal
		\draw[midarrow={stealth}] (1.55,1.5-.25-.38) -- +(0.01,0);
		\draw[midarrow={stealth}] (1.55,1.5-.25+.38) -- +(0.01,0);
		
		\draw (0.38, 1.5-.25-.38) -- +(1.62, 0);
		\draw[->] (0.38, 1.5-.25+.38) -- +(1.62, 0);
		
		\draw (0.38,1.5-.25+.38) arc (90:270:0.38);
		
	\node[anchor=west] at (2, 1.5-.25-.38) {$-\lambda$};
		
	% then the vertical border lines:
		\draw[->] (1,0) -- +(0,2.5); 
			\draw[midarrow={stealth}] (1,0.55) -- +(0,0.01);
			\draw[midarrow={stealth reversed}] (1,1.5+.13+.43)  -- +(0,0.01);	
	
	\node at (1,-0.3) {$\mu$};
	
		%left border
		
		%\draw[-stealth reversed] (0,1.5+.05-0.25) -- (0,1.5+.06-0.25);
		\draw[midarrow={stealth}] (.55,0.5+.75-.38) -- +(0.01,0);
		\draw[midarrow={stealth reversed}] (.55,0.5+.75+.38) -- +(0.01,0);
		\draw [midarrow={stealth}] (1,0.87+.43) -- +(0,0.01); 
			
\end{tikzpicture}
 \! $+$ \ \
\begin{tikzpicture}[baseline={([yshift=-.5*10pt*0.6+6]current bounding box.center)}, scale=0.6, font=\scriptsize]
%horizontal
		\draw[midarrow={stealth}] (1.55,1.5-.25-.38) -- +(0.01,0);
		\draw[midarrow={stealth}] (1.55,1.5-.25+.38) -- +(0.01,0);
		
		\draw (0.38, 1.5-.25-.38) -- +(1.62, 0);
		\draw[->] (0.38, 1.5-.25+.38) -- +(1.62, 0);
		
		\draw (0.38,1.5-.25+.38) arc (90:270:0.38);
		
	\node[anchor=west] at (2, 1.5-.25-.38) {$-\lambda$};
		
	% then the vertical border lines:
		\draw[->] (1,0) -- +(0,2.5); 
			\draw[midarrow={stealth}] (1,0.55) -- +(0,0.01);
			\draw[midarrow={stealth reversed}] (1,1.5+.13+.43)  -- +(0,0.01);	
	
	\node at (1,-0.3) {$\mu$};
	
		%left border
		
		%\draw[-stealth] (0,1.5+.05-0.25) -- (0,1.5+.06-0.25);
		\draw[midarrow={stealth reversed}] (.55,0.5+.75-.38) -- +(0.01,0);
		\draw[midarrow={stealth}] (.55,0.5+.75+.38) -- +(0.01,0);
		\draw [midarrow={stealth reversed}] (1,0.87+.43) -- +(0,0.01); 
			
\end{tikzpicture}
 \! $=$\quad $\dfrac{e^\gamma-e^{-\gamma}}{x e^\gamma}$\ $\times$\ \
\begin{tikzpicture}[baseline={([yshift=-.5*10pt*0.6+3]current bounding box.center)}, scale=0.6, font=\scriptsize]
%horizontal
		
		\draw (0.38, 1.5-.25-.38) -- +(0.62, 0);
		\draw[->] (0.38, 1.5-.25+.38) -- +(0.62, 0);
		
		\draw (0.38,1.5-.25+.38) arc (90:270:0.38);
		
	\node[anchor=west] at (1, 1.5-.25-.38) {$-\lambda$};
		
		\draw[midarrow={stealth}] (.55,0.5+.75-.38) -- +(0.01,0);
		\draw[midarrow={stealth}] (.55,0.5+.75+.38) -- +(0.01,0);
			
\end{tikzpicture}
\caption{A $k_c$ turn comes from the sum of a $k_+$ and $k_-$ turn, when removing the vertical line from the one size bigger lattice.}
\label{figure:kpkmiskc}
\end{figure}

\begin{figure}[ht]
\center
\subcaptionbox{\label{type1}}{%
\begin{tikzpicture}[baseline={([yshift=-.5*10pt*0.6+6]current bounding box.center)}, scale=0.6, font=\scriptsize]
\draw[midarrow={stealth}] (1.55,1.5-.25-.38) -- +(0.01,0);
		\draw[midarrow={stealth}] (1.55,1.5-.25+.38) -- +(0.01,0);
		
		\draw (0.38, 1.5-.25-.38) -- +(1.62, 0);
		\draw[->] (0.38, 1.5-.25+.38) -- +(1.62, 0);
		
		\draw (0.38,1.5-.25+.38) arc (90:270:0.38);
		
	\node[anchor=west] at (2, 1.5-.25-.38) {$-\lambda_i$};
		
	% then the vertical border lines:
		\draw[->] (1,0) -- +(0,2.5); 
			\draw[midarrow={stealth}] (1,0.55) -- +(0,0.01);
			\draw[midarrow={stealth reversed}] (1,1.5+.13+.43)  -- +(0,0.01);	
	
	\node at (1,-0.3) {$\mu_{m+1}$};
	
		%left border
		
		%\draw[-stealth reversed] (0,1.5+.05-0.25) -- (0,1.5+.06-0.25);
		\draw[midarrow={stealth}] (.55,0.5+.75-.38) -- +(0.01,0);
		\draw[midarrow={stealth reversed}] (.55,0.5+.75+.38) -- +(0.01,0);
		\draw [midarrow={stealth}] (1,0.87+.43) -- +(0,0.01); 
\end{tikzpicture}			
}\hfil
\subcaptionbox{\label{type2}}{%	
\begin{tikzpicture}[baseline={([yshift=-.5*10pt*0.6+6]current bounding box.center)}, scale=0.6, font=\scriptsize]
%horizontal
		\draw[midarrow={stealth}] (1.55,1.5-.25-.38) -- +(0.01,0);
		\draw[midarrow={stealth}] (1.55,1.5-.25+.38) -- +(0.01,0);
		
		\draw (0.38, 1.5-.25-.38) -- +(1.62, 0);
		\draw[->] (0.38, 1.5-.25+.38) -- +(1.62, 0);
		
		\draw (0.38,1.5-.25+.38) arc (90:270:0.38);
		
	\node[anchor=west] at (2, 1.5-.25-.38) {$-\lambda_i$};
		
	% then the vertical border lines:
		\draw[->] (1,0) -- +(0,2.5); 
			\draw[midarrow={stealth}] (1,0.55) -- +(0,0.01);
			\draw[midarrow={stealth reversed}] (1,1.5+.13+.43)  -- +(0,0.01);	
	
	\node at (1,-0.3) {$\mu_{m+1}$};
	
		%left border
		
		%\draw[-stealth] (0,1.5+.05-0.25) -- (0,1.5+.06-0.25);
		\draw[midarrow={stealth reversed}] (.55,0.5+.75-.38) -- +(0.01,0);
		\draw[midarrow={stealth}] (.55,0.5+.75+.38) -- +(0.01,0);
		\draw [midarrow={stealth reversed}] (1,0.87+.43) -- +(0,0.01); 
\end{tikzpicture}
}\hfil
\subcaptionbox{\label{type3}}{%	
\begin{tikzpicture}[baseline={([yshift=-.5*10pt*0.6+6]current bounding box.center)}, scale=0.6, font=\scriptsize]
%horizontal
		\draw[midarrow={stealth}] (1.55,1.5-.25-.38) -- +(0.01,0);
		\draw[midarrow={stealth}] (1.55,1.5-.25+.38) -- +(0.01,0);
		
		\draw (0.38, 1.5-.25-.38) -- +(1.62, 0);
		\draw[->] (0.38, 1.5-.25+.38) -- +(1.62, 0);
		
		\draw (0.38,1.5-.25+.38) arc (90:270:0.38);
		
	\node[anchor=west] at (2, 1.5-.25-.38) {$-\lambda_i$};
		
	% then the vertical border lines:
		\draw[->] (1,0) -- +(0,2.5); 
			\draw[midarrow={stealth}] (1,0.55) -- +(0,0.01);
			\draw[midarrow={stealth}] (1,1.5+.13+.43)  -- +(0,0.01);	
	
	\node at (1,-0.3) {$\mu_{m+1}$};
	
		%left border
		
		\draw[midarrow={stealth}] (.55,0.5+.75-.38) -- +(0.01,0);
		\draw[midarrow={stealth}] (.55,0.5+.75+.38) -- +(0.01,0);
		\draw[midarrow={stealth}] (1,0.87+.43) -- +(0,0.01); 
			
\end{tikzpicture}
}\hfil
\subcaptionbox{\label{type4}}{%	
\begin{tikzpicture}[baseline={([yshift=-.5*10pt*0.6+6]current bounding box.center)}, scale=0.6, font=\scriptsize]
%horizontal
		\draw[midarrow={stealth}] (1.55,1.5-.25-.38) -- +(0.01,0);
		\draw[midarrow={stealth}] (1.55,1.5-.25+.38) -- +(0.01,0);
		
		\draw (0.38, 1.5-.25-.38) -- +(1.62, 0);
		\draw[->] (0.38, 1.5-.25+.38) -- +(1.62, 0);
		
		\draw (0.38,1.5-.25+.38) arc (90:270:0.38);
		
	\node[anchor=west] at (2, 1.5-.25-.38) {$-\lambda_i$};
		
	% then the vertical border lines:
		\draw[->] (1,0) -- +(0,2.5); 
			\draw[midarrow={stealth reversed}] (1,0.55) -- +(0,0.01);
			\draw[midarrow={stealth reversed}] (1,1.5+.13+.43)  -- +(0,0.01);	
	
	\node at (1,-0.3) {$\mu_{m+1}$};
	
		%left border
		
		\draw[midarrow={stealth}] (.55,0.5+.75-.38) -- +(0.01,0);
		\draw[midarrow={stealth}] (.55,0.5+.75+.38) -- +(0.01,0);
		\draw[midarrow={stealth reversed}] (1,0.87+.43) -- +(0,0.01); 
\end{tikzpicture}
}
\caption{The $k_c$ turns in a lattice of size $2n\times m$ come from either one of these configurations in the lattice of size $2n\times (m+1)$, when letting $\mu_{m+1}\to\infty$.}
\label{figure:kcvertices}
\bigskip
\subcaptionbox*{}{%
\begin{tikzpicture}[baseline={([yshift=-.5*10pt*0.6+6]current bounding box.center)}, scale=0.6, font=\scriptsize]
%horizontal
		\draw[midarrow={stealth reversed}] (1.55,1.5-.25-.38) -- +(0.01,0);
		\draw[midarrow={stealth}] (1.55,1.5-.25+.38) -- +(0.01,0);
		
		\draw (0.38, 1.5-.25-.38) -- +(1.62, 0);
		\draw[->] (0.38, 1.5-.25+.38) -- +(1.62, 0);
		
		\draw (0.38,1.5-.25+.38) arc (90:270:0.38);
		
	\node[anchor=west] at (2, 1.5-.25-.38) {$-\lambda_i$};
		
	% then the vertical border lines:
		\draw[->] (1,0) -- +(0,2.5); 
			\draw[midarrow={stealth}] (1,0.55) -- +(0,0.01);
			\draw[midarrow={stealth}] (1,1.5+.13+.43)  -- +(0,0.01);	
	
	\node at (1,-0.3) {$\mu_{m+1}$};
	
		%left border
		
		%\draw[-stealth] (0,1.5+.05-0.25) -- (0,1.5+.06-0.25);
		\draw[midarrow={stealth reversed}] (.55,0.5+.75-.38) -- +(0.01,0);
		\draw[midarrow={stealth}] (.55,0.5+.75+.38) -- +(0.01,0);
		\draw [midarrow={stealth}] (1,0.87+.43) -- +(0,0.01); 
			
\end{tikzpicture}
}\hfil
\subcaptionbox*{}{%	
\begin{tikzpicture}[baseline={([yshift=-.5*10pt*0.6+6]current bounding box.center)}, scale=0.6, font=\scriptsize]
%horizontal
		\draw[midarrow={stealth reversed}] (1.55,1.5-.25-.38) -- +(0.01,0);
		\draw[midarrow={stealth}] (1.55,1.5-.25+.38) -- +(0.01,0);
		
		\draw (0.38, 1.5-.25-.38) -- +(1.62, 0);
		\draw[->] (0.38, 1.5-.25+.38) -- +(1.62, 0);
		
		\draw (0.38,1.5-.25+.38) arc (90:270:0.38);
		
	\node[anchor=west] at (2, 1.5-.25-.38) {$-\lambda_i$};
		
	% then the vertical border lines:
		\draw[->] (1,0) -- +(0,2.5); 
			\draw[midarrow={stealth reversed}] (1,0.55) -- +(0,0.01);
			\draw[midarrow={stealth reversed}] (1,1.5+.13+.43)  -- +(0,0.01);	
	
	\node at (1,-0.3) {$\mu_{m+1}$};
	
		%left border
		
		%\draw[-stealth] (0,1.5+.05-0.25) -- (0,1.5+.06-0.25);
		\draw[midarrow={stealth reversed}] (.55,0.5+.75-.38) -- +(0.01,0);
		\draw[midarrow={stealth}] (.55,0.5+.75+.38) -- +(0.01,0);
		\draw [midarrow={stealth reversed}] (1,0.87+.43) -- +(0,0.01); 
			
\end{tikzpicture}
}\hfil
\subcaptionbox*{}{%	
\begin{tikzpicture}[baseline={([yshift=-.5*10pt*0.6+6]current bounding box.center)}, scale=0.6, font=\scriptsize]
%horizontal
		\draw[midarrow={stealth}] (1.55,1.5-.25-.38) -- +(0.01,0);
		\draw[midarrow={stealth reversed}] (1.55,1.5-.25+.38) -- +(0.01,0);
		
		\draw (0.38, 1.5-.25-.38) -- +(1.62, 0);
		\draw[->] (0.38, 1.5-.25+.38) -- +(1.62, 0);
		
		\draw (0.38,1.5-.25+.38) arc (90:270:0.38);
		
	\node[anchor=west] at (2, 1.5-.25-.38) {$-\lambda_i$};
		
	% then the vertical border lines:
		\draw[->] (1,0) -- +(0,2.5); 
			\draw[midarrow={stealth}] (1,0.55) -- +(0,0.01);
			\draw[midarrow={stealth}] (1,1.5+.13+.43)  -- +(0,0.01);	
	
	\node at (1,-0.3) {$\mu_{m+1}$};
	
		%left border
		
		%\draw[-stealth reversed] (0,1.5+.05-0.25) -- (0,1.5+.06-0.25);
		\draw[midarrow={stealth}] (.55,0.5+.75-.38) -- +(0.01,0);
		\draw[midarrow={stealth reversed}] (.55,0.5+.75+.38) -- +(0.01,0);
		\draw [midarrow={stealth}] (1,0.87+.43) -- +(0,0.01); 
			
\end{tikzpicture}
}\hfil
\subcaptionbox*{}{%	
\begin{tikzpicture}[baseline={([yshift=-.5*10pt*0.6+6]current bounding box.center)}, scale=0.6, font=\scriptsize]
%horizontal
		\draw[midarrow={stealth}] (1.55,1.5-.25-.38) -- +(0.01,0);
		\draw[midarrow={stealth reversed}] (1.55,1.5-.25+.38) -- +(0.01,0);
		
		\draw (0.38, 1.5-.25-.38) -- +(1.62, 0);
		\draw[->] (0.38, 1.5-.25+.38) -- +(1.62, 0);
		
		\draw (0.38,1.5-.25+.38) arc (90:270:0.38);
		
	\node[anchor=west] at (2, 1.5-.25-.38) {$-\lambda_i$};
		
	% then the vertical border lines:
		\draw[->] (1,0) -- +(0,2.5); 
			\draw[midarrow={stealth reversed}] (1,0.55) -- +(0,0.01);
			\draw[midarrow={stealth reversed}] (1,1.5+.13+.43)  -- +(0,0.01);	
	
	\node at (1,-0.3) {$\mu_{m+1}$};
	
		%left border
		
		%\draw[-stealth reversed] (0,1.5+.05-0.25) -- (0,1.5+.06-0.25);
		\draw[midarrow={stealth}] (.55,0.5+.75-.38) -- +(0.01,0);
		\draw[midarrow={stealth reversed}] (.55,0.5+.75+.38) -- +(0.01,0);
		\draw [midarrow={stealth reversed}] (1,0.87+.43) -- +(0,0.01); 
			
\end{tikzpicture}
}
\vspace{-6mm}
\caption{The turns that are not $k_c$ turns in a lattice of size $2n\times m$ come from either one of these configurations in the lattice of size $2n\times (m+1)$, when letting $\mu_{m+1}\to\infty$.}
\label{figure:zeroconfigs}
\end{figure}

\begin{figure}[!htb]
\centering
\begin{tikzpicture}[baseline={([yshift=-.5*10pt*0.6+1pt]current bounding box.center)}, scale=0.6, font=\scriptsize]
	% first the horizontal border lines:
	\foreach \y in {1,3,5} {
		
		\draw (0.38, 1.5*\y-.25-.38) -- +(1.62, 0);
		\draw (0.38, 1.5*\y-.25+.38) -- +(1.62, 0);
		
		\foreach \x in {-1,...,1} \draw (2.4+.2*\x, 1.5*\y-.25-.38) node{$\cdot\mathstrut$};
		\foreach \x in {-1,...,1} \draw (2.4+.2*\x, 1.5*\y-.25+.38) node{$\cdot\mathstrut$};
		
		\draw (0.38,1.5*\y-.25+.38) arc (90:270:0.38);
	}
		
	\node[anchor=west] at (3, 1.5-.25-.38) {$-\lambda_i$};
	\node[anchor=west] at (3, 1.5*3-.25-.38) {$-\lambda_{k}$};
	\node[anchor=west] at (3, 1.5*5-.25-.38) {$-\lambda_j$};
		
	% then the vertical border lines:
	\foreach \x in {1} {
		
		\foreach \y in {-1,...,1} \draw (\x, 0-0.25+.2*\y) node{$\cdot\mathstrut$};
		\draw (\x,0-0.25+0.38) -- (1,3-0.25-.38); 
		\foreach \y in {-1,...,1} \draw (\x, 3-0.25+.2*\y) node{$\cdot\mathstrut$};
		\draw (\x,3-0.25+.38) -- (\x,6-0.25-.38); 
		\foreach \y in {-1,...,1} \draw (\x, 6-0.25+.2*\y) node{$\cdot\mathstrut$};
		\draw (\x,6-0.25+.38) -- (1,9-0.25-0.38); 
		\foreach \y in {-1,...,1} \draw (\x, 9-0.25+.2*\y) node{$\cdot\mathstrut$};
		\draw[->] (\x,9-0.25+0.38)  -- +(0,.2);

	}
	
	\node at (1, -0.25-.2-0.3) {${\mu_{m+1}}$};
	
		%left border
		%\draw[-stealth] (0,7.5+.05-0.25) -- (0,7.5+.06-0.25);
		\draw[midarrow={stealth}] (.55,6.5+.75-.38) -- +(0.01,0);
		\draw[midarrow={stealth}] (.55,6.5+.75+.38) -- +(0.01,0);
		
		%\draw[-stealth] (0,4.5+.05-0.25) -- (0,4.5+.06-0.25);
		\draw[midarrow={stealth reversed}] (.55,3.5+.75-.38) -- +(0.01,0);
		\draw[midarrow={stealth}] (.55,3.5+.75+.38) -- +(0.01,0);
		
		%\draw[-stealth] (0,3.05-0.25) -- (0,3.06-0.25);
		%\draw[midarrow={stealth reversed}] (.55,2.75-.38) -- +(0.01,0);
		%\draw[midarrow={stealth}] (.55,2.75+.38) -- +(0.01,0);
		
		%\draw[-stealth] (0,1.5+.05-0.25) -- (0,1.5+.06-0.25);
		\draw[midarrow={stealth}] (.55,0.5+.75-.38) -- +(0.01,0);
		\draw[midarrow={stealth}] (.55,0.5+.75+.38) -- +(0.01,0);
		
		%bulk horizontal lines
	
		\draw [midarrow={stealth}] (1.55, 7.5-.25-.38) -- +(0.01, 0);
		\draw [midarrow={stealth}] (1.55, 7.5-.25+.38) -- +(0.01, 0);
		
		\draw [midarrow={stealth}] (1.55, 4.5-.25-.38) -- +(0.01, 0);
		\draw [midarrow={stealth}] (1.55, 4.5-.25+.38) -- +(0.01, 0);
		
		%\draw [midarrow={stealth reversed}] (1.55, 3-.25-.38) -- +(0.01, 0);
		%\draw [midarrow={stealth}] (1.55, 3-.25+.38) -- +(0.01, 0);
		
		\draw [midarrow={stealth}] (1.55, 1.5-.25-.38) -- +(0.01, 0);
		\draw [midarrow={stealth}] (1.55, 1.5-.25+.38) -- +(0.01, 0);
		
	%bulk vertical lines
	
			%\draw [midarrow={stealth reversed}] (1,4.63+.43) -- +(0,0.01); 
			%\draw [midarrow={stealth reversed}] (1,5.37+.43) -- +(0,0.01); 
			%\draw [midarrow={stealth reversed}] (2,4.63+.43) -- +(0,0.01); 
			%\draw [midarrow={stealth reversed}] (2,5.37+.43) -- +(0,0.01); 

		\draw[midarrow={stealth}] (1,0.55) -- +(0,0.01);
			\draw [midarrow={stealth}] (1,0.87+.43) -- +(0,0.01); 
		  \draw [midarrow={stealth}] (1,1.63+.43) -- +(0,0.01); 
			%\draw [midarrow={stealth}] (5,2.37+.43) -- +(0,0.01); 
			\draw [midarrow={stealth}] (1,3.13+.43) -- +(0,0.01); 
			\draw [midarrow={stealth reversed}] (1,3.87+.43) -- +(0,0.01);
			\draw [midarrow={stealth reversed}] (1,4.63+.43) -- +(0,0.01); 
			%\draw [midarrow={stealth reversed}] (1,5.37+.43) -- +(0,0.01); 
			\draw [midarrow={stealth reversed}] (1,6.13+.43) -- +(0,0.01); 
			\draw [midarrow={stealth reversed}] (1,6.87+.43) -- +(0,0.01); 
			
		\draw[midarrow={stealth reversed}] (1,7.5+.13+.43)  -- +(0,0.01);	
			\end{tikzpicture}
			\quad \ $+$ \quad \
			\begin{tikzpicture}[baseline={([yshift=-.5*10pt*0.6+1pt]current bounding box.center)}, scale=0.6, font=\scriptsize]
	% first the horizontal border lines:
	\foreach \y in {1,3,5} {
		
		\draw (0.38, 1.5*\y-.25-.38) -- +(1.62, 0);
		\draw (0.38, 1.5*\y-.25+.38) -- +(1.62, 0);
		
		\foreach \x in {-1,...,1} \draw (2.4+.2*\x, 1.5*\y-.25-.38) node{$\cdot\mathstrut$};
		\foreach \x in {-1,...,1} \draw (2.4+.2*\x, 1.5*\y-.25+.38) node{$\cdot\mathstrut$};
		
		\draw (0.38,1.5*\y-.25+.38) arc (90:270:0.38);
	}
		
	\node[anchor=west] at (3, 1.5-.25-.38) {$-\lambda_i$};
	\node[anchor=west] at (3, 1.5*3-.25-.38) {$-\lambda_{k}$};
	\node[anchor=west] at (3, 1.5*5-.25-.38) {$-\lambda_j$};
		
	% then the vertical border lines:
	\foreach \x in {1} {
		
		\foreach \y in {-1,...,1} \draw (\x, 0-0.25+.2*\y) node{$\cdot\mathstrut$};
		\draw (\x,0-0.25+0.38) -- (1,3-0.25-.38); 
		\foreach \y in {-1,...,1} \draw (\x, 3-0.25+.2*\y) node{$\cdot\mathstrut$};
		\draw (\x,3-0.25+.38) -- (\x,6-0.25-.38); 
		\foreach \y in {-1,...,1} \draw (\x, 6-0.25+.2*\y) node{$\cdot\mathstrut$};
		\draw (\x,6-0.25+.38) -- (1,9-0.25-0.38); 
		\foreach \y in {-1,...,1} \draw (\x, 9-0.25+.2*\y) node{$\cdot\mathstrut$};
		\draw[->] (\x,9-0.25+0.38)  -- +(0,.2);

	}
	
	\node at (1, -0.25-.2-0.3) {${\mu_{m+1}}$};
	
		%left border
		%\draw[-stealth] (0,7.5+.05-0.25) -- (0,7.5+.06-0.25);
		\draw[midarrow={stealth}] (.55,6.5+.75-.38) -- +(0.01,0);
		\draw[midarrow={stealth}] (.55,6.5+.75+.38) -- +(0.01,0);
		
		%\draw[-stealth] (0,4.5+.05-0.25) -- (0,4.5+.06-0.25);
		\draw[midarrow={stealth}] (.55,3.5+.75-.38) -- +(0.01,0);
		\draw[midarrow={stealth reversed}] (.55,3.5+.75+.38) -- +(0.01,0);
		
		%\draw[-stealth] (0,3.05-0.25) -- (0,3.06-0.25);
		%\draw[midarrow={stealth reversed}] (.55,2.75-.38) -- +(0.01,0);
		%\draw[midarrow={stealth}] (.55,2.75+.38) -- +(0.01,0);
		
		%\draw[-stealth] (0,1.5+.05-0.25) -- (0,1.5+.06-0.25);
		\draw[midarrow={stealth}] (.55,0.5+.75-.38) -- +(0.01,0);
		\draw[midarrow={stealth}] (.55,0.5+.75+.38) -- +(0.01,0);
		
		%bulk horizontal lines
	
		\draw [midarrow={stealth}] (1.55, 7.5-.25-.38) -- +(0.01, 0);
		\draw [midarrow={stealth}] (1.55, 7.5-.25+.38) -- +(0.01, 0);
		
		\draw [midarrow={stealth}] (1.55, 4.5-.25-.38) -- +(0.01, 0);
		\draw [midarrow={stealth}] (1.55, 4.5-.25+.38) -- +(0.01, 0);
		
		%\draw [midarrow={stealth reversed}] (1.55, 3-.25-.38) -- +(0.01, 0);
		%\draw [midarrow={stealth}] (1.55, 3-.25+.38) -- +(0.01, 0);
		
		\draw [midarrow={stealth}] (1.55, 1.5-.25-.38) -- +(0.01, 0);
		\draw [midarrow={stealth}] (1.55, 1.5-.25+.38) -- +(0.01, 0);
		
	%bulk vertical lines
	
			%\draw [midarrow={stealth reversed}] (1,4.63+.43) -- +(0,0.01); 
			%\draw [midarrow={stealth reversed}] (1,5.37+.43) -- +(0,0.01); 
			%\draw [midarrow={stealth reversed}] (2,4.63+.43) -- +(0,0.01); 
			%\draw [midarrow={stealth reversed}] (2,5.37+.43) -- +(0,0.01); 

		\draw[midarrow={stealth}] (1,0.55) -- +(0,0.01);
			\draw [midarrow={stealth}] (1,0.87+.43) -- +(0,0.01); 
		  \draw [midarrow={stealth}] (1,1.63+.43) -- +(0,0.01); 
			%\draw [midarrow={stealth}] (5,2.37+.43) -- +(0,0.01); 
			\draw [midarrow={stealth}] (1,3.13+.43) -- +(0,0.01); 
			\draw [midarrow={stealth}] (1,3.87+.43) -- +(0,0.01);
			\draw [midarrow={stealth reversed}] (1,4.63+.43) -- +(0,0.01); 
			%\draw [midarrow={stealth reversed}] (1,5.37+.43) -- +(0,0.01); 
			\draw [midarrow={stealth reversed}] (1,6.13+.43) -- +(0,0.01); 
			\draw [midarrow={stealth reversed}] (1,6.87+.43) -- +(0,0.01); 
			
		\draw[midarrow={stealth reversed}] (1,7.5+.13+.43)  -- +(0,0.01);	
		\end{tikzpicture}
\caption{The $k_c$ turns of a state of size $2n\times m$ come from letting $\mu_{m+1}\to\infty$, i.e. removing the leftmost vertical line, in all states of size $2n\times (m+1)$ where all turns are the same as in the smaller lattice, except at some row $k$, where the $k_c$ turn in the smaller lattice corresponds to a $k_+$ or a $k_-$ turn at row $k$ in the bigger lattice. In the bigger lattice, all $k_c$ double rows below row $k$ are forced to be of the type in Figure~\ref{type3} and all $k_c$ double rows above are of the type in Figure~\ref{type4}.}
\label{figure:leftborder}
\end{figure}
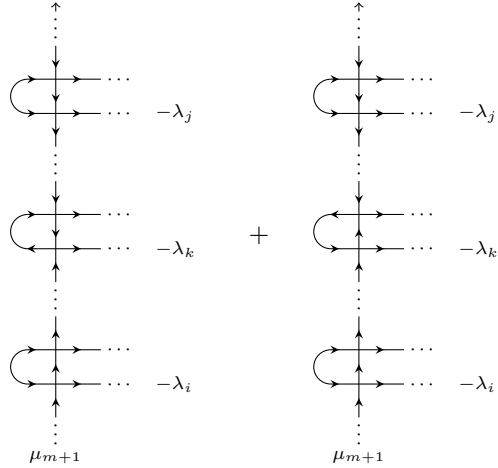

\begin{lemma}
Taking the limits $\mu_j\to\infty$, for $m+1\leq j\leq n$, in Tsuchiyas determinant formula $Z_n(\boldsymbol\lambda, \boldsymbol\mu)$, one after each other, starting with $\mu_n\to\infty$, we obtain
\begin{multline}
Z_{n,m}(\boldsymbol\lambda, \mu_1, \dots, \mu_m)\\
=\frac{\varphi^{n-m}}{(1-e^{-2(n-m)\gamma})\cdots(1-e^{-4\gamma})(1-e^{-2\gamma})}
\lim_{\mu_n, \mu_{n-1}, \dots, \mu_{m+1} \to\infty} Z_n(\boldsymbol\lambda, \boldsymbol\mu).
\label{recursionrelation}
\end{multline}
\end{lemma}

\begin{proof}
Consider a lattice of general size. In each state, the leftmost vertical line has an odd number of $c_\pm$ vertices, because of the DWBC and the ice rule. The $c_\pm$ vertices come as products $c_+(\lambda_i+\mu_m)c_+(\lambda_j-\mu_m)$ or $c_-(\lambda_i+\mu_m)c_-(\lambda_j-\mu_m)$, except for one single $c_+(\lambda_k+\mu_m)$ or $c_-(\lambda_k-\mu_m)$ vertex which is left alone. In the limit $\mu_m\to \infty$, the nonzero contribution to $Z_{n,m}(\boldsymbol\lambda, \mu_1, \dots, \mu_m)$ comes from the states with exactly one $c_\pm$ vertex in the $m$th column, so we only need to consider these states. In this case, each state contains exactly one of the configurations in Figure~\ref{type1} or Figure~\ref{type2}.

Each state has $n-m$ turns of type $k_c$. For $m<n$, each state of size $2n \times m$ comes from letting $\mu_{m+1}\to\infty$ in states of size $2n\times (m+1)$. In the smaller lattice all turns are the same as in the bigger lattice, except for at one double row $k$, where the $k_c$ turn corresponds to either a $k_+$ turn or a $k_-$ turn in a bigger lattice configuration, just as in the case $n=1$. 
Each $k_+$ or $k_-$ turn in the smaller lattice comes from one of the configurations in Figure~\ref{figure:zeroconfigs} in the bigger lattice. In the limit, each state in Figure~\ref{figure:zeroconfigs}, although consisting of two vertices and one turn, also have total weights $k_+(\lambda_i)$ or $k_-(\lambda_i)$ respectively. 
We are interested in how the partition function for the smaller lattice differs from the bigger lattice, and therefore we only need to consider the $k_c$ turns in the $2n\times m$ lattice, see Figure~\ref{figure:leftborder}. All turns of type $k_c$ below row $k$ come from turns of the type in Figure~\ref{type3}. All turns of type $k_c$ above row $k$ come from turns of the type in Figure~\ref{type4}. Each turn of type \ref{type3} has weight $e^{-2\gamma}k_c(\lambda_i)$ and each turn of type \ref{type4} has weight $k_c(\lambda_i)$ in the limit. 

The row $k$ can be chosen in $n-m$ different ways, so we need to sum over all these possibilities. Putting all of the above together, we conclude that
\begin{align}
\label{znm1}
Z_{n,m+1}(\boldsymbol\lambda, \mu_1, \dots, \mu_{m+1}) \to \frac{e^\gamma-e^{-\gamma}}{\varphi e^\gamma}\left(1+\frac{1}{e^{2\gamma}}+\cdots+\frac{1}{e^{2(n-m-1)\gamma}}\right)\nonumber\\
\times Z_{n,m}(\boldsymbol\lambda, \mu_1, \dots, \mu_m),
%\frac{1}{x}\left(1+\frac{1}{e^{2\gamma}}+\cdots+\frac{1}{e^{2(n-m-1)\gamma}-\frac{1}{e^{2\gamma}}-\cdots-\frac{1}{e^{2(n-m)\gamma}}\right)\\
%\cdot Z_{n,m}(\boldsymbol\lambda, \mu_1, \dots, \mu_m)
%\frac{1}{x}\left(1-\frac{1}{e^{2(n-m)\gamma}}\right)\\
%\cdot Z_{n,m}(\boldsymbol\lambda, \mu_1, \dots, \mu_m)
\end{align}
as $\mu_{m+1}\to\infty$. 
By iteration, the lemma follows. 
\end{proof}

We now prove the determinant formula by induction. The following proposition is the base step. 

\begin{prop}
\label{basestep}
For the 6V model of size $2n\times (n-1)$ with DWBC and one partially reflecting end, the partition function is
\begin{align*}
&Z_{n,n-1}(\boldsymbol\lambda, \mu_1, \dots, \mu_{n-1}) \nonumber\\
&=\varphi e^{\left(1-\binom{n+1}{2}\right)\gamma}\frac{\prod_{i=1}^{n-1} \left[e^{\mu_i+\zeta} f(\mu_i-\zeta)\right]\prod_{i=1}^n f(2\lambda_i)\prod_{i=1}^n\prod_{j=1}^{n-1} f(\mu_j\pm\lambda_i)}{\prod_{1\leq i<j\leq n-1}f(\mu_j\pm\mu_i)\prod_{1\leq i<j\leq n}\left[f(\lambda_i-\lambda_j)f(\lambda_i+\lambda_j+\gamma)\right]}\nonumber\\
&\phantom{=\ }\times\det 
\renewcommand{\arraystretch}{2}
\begin{pmatrix}
\dfrac{1}{f(\mu_1\pm\lambda_1)f(\mu_1\pm(\lambda_1+\gamma))} & \cdots & \dfrac{1}{f(\mu_1\pm\lambda_n)f(\mu_1\pm(\lambda_n+\gamma))}\\
\vdots & \ddots & \vdots\\
\dfrac{1}{f(\mu_{n-1}\pm\lambda_1)f(\mu_{n-1}\pm(\lambda_1+\gamma))} & \cdots & \dfrac{1}{f(\mu_{n-1}\pm\lambda_n)f(\mu_{n-1}\pm(\lambda_n+\gamma))}\\
1 & \cdots & 1
\end{pmatrix}.
%\label{P1}
\end{align*}
\end{prop}

\begin{proof}
We start from Tsuchiya's determinant formula \eqref{Tsuchiyadeterminantformula}. Using \eqref{recursionrelation} for $m=n-1$, we compute the limit 
$$Z_{n,n-1}(\boldsymbol\lambda, \mu_1, \dots, \mu_{n-1})=\frac{\varphi}{1-e^{-2\gamma}}\lim_{\mu_n\to\infty} Z_n(\boldsymbol\lambda, \mu_1, \dots, \mu_n).$$
Absorb everything that has to do with $\mu_n$ into the last row of the determinant. Then
\begin{align*}
&Z_{n,n-1}(\boldsymbol\lambda, \mu_1, \dots, \mu_{n-1})\\
&= \frac{\varphi e^{-\binom{n+1}{2}\gamma} f(\gamma)^n}{1-e^{-2\gamma}}\frac{\prod_{i=1}^{n-1} \left[e^{\mu_i+\zeta} f(\mu_i-\zeta)\right]\prod_{i=1}^n f(2\lambda_i)\prod_{i=1}^n\prod_{j=1}^{n-1} f(\mu_j\pm\lambda_i)}{\prod_{1\leq i<j\leq n-1}f(\mu_j\pm\mu_i)\prod_{1\leq i<j\leq n}\left[f(\lambda_i-\lambda_j)f(\lambda_i+\lambda_j+\gamma)\right]}\\
&\phantom{=\ }\times\lim_{\mu_n\to\infty}  \det 
\renewcommand{\arraystretch}{1.5}
\begin{pmatrix}
F(\lambda_1, \mu_1) & \cdots & F(\lambda_n, \mu_1)\\
\vdots & \ddots & \vdots\\
F(\lambda_1, \mu_{n-1}) & \cdots & F(\lambda_n, \mu_{n-1})\\
d_1^n & \cdots & d_n^n
\end{pmatrix}, 
\end{align*}
where
\begin{equation}
\label{Flambdamu}
F(\lambda_i, \mu_j)=\frac{1}{f(\mu_j\pm\lambda_i)f(\mu_j\pm(\lambda_i+\gamma))},
\end{equation}
and
\begin{equation}
\label{djk}
d_j^k=\frac{e^{\mu_k+\zeta} f(\mu_k-\zeta)\prod_{i=1}^n f(\mu_k\pm\lambda_i)}{\prod_{i=1}^{k-1}f(\mu_k\pm\mu_i)} F(\lambda_j, \mu_k).
\end{equation}
Write $v_i=2\lambda_i+\gamma$. Then 
\begin{equation*}
F(\lambda_i, \mu_j)=\frac{e^{-4\pi}}{(1-e^{-2\mu_j+v_i+\gamma})(1-e^{-2\mu_j+v_i-\gamma})(1-e^{-2\mu_j-v_i+\gamma})(1-e^{-2\mu_j-v_i-\gamma})}.
\end{equation*}
By Taylor expansion 
\begin{equation*}
F(\lambda_i, \mu_j)=\sum_{k=0}^\infty C_k(\lambda_i) e^{-2(k+2)\mu_j},
\end{equation*}
where
\begin{equation}
\label{Ck}
C_k(\lambda_i)=\quad\sum_{\mathclap{\substack{k_1, k_2, k_3, k_4\geq 0 \\ k_1+k_2+k_3+k_4=k}}} e^{(k_1+k_2-k_3-k_4)v_i+(k_1-k_2+k_3-k_4)\gamma}.
\end{equation}
Taking the limit yields
$
\lim_{\mu_n\to\infty} d_j^n %= \lim_{\mu_n\to\infty} \frac{e^{\mu_n+\zeta} f(\mu_n-\zeta)\prod_{i=1}^n f(\mu_n\pm\lambda_i)}{\prod_{i=1}^{n-1}f(\mu_n\pm\mu_i)} \sum_{l=0}^\infty C_l(\lambda_j) e^{-2\mu_n(l+2)}\\
%=\lim_{\mu_n\to\infty} \frac{(e^{2\mu_n}-e^{2\zeta})\prod_{i=1}^n \left[(e^{\mu_n+\lambda_i}-e^{-\mu_n-\lambda_i})(e^{\mu_n-\lambda_i}-e^{-\mu_n+\lambda_i})\right]}{\prod_{i=1}^{n-1}\left[(e^{\mu_n+\mu_i}-e^{-\mu_n-\mu_i})(e^{\mu_n-\mu_i}-e^{-\mu_n+\mu_i})\right]} \sum_{l=0}^\infty C_l(\lambda_j) e^{-2\mu_n(l+2)}\\
%=\lim_{\mu_n\to\infty} \frac{e^{4\mu_n}(1-e^{2\zeta-2\mu_n})\prod_{i=1}^n \left[(1-e^{-2\mu_n-2\lambda_i})(1-e^{-2\mu_n+2\lambda_i})\right]}{\prod_{i=1}^{n-1}\left[(1-e^{-2\mu_n-2\mu_i})(1-e^{-2\mu_n+2\mu_i})\right]} \sum_{l=0}^\infty C_l(\lambda_j) e^{-2\mu_n(l+2)}\\
%=\lim_{\mu_n\to\infty} \frac{(1-e^{2\zeta-2\mu_n})\prod_{i=1}^n \left[(1-e^{-2\mu_n-2\lambda_i})(1-e^{-2\mu_n+2\lambda_i})\right]}{\prod_{i=1}^{n-1}\left[(1-e^{-2\mu_n-2\mu_i})(1-e^{-2\mu_n+2\mu_i})\right]} \sum_{l=0}^\infty C_l(\lambda_j) e^{-2l \mu_n}\\
=C_0(\lambda_j)=1,
$
which proves the proposition.
\end{proof}

The next proposition yields the induction step. 

\begin{prop}
For the 6V model with DWBC and a partially reflecting end on a lattice of size $2n\times m$, $m\leq n$, the partition function is
\begin{align}
&Z_{n,m}(\boldsymbol\lambda, \mu_1, \dots, \mu_m)=\frac{\varphi^{n-m}e^{-\binom{n+1}{2}\gamma} f(\gamma)^n}{(1-e^{-2(n-m)\gamma})\cdots(1-e^{-4\gamma})(1-e^{-2\gamma})}\nonumber\\
&\phantom{=\ }\times
\frac{\prod_{i=1}^{m} \left[e^{\mu_i+\zeta} f(\mu_i-\zeta)\right]\prod_{i=1}^n f(2\lambda_i)\prod_{i=1}^n\prod_{j=1}^{m} f(\mu_j\pm\lambda_i)}{\prod_{1\leq i<j\leq m}f(\mu_j\pm\mu_i)\prod_{1\leq i<j\leq n}\left[f(\lambda_i-\lambda_j)f(\lambda_i+\lambda_j+\gamma)\right]}\nonumber\\
&\phantom{=\ }\times\det 
\renewcommand{\arraystretch}{1.5}
\begin{pmatrix}
\dfrac{1}{f(\mu_1\pm\lambda_1)f(\mu_1\pm(\lambda_1+\gamma))} & \cdots & \dfrac{1}{f(\mu_1\pm\lambda_n)f(\mu_1\pm(\lambda_n+\gamma))}\\[2ex]
\vdots & \ddots & \vdots\\[1.5ex]
\dfrac{1}{f(\mu_{m}\pm\lambda_1)f(\mu_{m}\pm(\lambda_1+\gamma))} & \cdots & \dfrac{1}{f(\mu_{m}\pm\lambda_n)f(\mu_{m}\pm(\lambda_n+\gamma))}\\[2ex]
C_{n-m-1}(\lambda_1) & \cdots & C_{n-m-1}(\lambda_n)\\
\vdots & \ddots & \vdots\\
C_0(\lambda_1) & \cdots & C_0(\lambda_n)
\end{pmatrix},
\label{Pnm}
\end{align}
where
$C_i(\lambda_j)$ are as in \eqref{Ck}. 
\end{prop}

\begin{proof}
Let $P_{n-m}$ be the claim that \eqref{Pnm} holds. We have already proven in Proposition~\ref{basestep} that $P_1$ is true. Assume $P_{n-m}$. We will prove that $P_{n-m+1}$ also holds.
From \eqref{znm1}, we have
\begin{equation*}
Z_{n,m-1}(\boldsymbol\lambda, \mu_1, \dots, \mu_{m-1})=\frac{\varphi}{1-e^{-2(n-m+1)\gamma}}\lim_{\mu_m\to\infty}Z_{n,m}(\boldsymbol\lambda, \mu_1, \dots, \mu_{m}) 
\end{equation*}
Insert \eqref{Pnm}. As in the proof of the previous proposition, absorb factors that have to do with $\mu_m$ into the $m$th row of the matrix,
\begin{align*}
&Z_{n,m-1}(\boldsymbol\lambda, \mu_1, \dots, \mu_{m-1})
%&=\left(\frac{x^{n-m+1}}{(1-e^{-2(n-m+1)\gamma})\cdots(1-e^{-4\gamma})(1-e^{-2\gamma})}\right)\lim_{\mu_n, \mu_{n-1}, \dots, \mu_{m} \to\infty} (Z_n(\boldsymbol\lambda, \boldsymbol\mu))\\
=
\frac{\varphi^{n-m+1}e^{-\binom{n+1}{2}\gamma} f(\gamma)^n}{(1-e^{-2(n-m+1)\gamma})\cdots(1-e^{-4\gamma})(1-e^{-2\gamma})}\nonumber\\
&\hphantom{=\ }\times
\frac{\prod_{i=1}^{m-1} \left[e^{\mu_i+\zeta} f(\mu_i-\zeta)\right]\prod_{i=1}^n f(2\lambda_i)\prod_{i=1}^n\prod_{j=1}^{m-1} f(\mu_j\pm\lambda_i)}{\prod_{1\leq i<j\leq m-1}f(\mu_j\pm\mu_i)\prod_{1\leq i<j\leq n}\left[f(\lambda_i-\lambda_j)f(\lambda_i+\lambda_j+\gamma)\right]}\\
&\hphantom{=\ }\times\lim_{\mu_m\to\infty}\det 
\renewcommand{\arraystretch}{1.5}
\begin{pmatrix}
F(\lambda_1, \mu_1) & \cdots & F(\lambda_n, \mu_1)\\
\vdots & \ddots & \vdots\\
F(\lambda_1, \mu_{m-1}) & \cdots &F(\lambda_n, \mu_{m-1})\\
d_1^{m}(\epsilon_m) & \cdots & d_n^{m}(\epsilon_m)\\
C_{n-m-1}(\lambda_1) & \cdots & C_{n-m-1}(\lambda_n)\\
\vdots & \ddots & \vdots\\
C_0(\lambda_1) & \cdots & C_0(\lambda_n)
\end{pmatrix},
\end{align*}
where $F(\lambda_j, \mu_k)$, $d_j^k$ and $C_i(\lambda_j)$ are as defined in \eqref{Flambdamu}, \eqref{djk} and \eqref{Ck} in the proof above. Consider
\begin{equation*}
d_j^m=\frac{e^{\mu_m+\zeta} f(\mu_m-\zeta)\prod_{i=1}^n f(\mu_m\pm\lambda_i)}{\prod_{i=1}^{m-1}f(\mu_m\pm\mu_i)} \sum_{l=0}^\infty C_l(\lambda_j) e^{-2(l+2)\mu_m}.
\end{equation*}
By row reduction in the determinant all terms with $C_l(\lambda_j)$, for $0\leq l \leq n-m-1$, can be removed from the sum. The elements of row $m$ become
\begin{align*}
\hat d_j^m=\frac{e^{\mu_m+\zeta} f(\mu_m-\zeta)\prod_{i=1}^n f(\mu_m\pm\lambda_i)}{\prod_{i=1}^{m-1}f(\mu_m\pm\mu_i)} \sum_{l=n-m}^\infty C_l(\lambda_j) e^{-2(l+2)\mu_m}.
\end{align*}
Now we take the limit,
\begin{multline*}
\lim_{\mu_m\to\infty}\hat d_j^m %= \lim_{\mu_m\to\infty} \frac{e^{\mu_m+\zeta} f(\mu_m-\zeta)\prod_{i=1}^n f(\mu_m\pm\lambda_i)}{\prod_{i=1}^{m-1}f(\mu_m\pm\mu_i)} \sum_{l=n-m}^\infty C_l(\lambda_j) e^{-2\mu_m(l+2)}\\
%=\lim_{\mu_m\to\infty} \frac{(e^{2\mu_m}-e^{2\zeta})\prod_{i=1}^n \left[(e^{\mu_m+\lambda_i}-e^{-\mu_m-\lambda_i})(e^{\mu_m-\lambda_i}-e^{-\mu_m+\lambda_i})\right]}{\prod_{i=1}^{m-1}\left[(e^{\mu_m+\mu_i}-e^{-\mu_m-\mu_i})(e^{\mu_m-\mu_i}-e^{-\mu_m+\mu_i})\right]} \sum_{l=n-m}^\infty C_l(\lambda_j) e^{-2\mu_m(l+2)}\\
=\lim_{\mu_m\to\infty} \frac{(1-e^{2\zeta-2\mu_m})\prod_{i=1}^n \left[(1-e^{-2\mu_m-2\lambda_i})(1-e^{-2\mu_m+2\lambda_i})\right]}{\prod_{i=1}^{m-1}\left[(1-e^{-2\mu_m-2\mu_i})(1-e^{-2\mu_m+2\mu_i})\right]} \\
\times e^{2(n-m+2)\mu_m}\sum_{l=n-m}^\infty C_l(\lambda_j) e^{-2(l+2)\mu_m}
%=\lim_{\mu_m\to\infty} \frac{(1-e^{2\zeta-2\mu_m})\prod_{i=1}^n \left[(1-e^{-2\mu_m-2\lambda_i})(1-e^{-2\mu_m+2\lambda_i})\right]}{\prod_{i=1}^{m-1}\left[(1-e^{-2\mu_m-2\mu_i})(1-e^{-2\mu_m+2\mu_i})\right]} \\
%\times e^{2(n-m)\mu_m}\sum_{l=n-m}^\infty C_l(\lambda_j) e^{-2l \mu_m}
=C_{n-m}(\lambda_j), 
\end{multline*}
so $P_{n-m+1}$ is true. 
\end{proof}

The last thing left to do is to simplify the last $n-m$ rows of the determinant. 

\begin{theorem}[Theorem~\ref{thm:determinantformula}]
\label{thm:partfcn}
For the 6V model with DWBC and a partially reflecting end on a lattice of size $2n\times m$, $m\leq n$, the partition function is
\begin{align*}
&Z_{n,m}(\boldsymbol\lambda, \mu_1, \dots, \mu_m)\nonumber\\
&=\varphi^{n-m}e^{\left(\binom{m}{2}-nm\right)\gamma}f(\gamma)^m
\frac{\prod_{i=1}^{m} \left[e^{\mu_i+\zeta} f(\mu_i-\zeta)\right]\prod_{i=1}^n f(2\lambda_i)\prod_{i=1}^n\prod_{j=1}^{m} f(\mu_j\pm\lambda_i)}{\prod_{1\leq i<j\leq m}f(\mu_j\pm\mu_i)\prod_{1\leq i<j\leq n}\left[f(\lambda_i-\lambda_j)f(\lambda_i+\lambda_j+\gamma)\right]}\nonumber\\
&\phantom{=\ }\times\det 
\renewcommand{\arraystretch}{1.5}
\begin{pmatrix}
\dfrac{1}{f(\mu_1\pm\lambda_1)f(\mu_1\pm(\lambda_1+\gamma))} & \cdots & \dfrac{1}{f(\mu_1\pm\lambda_n)f(\mu_1\pm(\lambda_n+\gamma))}\\[2ex]
\vdots & \ddots & \vdots\\[1ex]
\dfrac{1}{f(\mu_m\pm\lambda_1)f(\mu_m\pm(\lambda_1+\gamma))} & \cdots & \dfrac{1}{f(\mu_m\pm\lambda_n)f(\mu_m\pm(\lambda_n+\gamma))}\\[2ex]
h(2(n-m-1)(\lambda_1+\gamma/2)) & \cdots & h(2(n-m-1)(\lambda_n+\gamma/2))\\
\vdots & \ddots & \vdots\\
h(2(\lambda_1+\gamma/2)) & \cdots & h(2(\lambda_n+\gamma/2))\\
1 & \cdots & 1
\end{pmatrix},
\end{align*}
where $f(x)=2\sinh x$ and $h(x)=2\cosh x$. 
\end{theorem}

\begin{proof}
We start from the determinant \eqref{Pnm} from the last proposition. Focus on the entries of the lower part of the determinant, defined in \eqref{Ck} as
\begin{equation*}
C_k(\lambda_j)=\quad\sum_{\mathclap{\substack{k_1, k_2, k_3, k_4\geq 0 \\ k_1+k_2+k_3+k_4=k}}} e^{(k_1+k_2-k_3-k_4)v_j+(k_1-k_2+k_3-k_4)\gamma}.
\end{equation*}
These are clearly Laurent polynomials of degree $k$ in $e^{v_j}$, and they are even in $v_j$. We have already seen that $C_0(\lambda_j)=1$. For $k>0$, the leading coefficient is 
\begin{equation*}
\sum_{\mathclap{\substack{k_1, k_2\geq 0 \\ k_1+k_2=k}}} e^{(k_1+k_2)v_j+(k_1-k_2)\gamma}%=\sum_{i=0}^k e^{{k-2i}\gamma}
=\frac{f((k+1)\gamma)}{f(\gamma)}.
\end{equation*}
Thus by row reduction in the determinant, we can replace $C_k(\lambda_j)$ by $\frac{f((k+1)\gamma)}{f(\gamma)}(e^{k v_j}+e^{-k v_j})$ for $k>0$.
Switching back to the variables $\lambda_j$ and factoring out $\frac{f((k+1)\gamma)}{f(\gamma)}$ from each row of the lower part of the determinant yields the desired result. 
\end{proof}

\emergencystretch 1em
\sloppy
\printbibliography[heading=bibintoc]

%\addcontentsline{toc}{section}{References}
%\bibliographystyle{acm}
%\bibliography{References}

\end{document}